%% file: Main.tex
\RequirePackage{fix-cm}
\documentclass[a4paper]{article}

\usepackage[left=2.5cm, right=2.5cm, top=2.5cm, bottom=2.5cm]{geometry}
\usepackage[english]{babel}
\usepackage[ansinew]{inputenc}
\usepackage{amssymb,amsmath,amsthm}
\usepackage{graphicx}
\usepackage{url}
\usepackage{calc}
\usepackage{subcaption}
\usepackage{tabularx}
\usepackage[unicode, bookmarks=false]{hyperref}
\usepackage{paralist}
\usepackage{enumitem}
\usepackage{tabularx}
\usepackage{fullpage}
\usepackage{xspace}
\usepackage[boxed,commentsnumbered,linesnumbered]{algorithm2e}
\usepackage[numbers,sort&compress,sectionbib]{natbib}
\usepackage{tikz}
\usetikzlibrary{calc}
\usetikzlibrary{decorations.markings}
\usetikzlibrary{shapes.geometric,positioning}
\usetikzlibrary{arrows}

\usepackage{xspace}

\tikzset{vertex/.style={minimum size=1.8mm,circle,fill=black,draw,inner sep=0pt},
         decoration={markings,mark=at position .5 with {\arrow[black,thick]{stealth}}}}

\theoremstyle{plain}
\newtheorem{theorem}{Theorem}
\newtheorem{lemma}{Lemma}
\newtheorem{corollary}{Corollary}

\newtheorem{observation}{Observation}
\newtheorem{proposition}{Proposition}

\newcommand{\out}{\mbox{out}\xspace}

\newcommand{\Opt}{\ensuremath{\mathrm{Opt}}\xspace}
\newcommand{\cost}{\ensuremath{\mathrm{cost}}\xspace}
\newcommand{\SER}{\ensuremath{\mathrm{SER}}\xspace}
\newcommand{\SERup}{\ensuremath{\mathrm{SER}^+}\xspace}
\newcommand{\OptSER}{\ensuremath{\mathrm{Opt}_{\mathrm{SER}}}\xspace}
\newcommand{\OptSERup}{\ensuremath{\mathrm{Opt}_{\mathrm{SER}^+}}\xspace}

\newcommand{\startS}{\ensuremath{S_{\mathrm{start}}}}

\newcommand{\pendS}{\ensuremath{S_{\mathrm{pseudo-end}}}}

\newcommand{\startSp}{\ensuremath{S'_{\mathrm{start}}}}

\newcommand{\pendSp}{\ensuremath{S'_{\mathrm{pseudo-end}}}}

\newcommand{\apset}{\ensuremath{\mathcal{Q}}\xspace}
\newcommand{\degree}[1]{\ensuremath{\textrm{deg}_G(#1)}\xspace}
\newcommand{\indeg}[1]{\ensuremath{\textrm{deg}_G^-(#1)}\xspace}
\newcommand{\outdeg}[1]{\ensuremath{\textrm{deg}_G^+(#1)}\xspace}

\newcommand{\ie}{i.\,e.}

\title{Improved Integrality Gap Upper Bounds for TSP\\
        with Distances One and Two\thanks{
    This research was partially supported by 
    DFG grant BL511/10-1 and 
    ERC Starting Grant 306465 (BeyondWorstCase).
    }    
}

\author{Matthias Mnich\thanks{    
    Universit{\"a}t Bonn, Institut f{\"u}r Informatik, Friedrich-Ebert-Allee 144, 53113 Bonn, Germany,
    \texttt{mmnich@uni-bonn.de}    
    }     
    \and Tobias M{\"o}mke\thanks{
    Universit{\"a}t des Saarlandes, Campus E1-3, 66123 Saarbr{\"u}cken, Germany,
    \texttt{moemke@cs.uni-saarland.de} 
    }
}
\date{}

\newcommand{\STSP}{\ensuremath{(1,2)\mbox{-STSP}}\xspace}
\newcommand{\ATSP}{\ensuremath{(1,2)\mbox{-ATSP}}\xspace}
\newcommand{\NP}{\ensuremath{\mathsf{NP}}\xspace}

\begin{document}
\maketitle

\begin{abstract}
  We study the structure of solutions to linear programming formulations for the traveling salesperson problem (TSP).
  
  \smallskip
  We perform a detailed analysis of the support of the subtour elimination linear programming relaxation, which leads to algorithms that find $2$-matchings with few components in polynomial time.
  The number of components directly leads to integrality gap upper bounds for the TSP with distances one and two, for both undirected and directed graphs.
  
  \smallskip
  \noindent
  Our main results concern the subtour elimination relaxation with one additional cutting plane inequality:
  \begin{itemize}
    \item For undirected instances we obtain an integrality gap upper bound of
    $5/4$ without any further restrictions, of $7/6$ if the optimal LP solution
    is half-integral.
    \item For instances of order $n$ where the fractional LP value has a cost of $n$, we obtain a tight integrality gap upper bound of $10/9$ if there is an optimal solution with subcubic support graph. 
    The latter property that the graph is subcubic is implied if the solution is a basic solution in the fractional $2$-matching polytope.
    \item For directed instances we obtain an integrality gap upper bound of $3/2$, and of $4/3$ if given an optimal $1/2$-integral solution.
  \end{itemize}
  In the case of undirected graphs, we can avoid to add the cutting plane inequality if we accept slightly increased values.
  For the tight result, the cutting plane is not required.
  \smallskip

  Additionally, we show that relying on the structure of the support is not an artefact of our algorithm, but is necessary under standard complexity-theoretic assumptions: we show that finding improved solutions via local search is $\mathsf{W}[1]$-hard for $k$-edge change neighborhoods even for the TSP with distances one and two, which strengthens a result of D{\'a}niel Marx.
\end{abstract}

\input{Intro.tex}

\input{Prelim.tex}

\input{Singleton.tex}

\input{12TSP.tex}

\end{document}

%% file: Intro.tex

\section{Introduction}
\label{sec:introduction}
The traveling salesperson problem (TSP) in metric graphs is one of the most fundamental $\mathsf{NP}$-hard optimization problems.
Given an undirected or directed graph $G$ with a metric on its edges, we seek a tour~$\mathcal T$ (a Hamiltonian cycle) of minimum cost in $G$, where the cost of $\mathcal T$ is the sum of costs of edges traversed by~$\mathcal T$.

Despite a vast body of research, the best approximation algorithm for metric TSP is still Christofides' algorithm~\cite{Christofides1976} from 1976, which has a performance guarantee of~$3/2$. 
Recall that the performance guarantee or approximation ratio of an algorithm for a problem is defined as a number $\alpha$ such that, in polynomial time, the algorithm computes a solution whose value is within a factor $\alpha$ of the optimal value.
Generally the bound $3/2$ is not believed to be tight.
However, the currently largest known lower bound on the performance guarantee obtainable in polynomial time is as low as $123/122$~\cite{KarpinskiEtAl2013}.

One of the most promising techniques to obtain an improved performance guarantee is to use a linear programming (LP) formulation of TSP.
Upper bounds on the integrality gap of the LP usually translate to approximation guarantees.
In this context, the \emph{subtour elimination relaxation} (\SER), or Held-Karp relaxation~\cite{HeldKarp1970}, is particularly important.
Its integrality gap is between $4/3$ and $3/2,$ and the value $4/3$ is conjectured to be tight~\cite{Goemans1995}.
For relevant special cases, the conjecture is known to be true~\cite{BoydEtAl2011,MomkeSvensson2011}.
It is also known that \SER has a close relation to $2$-matchings, as was pointed out for instance by Schalekamp et al.~\cite{SchalekampEtAl2014} in the context of perfect $2$-matchings.

\subsection{Our Contributions}
\label{sec:ourcontributions}
We investigate the structure of the \emph{support graph} of solutions to \SER, that is, the graph of edges with non-zero value in an optimal solution to \SER.
We show how to find a $2$-matching (\ie, a collection of paths and cycles) that approximates the minimum the number of components in the \SER support graph.
While we consider our structural findings to be valuable by themselves, they have a direct impact on the integrality gap of \SER for several TSP variants.
In particular, we obtain improved integrality gap upper bounds for the asymmetric and symmetric TSP with distances one and two, a classical and well-studied variant of TSP~\cite{Williamson1990,Vishwanathan1992,PapadimitriouYannakakis1993,BlaserRam2005,BlaserManthey2005,BermanKarpinski2006,KarpinskiSchmied2012,QSWvZ15}.
We refer to the symmetric variant (in undirected graphs) as \STSP and to the asymmetric variant as \ATSP.
\medskip

First, we augment \SER by a single cutting hyperplane to a linear program \SERup; this way, we are enforcing that an optimal solution to \SERup takes an integer value.
The modification requires an integer cost function, but is not specific to edge costs one and two.
We then consider the support graph of an optimal \SERup solution.
We show how to modify this graph in such a way that it allows us to prove integrality gap upper bounds for TSP by computing a $2$-matching with few components.
To this end, we define certain types of improvements that extend the improvements used by Berman and Karpinski~\cite{BermanKarpinski2006} in their approximation algorithm for \STSP.
For instance, we show how to transform a $2$-matching in the support to another $2$-matching not containing isolated vertices without increasing the number of components, in Sect.~\ref{sec:removingsingletoncomponents}.
With further improvements obtained by applying alternating paths, in Sect.~\ref{sec:integralitygapupperboundsfor12stsp} we find in polynomial time a $2$-matching with at most~$n/4$ components.
This $2$-matching then implies an integrality gap upper bound of $5/4$ for arbitrary \STSP instances.
\medskip

Second, we consider half-integral instances of \STSP.
A conjecture of Schalekamp et al.~\cite{SchalekampEtAl2014} implies that instances exist for which the integrality gap of \SER and \SERup is tight and which at the same time are basic solutions to the fractional perfect $2$-matching polytope.
These basic solutions are well understood and they have a quite specific structure~\cite{Balinski1965}.
In particular, they are half-integral (all LP values are multiples of~$1/2$) and they are subcubic (all degrees in the support are at most three). 
We show that if the half-integrality part of the conjecture is true, then integrality gap of \SERup is at most~$7/6$.
\medskip

Third, we consider the conjecture of Schalekamp et al.~\cite{SchalekampEtAl2014} without requiring half-integrality, in Sect.~\ref{sec:subcubicfractionallyhamiltoniansupports}.
For \STSP on instances $G$ that admit optimal basic solutions with subcubic support of $\SER(G)$ and $\OptSER(G) = |V(G)|$, we obtain a \emph{tight} integrality gap and an improved approximation guarantee of~$10/9$.
We think that the restriction to instances where $\OptSER(G) = |V(G)|$ is quite benign since any instance satisfies $\OptSER(G) \ge |V(G)|$, and usually the integrality gaps of linear programs decrease with increasing LP value.
\medskip

Fourth, we transfer our results from \STSP to \ATSP in a natural way.
We prove an \SERup integrality gap upper bound of $3/2$ for general instances of \ATSP, and~$4/3$ if additionally there is a half-integral optimal solution. 
\medskip

Fifth, we show that the \SERup integrality gap upper bound converges to the \SER integrality gap upper bound with increasing instance size.
By an amplification technique, we use known computational results for small instances to show $\SER$ integrality gap upper bounds that differ by less than two percent from our $\SERup$ integrality gap upper bounds.
Our results for \SER provide the currently best integrality gap upper bound for instances with given optimal half-integral solution.
\medskip

Sixth, we identify structures in \STSP and \ATSP instances that allow us to increase the instance size by an arbitrary factor without decreasing the integrality gap of the instance.
One consequence of these results is that there are infinitely many \ATSP instances with integrality gap at least $6/5$.
The gained insights give raise to conjecture that the integrality gaps of \SER and \SERup coincide for both \STSP and \ATSP.
\medskip

Finally, we strengthen a result of Marx~\cite{Marx2008} about finding cheaper solutions to a \STSP instance by local search, in Sect.~\ref{sec:intractabilityoflocalsearch}.
Precisely, we show that finding a cheaper solution compared to a given tour by exchanging at most~$k$ edges is $\mathsf{W}[1]$-hard, even in undirected TSP instances with distances one and two.
This hardness result intuitively says that a brute force search of all subsets of $k$ edges---in time~ $n^{O(k)}$ for \STSP instances of order $n$---is essentially optimal, unless many canonical $\mathsf{NP}$-complete problems admit subexponential-time algorithms.
Such an intractability result was known before only for TSP instances with \emph{three} distinct city distances, due to Marx~\cite{Marx2008}.
This result suggests that a simple search for local improvements is not efficient.
Our proof is similar to that of Marx, with some simplifications.

\subsection{Overview of Techniques}
To show that the subtour elimination support graph contains a $2$-matching with few components, we apply a sequence of local improvements.
One important step is that we can exclude the solution computed by our algorithm to contain isolated vertices.
We use an induction that creates a tree of alternating paths, and show that it is always possible to increase the size of the tree unless there is an improved $2$-matching (i.\,e., with fewer components).

All five of our integrality gap upper bound results use an accounting technique that distributes an amount of $n$ coins to components of the $2$-matching.
We assign a sufficient amount of coins to each component of the $2$-matching to ensure that the total number of components cannot exceed the aimed-for upper bound.
However, we employ two entirely different schemes in order to provide a distribution of coins.
Our result for degree-$3$-bounded support graphs initially assigns one coin to each vertex, and then redistributes the coins to the components.
A similar accounting technique has been used by Berman and Karpinski~\cite{BermanKarpinski2006}.
However, we exploit properties of of the subtour elimination constraints in order to ensure the existence of fractional coins that are available provided that no local improvements are possible.

The remaining integrality gap upper-bound results use the LP values of the subtour elimination relaxation directly. 
One way to see the technique is to initially distribute $n$ coins to \emph{edges}, where each edge obtains a fraction of the coin according to its LP value.
As key idea, we formulate the distribution of values to components via a new linear program. 
The linear program takes the \SER solution~$x^*$ and the aimed-for number of coins per component as parameter and this way we reduce the analysis to finding a feasible solution to the linear program.
To find a feasible solution to the new linear program, we split the edges of the \SER support graph into sub-edges such that each sub-edge $e$ has the same LP value~$x^*_e$. 
In the resulting multigraph, we obtain a collection of disjoint alternating paths of which certain types then lead to improved $2$-matchings.

\subsection{Related Work}
\label{sec:relatedwork}
Both \STSP and \ATSP are well-studied from the approximation point of view.
For the problem \STSP, it is $\mathsf{NP}$-hard to obtain a performance guarantee better than $535/534$~\cite{KarpinskiSchmied2012}.
Papadimitriou and Yannakakis~\cite{PapadimitriouYannakakis1993} gave a $7/6$-approximation algorithm for \STSP.
The approximation factor was improved by Bl{\"a}ser and Ram~\cite{BlaserRam2005} to $65/56$, and then to $8/7$ by Berman and Karpinski~\cite{BermanKarpinski2006}.

The best known integrality gap lower bound of \STSP is $10/9$, due to Williamson~\cite{Williamson1990}.
Previous to our result, Qian et al.~\cite{QianEtAl2012} showed an integrality gap upper bound of $19/15$ 
for \STSP, and of $7/6$ if the integrality gap is attained by a basic solution of the fractional $2$-matching polytope.
After we obtained our results \cite{MM13}, Qian et al.~improved the integrality gap upper bound to $5/4$ and to $26/21$ for fractionally Hamiltonian instances with techniques distinct from ours~\cite{QSWvZ15}. 
With the additional assumption that a certain type of modification maintains the $2$-vertex connectedness of the support graph, they were able to show a tight integrality gap of~$10/9$.

For \ATSP, it is \NP-hard to obtain a performance ratio better than 207/206~\cite{KarpinskiSchmied2012}.
The first non-trivial approximation algorithm for \ATSP was given by Vishwanathan~\cite{Vishwanathan1992}, with an approximation factor of $17/12$.
This was improved to $4/3$ by Bl{\"a}ser and Manthey~\cite{BlaserManthey2005}.
The currently best approximation factor is $5/4$, and is due to Bl{\"a}ser~\cite{Blaser2004}.

%% file: Prelim.tex

\section{Subtour Elimination Linear Programming Relaxation}
\label{sec:preparatorywork}
\subsection{Notation}
Let us start with some notation that will be used throughout the paper.
For an undirected (directed) graph $G$, let~$V(G)$ denote its set of vertices and $E(G)$ its set of edges (arcs).

Let $G$ be an undirected graph.
For vertex sets $S,S' \subseteq V(G)$, let $\delta_G(S,S') = \{\{v,w\}~|~v \in S, w \in S'\}$ be the set of edges between $S$ and $S'$.
We use $\delta_G(S) = \delta(S,V(G) \setminus S)$.
For a single vertex $v$, we write~$\delta_G(v)$ instead of $\delta_G(\{v\})$.
We define $N_G(v) = \{ u \in V(G)~|~\{u,v\} \in \delta_G(v)\}$ to be the neighborhood of $v$.
For each vertex $v\in V(G)$, let $\degree{v}=|N_G(v)|$ be its \emph{degree} in $G$.

Let $G$ be a directed graph with possibly bidirected arcs but without loops.
For each vertex $v\in V(G)$, let $\indeg{v}=|N^-_G(v)|$ and $\outdeg{v}=|N^+_G(v)|$ be its \emph{in-degree} and \emph{out-degree}, and let $\degree{v} = \indeg{v} + \outdeg{v}$ be its degree.
In directed graphs, we sometimes need to distinguish between arcs leaving a set of vertices and those entering the set, which we mark by superscripts $+$ or~$-$.
For vertex sets $S,S'\subseteq V(G)$, define $\delta_G(S,S') = \{(v,w),(w,v) \in E(G)~|~v \in S, w \in S'\}$, $\delta_G^+(S,S') = \{(v,w) \in E(G)~|~v \in S, w \in S'\}$, and $\delta_G^-(S,S') = \{(w,v) \in E(G) : v \in S, w \in S'\}$.
Analogous to the undirected case, we define $\delta_G^+(S) = \delta_G^+(S,V(G) \setminus S)$ and $\delta_G^-(S) := \delta_G^-(S,V(G) \setminus S)$.
We define $N_G(v) = \{ u \in V(G)~|~(u,v) \in \delta_G^-(v)\mbox{ or } (v,u) \in \delta_G^+(v)\}$,  $N_G^+(v) = \{ u \in V(G)~|~(v,u) \in \delta_G^+(v)\}$, and $N_G^-(v) = \{ u \in V(G)~|~(u,v) \in \delta_G^-(v)\}$.

\subsection{Subtour Elimination Linear Program}
For an instance $G$ of \STSP, we define a linear programming relaxation of an integer linear program that models an optimal TSP tour in $G$.
Given an undirected graph $G$, we introduce one variable $x_e$ for each edge $e \in E(G)$.
Variable $x_e$ models whether edge $e$ belongs to the TSP tour ($x_e = 1$) or not~($x_e = 0$).
For a set of edges $E' \subseteq E(G)$, we write $x(E') = \sum_{e \in E'} x_e$.
We formulate the \emph{subtour elimination linear programming relaxation} for $G$ as follows:
\begin{align}
                        \min \quad \sum_{e \in E(G)} \cost(e) x_e& \qquad \notag \\
\textnormal{subject to} \quad \sum_{e \in \delta_G(v)} x_e   &= 2 \qquad\mbox{ for } v \in V(G), \label{eqn:ser_general_1}\\
                        \sum_{e \in \delta_G(S)} x_e & \ge 2 \qquad\mbox{ for } \emptyset \neq S \subset V(G), \label{eqn:ser_general_2}\\
                                                 x_e & \ge 0 \qquad\mbox{ for } e \in E(G) \notag
\end{align}
where $\cost(e)$ is the cost of edge $e$.

We refer to the first set \eqref{eqn:ser_general_1} of constraints as the \emph{equality constraints}, and to the second set \eqref{eqn:ser_general_2} as the \emph{subtour elimination constraints}.
The equality constraints model that each node of a tour must be connected to other nodes by exactly two edges, and the subtour elimination constraints model that any non-empty proper subset of nodes must be connected by at least two edges with the remaining set of nodes.
If $G$ is a directed graph, we replace the constraints by \eqref{eqn:ser_general_1}--\eqref{eqn:ser_general_2} by
\begin{align}
  x(\delta_G^-(v)) = 1,  &\qquad & x(\delta_G^+(v)) = 1, v \in V(G), \label{eq:asubtour1}\\
  x(\delta_G^-(S)) \ge 1 &\qquad & x(\delta_G^+(S)) \ge 1, \emptyset \neq S \subset V(G), \label{eq:asubtour}
\end{align}
and $x_e \ge 0$ for all $e \in E(G)$.
The constraints on the right hand side of \eqref{eq:asubtour} are redundant, but we keep them for convenience.
We will refer to the polyhedra of both linear programs as $\SER(G)$, where the constraints depend on whether $G$ is directed or not.
The \emph{cost} of a solution $x$ to $\SER(G)$ is defined as $\cost(x) = \sum_{e \in E(G)}\cost(e)x_e$.
We refer to the \emph{value} of an optimal solution $x^*$ to $\SER(G)$ by $\OptSER(G)=\cost(x^*)$, and write $\Opt(G)$ for the cost of an optimal integral solution of $\SER(G)$.
We write \SER as shorthand for ``subtour elimination linear programming relaxation,'' for both the directed and undirected version.

We show our main results for a slightly modified version of \SER using the standard
technique of introducing a cutting hyperplane that rounds up the value of the objective function to the next integer.
Formally, for a given instance $G$ we first compute an optimal $\SER(G)$ solution and we obtain \SERup from \SER by adding the constraint
\begin{equation}
\label{eq:roundup}
\sum_{e \in E} \cost(e) x_e  \ge \lceil \OptSER(G) \rceil \enspace .
\end{equation}
The relaxation \SERup is valid for all TSP instances with integer cost functions, which includes graph-TSP and, by scaling, any TSP instance with rational edge costs.

We observe that \eqref{eq:roundup} is indeed a cutting hyperplane because all points removed from the polytope have a
fractional objective function value for any integer objective function. This
implies that none of the removed points can be located within the convex hull of
the integer points.
We refer to the polytope of the linear program for an instance $G$ by $\SERup(G)$. 

For notational convenience, we give the following definitions only for $\SER$ instead of considering generic linear programs; they extend directly to \SERup.
The \emph{integrality gap} of $\SER(G)$ is defined as $\Opt(G)/\OptSER(G)$, and the \emph{integrality gap} of $\SER$ is the supremum over the integrality gaps of all instances $G$.
The \emph{support graph} of a solution $x \in \SER(G)$ is the graph $G_x$ with $V(G_x) = V(G)$ and $E(G_x) = \{e \in E(G) : x_e > 0\}$.
An edge $e \in E(G_x)$ is called a \emph{$1$-edge}, if $x_e = 1$.
 
In the following proposition, we changed the formulation of the statement to fit our needs.
\begin{proposition}[Qian et al.~\cite{QSWvZ15}]
\label{pro:aux}
  For any instance $G$ of \STSP there is an instance $G'$ and an $\varepsilon \in [0,1)$ such that the \SER integrality gaps of $G$ and $G'$ are identical and $\OptSER(G') \le |V(G')| + \varepsilon$.
\end{proposition}
Thus loosely speaking they observed that there is a $\varepsilon \in [0,1)$ such that to show integrality gap upper bounds of \SER for \STSP, it suffices to consider graphs $G$ of order $n$ with $\OptSER(G) = n+\varepsilon$.
The key insight used in the proof is that, if $\OptSER(H) - |V(H)| \ge k$ for some integer $k \ge 1$, the LP values of all cost two edges sum up to at least $k$.
This allows to transform~$H$ into a graph $H'$ with $k$ auxiliary vertices such that each of them has a cost one edge to each vertex of $G$ that is incident to a cost two edge of the support graph.
We note that this transformation does not destroy half-integrality and that the reasoning directly translates to \ATSP.

The main purpose of using \SERup instead of \SER is the following insight.
\begin{lemma}
\label{lem:pluseps}
  For any instance $G$ of \STSP or \ATSP there is an instance $G'$ such that the \SERup integrality gaps of $G$ and $G'$ are identical and $\OptSER(G') = \OptSERup(G') = |V(G')|$.
  In particular, any optimal solutions $x^*$ to $\SERup(G')$ has the property that all edges in the support graph $G_{x^*}$ have a cost of one.
\end{lemma}
\begin{proof}
Since any optimal solution $x$ of $\SERup(H)$ is a feasible solution of $\SER(H)$, we can apply Proposition~\ref{pro:aux} based on $x$ in order to obtain the aimed-for instance $G$.
By the equality constraints, $\sum_{e \in E} x_e = n$ and therefore 
\[
\sum_{e \in E\colon \cost(e) = 2} x_e = \sum_{e \in E} \cost(e) x_e - \sum_{e \in E} x_e = \OptSERup(H) - n
\]
is an integer.
We conclude that $\lfloor \OptSER(G) \rfloor = \OptSER(G) = \OptSERup(G)$. 

Now the last claim follows easily from
  \begin{equation*}
    \sum_{e \in E(G)} x^*_e = n = \sum_{e \in E(G)} \cost(e) x^*_e \enspace .
  \end{equation*}
\end{proof}

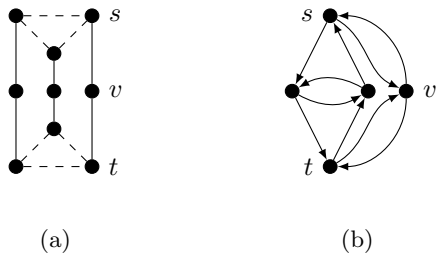
\begin{figure}[htb]
  \centering
  \begin{subfigure}[b]{0.22\textwidth}
    \centering
    \begin{tikzpicture}[scale=0.5]
      \node (1) at (0,0) [vertex, label=below:\phantom{M}]{};
      \node (2) at (2,0) [vertex, label=right:$t$]{};
      \node (3) at (1,1) [vertex]{};
      \node (4) at (0,2) [vertex]{};
      \node (5) at (1,2) [vertex]{};
      \node (6) at (2,2) [vertex, label=right:$v$]{};
      \node (7) at (1,3) [vertex]{};
      \node (8) at (0,4) [vertex]{};
      \node (9) at (2,4) [vertex, label=right:$s$]{};
      \draw[dashed](1)--(2);
      \draw[dashed](1)--(3);
      \draw[](1)--(4);
      \draw[dashed](2)--(3);
      \draw[](2)--(6);
      \draw[](3)--(5);
      \draw[](4)--(8);
      \draw[](5)--(7);
      \draw[](6)--(9);
      \draw[dashed](7)--(8);
      \draw[dashed](7)--(9);
      \draw[dashed](8)--(9);
    \end{tikzpicture}
    \caption{$~$}
    \label{fig:simple_gap_undirected}
  \end{subfigure}
  \quad
  \begin{subfigure}[b]{0.22\textwidth}
    \centering
    \begin{tikzpicture}[->,scale=0.5]
      \node (1) at (0,2) [vertex]{};
      \node (2) at (1,0) [vertex, label=below:\phantom{M}, label=left:$t$]{};
      \node (3) at (1,4) [vertex, label=left:$s$]{};
      \node (4) at (2,2) [vertex]{};
      \node (5) at (3,2) [vertex, label=right:$v$]{};
      \draw[>=latex] (1) -> (2);
      \draw[>=latex] (2) -> (4);
      \draw[>=latex] (4) -> (3);
      \draw[>=latex] (3) -> (1);
      \draw[>=latex] (1) to[bend right] (4);
      \draw[>=latex] (4) to[bend right] (1);
      \draw[>=latex] (2) to[out=30, in=205] (5);
      \draw[>=latex] (5) to[out=270, in = 0] (2);
      \draw[>=latex] (3) to[out=330, in=155] (5);
      \draw[>=latex] (5) to[out=90, in=0] (3);
    \end{tikzpicture}
    \caption{$~$}
    \label{fig:simple_gap_directed}
  \end{subfigure}
  \caption{
    \label{fig:simple_intgap}
    Integrality gap lower bound instances for \STSP and \ATSP, where edges/arcs are drawn if and only if they are of cost one. 
    For the undirected graphs, there is a solution $x$ to \SER such that $x_e=1/2$ for all dashed edges $e$ and $x_e=1$ otherwise. 
    For the directed graphs, $x_e=1/2$ for all depicted arcs $e$.}
\end{figure}
For \SER, Fig.~\ref{fig:simple_gap_undirected} shows a well-known \STSP instance with integrality gap~$10/9$. 
Since the fractional optimal solution has integer cost, the same instance also shows an integrality gap lower bound of $10/9$ for \SERup.

For \ATSP, we claim that the instance depicted in Fig.~\ref{fig:simple_gap_directed} provides a lower bound of~$6/5$ on the integrality gap for both \SER and \SERup.
Despite its simplicity, we are not aware of a previous appearance in the literature.
\begin{theorem}
\label{lem:simple_intgap}
The integrality gap of \SER and \SERup for \ATSP is at least $6/5$.
\end{theorem}
\begin{proof}
Let us consider the instance of Fig.~\ref{fig:simple_gap_directed}.
Clearly, by checking all cuts we can observe that assigning $x_e=1/2$ to each arc yields a feasible solution to \SER and to \SERup of cost $5$.

In any optimal solution, there is at least one arc of cost two.
Otherwise, due to the symmetry, we can assume without loss of generality that an integral solution contains $(v,t)$ and there is only one arc of cost one available from $t$.
Either of the two possible next steps forces any integral solution to use an arc of cost $2$, a contradiction.
Therefore, an optimal integral solution has cost $6$, resulting in an integrality gap of $6/5$.
\end{proof}

\subsubsection{2-Matchings in the Support Graph}
For the argumentation within our proofs, we change the point of view to ``$2$-matchings''.
We define a \emph{$2$-matching} $M$ of an undirected graph $G$ as a subgraph of $G$ such that $\degree{v} \le 2$ for all $v \in V(G)$.
This allows us to talk about components of a 2-matching.
Note that in the literature, the term $2$-matching is sometimes used for perfect $2$-matchings, where all degrees are exactly two~\cite{QSWvZ15,SchalekampEtAl2014}.

For a directed graph $G$, we define a \emph{directed $2$-matching} $M$ to be a subgraph of $G$ such that the in-degree $\indeg{v} \le 1$ and the out-degree $\outdeg{v} \le 1$ for each vertex $v \in V(G)$.
In other words, we use the simpler term ``directed $2$-matching'' for a degree-two bounded $1$-transshipment.

Let $G$ be a graph of order $n$ that forms an instance of \STSP or \ATSP. 
Let $k$ be the minimum number such that there is a (directed) $2$-matching $M$ in $G$ with $k$ components.
Then it is not hard to see that $\Opt(G) = n+k$ if $k \ge 2$.
This observation motivates to focus only on the number of components of $M$.
In particular, we do not consider edges/arcs of cost two when showing our integrality gap upper bounds.
Given the (directed) 2-matching $M$, we now ``improve'' this to a 2-matching $M'$.
To this end, we define a \emph{singleton component} of a $2$-matching as a single vertex without incident edges/arcs.
Then similar to Berman and Karpinski~\cite{BermanKarpinski2006}, an \emph{improvement} of a 2-matching $M$ is a transformation to a $2$-matching $M'$ such that one of the following conditions is satisfied: 
\begin{enumerate}
  \item $M'$ has fewer components than $M$.\label{imp1}
  \item $M'$ has the same number of components as $M$ and more cycles.\label{imp2}
  \item $M'$ has the same number of components and the same number of of cycles as $M$, but more edges in cycles.\label{imp2prime}
  \item $M'$ has the same number of components and cycles as $M$ and the same
number of edges in cycles but fewer singleton components.\label{imp3}
  \item $M'$ has the same number of components and cycles as $M$, the same number of edges in cycles, and the same number of singleton components but fewer
components of size two.\label{imp4}
\end{enumerate}
Any improvement can only be applied linearly often and only revert improvements with higher indices, so there are at most $n^{O(1)}$ improvements in total.

If $G$ is directed, let $M$ be a directed 2-matching of $G$.
A vertex $v$ is a \emph{start vertex} of $M$ if $\delta_M^-(v) = \emptyset$ or the component containing $v$ is a cycle.
A vertex $v$ is an \emph{end vertex} of $M$ if $\delta_M^+(v) = \emptyset$ or the component containing $v$ is a cycle.
The reason to define vertices in cycles to be start/end vertices is that they become path starts/ends by removing an incident arc.

If there is an arc from an end vertex to a start vertex within $M$, simply adding the arc to~$M$ and possibly removing arcs from cycles leads to an improvement, unless both vertices are in one
cycle (that is, we obtain a new 2-matching with fewer components or with more cycles).
We call such improvements \emph{basic}.

For undirected graphs, we do not distinguish between start and end vertices.
Let $M$ be a 2-matching of $G$.
A vertex $v$ is an \emph{end vertex} of $M$ if $v$ has at most one adjacent vertex within its component or the component containing $v$ is a cycle.
Basic improvements are analogous to those of directed graphs, but we simply require an edge between two end vertices instead of an arc from an end vertex to a start vertex.

%% file: Singleton.tex

\section{Removing Singleton Components}
\label{sec:removingsingletoncomponents}

In this section we show how to reduce the number of components in a 2-matching by removing singleton components.
Let $G$ be an instance of \STSP or \ATSP and let $M$ be a (directed) $2$-matching $M$ in~$G_x$ for some $x \in \SER(n,G)$.
Observe that the subgraph composed of all $1$-edges/arcs of $x$ is a (directed) $2$-matching in $G_x$.

\begin{lemma}
\label{lem:nosingle}
  There is an efficient algorithm that, given a (directed) 2-matching $M$ with a component that is a single vertex, finds a (directed) $2$-matching $M'$ in $G_x$ that improves~$M$.
  If $M$ contains all $1$-arcs/edges of $G_x$, then so does~$M'$.
\end{lemma}
\begin{proof}
We show by induction that there is a tree of special alternating paths starting from the singleton vertex such that we can grow the tree until we find an improvement.

  First, we show the lemma for directed graphs.
  Afterwards, the analogous result for undirected graphs follows easily.
  We assume that there are no basic improvements of $M$ as otherwise we are done.

  Within the proof, we write $\delta(S)$ as shorthand for $\delta_{G_x}(S)$.
  The basic idea of how to reduce the number of singleton components is as follows.
  Let $v$ be a vertex that forms a component in $M$. 
  If there is a vertex $w \in N^+(v)$ such that $w \in C$ for a component $C$ of $M$ and removing~$\delta^-_M(w)$ does not create a singleton component, we have found a suitable transformation by including~$(v,w)$ and removing $\delta^-_M(w)$.
  Since $x(\delta^-(w)) = 1$, $x(\delta^-_M(w)) < 1$ and thus we do not remove a $1$-arc.
  However, in general we need a sequence of transformations in order to ensure that the removal does not create a singleton component.
  In the following, we show how to find a sequence of such transformations.

  We show the claim of the lemma by induction on the size of a certain set~$S$ of vertices.
  Initially,~$S$ only contains $v$.
  A \emph{pseudo-component} $C'$ is a path of length one such that $C'$ is contained in a path~$P$ of~$M$, where $C'$ contains the start-vertex of $P$.
  In particular, a path of length one in $M$ is at the same time a pseudo-component.
  Let $\startS$ be the sets of start vertices in $S$, and $\pendS$ the set of pseudo-end vertices in $S$, \ie, ends of pseudo-components.

  We use a special type of alternating path. An arc in the path is a forward arc, if it is oriented towards the end and a backward arc if it is oriented towards the start.
  An end vertex is the end of a path in $M$ or any vertex of a cycle in $M$ (since by removing one arc, it becomes the end of a path).
  A sequence of arcs $Q$ is an \emph{alternating path of $M$} if 
  \begin{itemize}
    \item $Q$ is a path with alternating orientations of its arcs, \ie, for three vertices $\{z,y,z'\}$ such that $z$ and~$z'$ are adjacent to $y$ in $Q$,  the arcs are either $(z,y),(z',y)$ or $(y,z),(y,z')$;
    \item $Q$ starts with a tail of an arc at an end vertex;
    \item no forward arc of $Q$ is in $E(M)$;
    \item all backward arcs of $Q$ are in $E(M)$.
  \end{itemize}
  For a given set of vertices $S$, we say that $Q$ is a \emph{start-alternating path of $(M,S)$}, if additionally
  \begin{itemize}
    \item all vertices of $Q$ are in $S$;
    \item all tails of arcs in $Q$ are in \startS;
    \item all heads of arcs in $Q$ are in \pendS;
    \item the path $Q$ starts with $v$ as tail;
    \item the path $Q$ ends with the tail of an arc.
  \end{itemize} 
  We want $S$ to maintain the following invariants (see Fig.~\ref{fig:invariants}):
  \begin{enumerate}
    \item $S$ contains only $v$ and whole (vertex sets of) pseudo-components.\label{inv:whole}
    \item There is a collection $\mathcal{Q}$ of start-alternating paths of $(M,S)$ such that each vertex of $S$ is contained in at least one of the paths and the paths form a tree $T$ with $v$ as root such that all vertices in $\pendS$ have a degree of at most two within $T$.\label{inv:atree}
    \item There is no $1$-arc in the subgraph of $G$ induced by $S$.\label{inv:noone}
    \item If the total number of pseudo-components in $S$ is $k$ and we remove an arbitrary start vertex from $S$, there is a decomposition of the remaining vertices in $S$ into $k-1$ pseudo-components such that there are no singleton components.\label{inv:reduce}
    \item There is no arc $(s',s'')$ for any $s',s'' \in \startS$. \label{inv:start}
    \item $x(\delta^+(S) \cap \delta^+(\startS)) \ge 1.$ \label{inv:cut} 
  \end{enumerate}

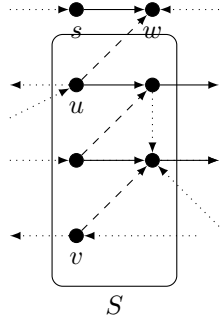
\begin{figure}[tb]
    \centering
      \centering
      \begin{tikzpicture}[->,scale=1]
        \node (v)  at (-3,0) [vertex, label=below:$v$]{};
        \node (s1) at (-3,1) [vertex]{};
        \node (s2) at (-3,2) [vertex, label=below:$u$]{};
        \node (s3) at (-3,3) [vertex, label=below:$s$]{};
        \node (e1) at (-2,1) [vertex]{};
        \node (e2) at (-2,2) [vertex]{};
        \node (e3) at (-2,3) [vertex, label=below:$w$]{};
        \node (l1) at (-1,0) {};
        \node (l1p)at (-1.3,0) {};
        \node (l2) at (-1,1) {};
        \node (l3) at (-1,2) {};
        \node (l4) at (-1,3) {};
        \node (r1) at (-4,0) {};
        \node (r2) at (-4,1) {};
        \node (r3) at (-4,2) {};
        \node (r3-1) at (-4,1.5) {};
        \node (r4) at (-4,3) {};
        \node (dummy) at (0,-1){};
        \draw[>=latex,dashed] (v) to node[midway,below] {} (e1);
        \draw[>=latex] (s1) -> (e1);
        \draw[>=latex, dashed] (s1) to node[midway,right] {}  (e2);
        \draw[>=latex] (s2) -> (e2);
        \draw[>=latex, dashed] (s2) to node[midway,right] {}  (e3);
        \draw[>=latex] (s3) -> (e3);
        \draw[>=latex, dotted] (l1) to node[midway,below] {} (e1);
        \draw[>=latex] (e1) -> (l2);
        \draw[>=latex] (e2) -> (l3);
        \draw[>=latex, dotted] (e2)to node[midway,right] {}  (e1);
        \draw[>=latex, dotted] (s2) to node[midway,below] {}  (r3);
        \draw[>=latex, dotted] (r2) to node[midway,below] {}  (s1);
        \draw[>=latex, dotted] (r3-1) to node[midway,below] {}  (s2);
        \draw[>=latex, dotted] (r4) to node[midway,below] {}  (s3);
        \draw[>=latex, dotted] (v) to node[midway,below] {} (r1);
        \draw[>=latex, dotted] (s1) -> (l2);
        \draw[>=latex, dotted] (l1p) to node[midway,below] {} (v);
        \draw[>=latex, dotted] (l4) to node[midway,below] {} (e3);
        \node[rectangle, draw, text width=4em, text centered, rounded corners, minimum height=9.5em] (S) at (-2.5,1) [label=below:$S$]{};
    \end{tikzpicture}
    \caption{
        \label{fig:invariants}
        Example of a set $S$ satisfying the invariants. 
        The solid arcs are arcs from $M$. 
        The path formed by dashed and solid arcs from $v$ to $u$ within $S$ forms a valid start-alternating path. The dotted arcs depict some types of the possible arcs of $G_x$ that are not in $M$.
        }
\end{figure}

  We show that we can either increase the size of $S$ or we find an improvement.
  Clearly, $S=\{v\}$ satisfies all invariants.

  Now suppose that $S$ contains $k$ pseudo-components.
  Then, by Invariant \ref{inv:cut}, there is an arc $e=(u,w)$ such that $u \in \startS$ and $w \notin S$.
  If~$w$ is a start vertex not in $S$, we add the arc $(u,w)$ and remove $u$ from $S$.
  By the induction hypothesis, this allows us to obtain $k-1$ pseudo-components in $S$ and, since pseudo-components are 
  starts of components in $G_x$, we are done.
  Thus we assume that $w$ is not a start vertex and therefore a vertex of a path~$P$ in $M$.
  If the arc $e' \in \delta_M^-(w)$ does not contain the start vertex of $P$,  we are also done: since $x(\delta^-(w))=1$ and $x_{(v,w)} > 0$, $x_{e'} < 1$ and the component of~$P$ without $e'$ that does not contain $w$ has at least two vertices.
  Therefore we may add $e$ to $M$, remove $e'$, and apply the induction hypothesis.

  The remaining possibility is that $w$ is a pseudo-end vertex of a pseudo-component.
  We claim that including the two vertices of $e'$ into $S$ still satisfies the invariants or we have found an improvement.
  Let $S' = S \cup \{w,s\}$, where $s$ is the second vertex of $e'$.

  Clearly, Invariant \ref{inv:whole} is satisfied.
  Similarly, Invariant \ref{inv:atree} follows easily since the two arcs $(u,w)$ and $(s,w)$ extend a start-alternating path or a sub-path.
  In the following, $\mathcal{Q}'$ is the set of start-alternating paths in $S'$, obtained from $\mathcal{Q}$ by adding the start-alternating path from $v$ to $s$ by extending a sub-path in $\mathcal{Q}$ (and therefore the union of paths in $\mathcal{Q}'$ is a tree again).
  Furthermore, $x_{(u,w)}<1$, since otherwise $w$ is a start vertex and we obtain an improvement by applying the induction hypothesis.
  Since $x_{(u,w)}>0$, $x_{(s,w)}<1$ which implies that Invariant \ref{inv:noone} is satisfied.
  To show that Invariant~\ref{inv:reduce} is valid for $S'$, let $k$ be the number of components in $S$ and thus $k+1$ is the number of components in $S'$.
  If we remove $s$, the component obtained by adding the arc $(u,w)$ with $u$ as start vertex has at least two vertices and by the induction hypothesis, removing $u$ from~$S$ leads to $k-1$ remaining components, that is, there are $k$ components left in $S'$.
  Removing another start vertex of $S$ also leads to~$k$ components by applying the induction hypothesis to~$S$ and keeping $(s,w)$ as separate component.

  For the Invariant~\ref{inv:start}, we show that the lemma follows directly if there are arcs from~$s$ to start vertices in $\startSp$.
  Suppose there is an arc $a = (s,s') \in \delta^+(\{s\},\startSp)$ for some vertex $s' \in \startSp$.
  Let~$Q$ be the start-alternating path in $\mathcal{Q}'$ that ends in $s$.
  Starting from~$v$, path~$Q$ has alternating forward- and backward arcs.
  We modify $M$ and obtain a $2$-matching~$M'$ by including all forward arcs of $Q$ and removing all backward arcs of $Q$.
  Additionally we include the arc $(s,s')$.
  Thus the number of added arcs is larger than the number of removed arcs and therefore it is sufficient to show that we did not create a cycle or a new singleton component. 
  There are no new cycles in $M'$, because all vertices except $v$ and $s'$ have the same degree in $M$ and $M'$, every second vertex of $Q$ is of degree one, and both the component containing $v$ and the one containing $s'$ are paths in $M'$.
  We did not create singleton components because each vertex in $Q$ has an incident arc in $M'$.
  The vertices not in $Q$ keep all of their incident arcs of~$M$.

  Finally, we show Invariant~\ref{inv:cut}. Let $k$ be the number of pseudo-components contained in $S'$. 
  Since $v$ is not part of a pseudo-component, $|\startSp| = k+1$ and thus, due to Invariant~\ref{inv:start}, $x(\delta^+(\startSp)) = k+1$.
  Note that $S' \setminus \startSp = \pendSp$ and $|\pendSp| = k$. 
  Therefore $x(\delta^-(\pendSp)) \le k$ and accordingly $x(\delta(\startSp,\pendSp)) \le k$. 
  We obtain that 
  \[
  x(\delta^+(S') \cap \delta^+(\startSp)) = x(\delta^+(\startSp) \setminus \delta^-(\pendSp)) \ge (k+1) - k = 1 \enspace .
  \]
  For an undirected graph $G$, we generate a directed graph $\overrightarrow{G}$ by replacing each edge $e=\{u,v\}$ by the two arcs $(u,v)$ and $(v,u)$.
  For a given solution $x$ to $\SER(n,G)$, we define $\overrightarrow{x}_{(u,v)} = \overrightarrow{x}_{(v,u)} = x_e/2$.
  It is not hard to check that if $x$ is a solution in $\SER(n,G)$, then $\overrightarrow{x}$ is a solution in $\SER(\overrightarrow{G},n)$.
  In $\overrightarrow{G}$, we can orient any path of $G$ in any direction.
  In particular, we can transform a given $2$-matching~$M$ in $G_x$ into a directed $2$-matching $\overrightarrow{M}$ in $\overrightarrow{G}_{\overrightarrow{x}}$ without changing the number of components, cycles, edges in cycles, or singleton components.
  Now we apply the already proved directed version of Lemma~\ref{lem:nosingle} to $\overrightarrow{G}$  and obtain an improved matching $\overrightarrow{M'}$. 
  If we apply the reverse transformation and replace each 2-cycle of $\overrightarrow{G}$ by a simple edge, we obtain an improved matching $M'$ in $G_x$.
  In order to include all $1$-edges into $M'$, we have to take some extra care.
  Note that since we do not have to distinguish between start and end vertices, the only possibility of a $1$-edge to be in a start-alternating path is that there is a pseudo component with a $1$-edge that belongs to a path in $M$ of length at least $2$.
  However, not both of the edges of the path can be $1$-edges and therefore, by reversing the orientation of all such paths, we avoid the problem.
\end{proof}

%% file: 12TSP.tex
\section{Alternating Paths}
\label{sec:alternating}
\label{sec:subcubicfractionallyhamiltoniansupports}
We now define undirected alternating paths.
A similar concept of alternating paths was used by Berman and Karpinski~\cite{BermanKarpinski2006}, but the details in our approach are different.
Let $G$ be an instance of \STSP, let $x$ be a solution to $\SERup(G)$, and let~$M$ be a $2$-matching in the support graph $G_x$.
We call edges in~$E(M)$ \emph{matching edges} and edges in $E(G_x) \setminus E(M)$ \emph{connecting edges} as they connect vertices that are not adjacent
within $M$.
A path $Q$ in $G_x$ with end vertices $s,t$ is \emph{alternating} if
\begin{enumerate}
  \item $s$ is an end vertex of $M$ and $s$ is incident to a connecting edge in $Q$;\label{alt:start}
  \item $t$ is an end vertex of $M$ and $t$ is incident to a connecting edge in $Q$;\label{alt:end}
  \item no internal vertex of $Q$ (\ie, in $V(Q) \setminus \{s,t\}$) is in a cycle of $M$;\label{alt:int-cycle}
  \item if $Q=\{s,t\}$, $s$ and $t$ are not both in the same cycle;\label{alt:no-cord}
  \item in $Q$, matching edges and connecting edges alternate.\label{alt:alt}
\end{enumerate}

Note that possibly $s=t$; then $s$ is incident to at least two connecting edges of~$Q$.
\newcommand{\inward}{inward\xspace}
An alternating path~$Q$ is \emph{\inward} if for any connecting edge $\{u,v\}$ in~$Q$,
\begin{enumerate}\addtocounter{enumi}{5}
  \item if $u,v$ are both in the same path~$P$ of $M$, either $u$ is an end vertex of~$Q$ or $\{u,v\},\{u,u'\}$ are consecutive edges of~$Q$, where $u'$ is the vertex adjacent to $u$ in $Q$ such that sub-path of~$P$ between $u$ and $v$ contains~$u'$ (see Fig.~\ref{fig:inwardexample}(a)).\label{alt:inward}
\end{enumerate}
\begin{figure}[tb]
  \centering
  \begin{tikzpicture}[scale=0.8]
    \node (u1) at (0,2) [vertex, label=below:$u$]{};
    \node (u2) at (1,2) [vertex]{};
    \node (u3) at (2,2) [vertex, label=below right:$u'$]{};
    \node (u4) at (3,2) [vertex, label=below:$v$]{};
    \node (u5) at (4,2) [vertex]{};
    \node (u6) at (5,2) [vertex]{};
    \node (m1) at (4,1) [vertex]{};
    \node (m2) at (5,1) [vertex]{};
    \node (l1) at (0,0) [vertex]{};
    \node (l2) at (1,0) [vertex]{};
    \node (l3) at (2,0) [vertex]{};
    \node (l4) at (3,0) [vertex]{};
    \node (l5) at (4,0) [vertex]{};
    \node (l6) at (5,0) [vertex]{};
    \node at (2.5,-1) {(a)};
    \draw[ultra thick] (u1) edge[bend left] (u4);
    \draw[ultra thick] (u4) edge[] (u3);
    \draw[ultra thick] (u3) edge[] (l3);
    \draw[ultra thick] (l3) edge[] (l2);
    \draw[ultra thick] (l2) edge[bend right] (l5);
    \draw[ultra thick] (l5) edge[] (l4);
    \draw[ultra thick] (l4) edge[] (m1);
    \draw[dashed] (u1) edge[] (u2);
    \draw[dashed] (u2) edge[] (u3);
    \draw[dashed] (u3) edge[] (u4);
    \draw[dashed] (u4) edge[] (u5);
    \draw[dashed] (u5) edge[] (u6);
    \draw[dashed] (m1) edge[] (m2);
    \draw[dashed] (l1) edge[] (l2);
    \draw[dashed] (l2) edge[] (l3);
    \draw[dashed] (l3) edge[] (l4);
    \draw[dashed] (l4) edge[] (l5);
    \draw[dashed] (l5) edge[] (l6);
  \end{tikzpicture}
  \hspace{2cm}
  \begin{tikzpicture}[scale=0.8]
    \node (u1) at (0,2) [vertex,label=below:$u$]{};
    \node (u2) at (1,2) [vertex]{};
    \node (u3) at (2,2) [vertex, label=below:$v$]{};
    \node (u4) at (3,2) [vertex, label=below right:$u'$]{};
    \node (u5) at (4,2) [vertex]{};
    \node (u6) at (5,2) [vertex]{};
    \node (m1) at (4,1) [vertex]{};
    \node (m2) at (5,1) [vertex]{};
    \node (l1) at (0,0) [vertex]{};
    \node (l2) at (1,0) [vertex]{};
    \node (l3) at (2,0) [vertex]{};
    \node (l4) at (3,0) [vertex]{};
    \node (l5) at (4,0) [vertex]{};
    \node (l6) at (5,0) [vertex]{};
    \node at (2.5,-1) {(b)};
    \draw[ultra thick] (u1) edge[bend left] (u3);
    \draw[ultra thick] (u4) edge[] (u3);
    \draw[ultra thick] (u4) edge[] (l3);
    \draw[ultra thick] (l3) edge[] (l2);
    \draw[ultra thick] (l2) edge[bend right] (l5);
    \draw[ultra thick] (l5) edge[] (l4);
    \draw[ultra thick] (l4) edge[] (m1);
    \draw[dashed] (u1) edge[] (u2);
    \draw[dashed] (u2) edge[] (u3);
    \draw[dashed] (u3) edge[] (u4);
    \draw[dashed] (u4) edge[] (u5);
    \draw[dashed] (u5) edge[] (u6);
    \draw[dashed] (m1) edge[] (m2);
    \draw[dashed] (l1) edge[] (l2);
    \draw[dashed] (l2) edge[] (l3);
    \draw[dashed] (l3) edge[] (l4);
    \draw[dashed] (l4) edge[] (l5);
    \draw[dashed] (l5) edge[] (l6);
  \end{tikzpicture}
  \caption{Example of (a) an \inward alternating path and (b) a violation of the inward property.}
\label{fig:inwardexample}
\end{figure}
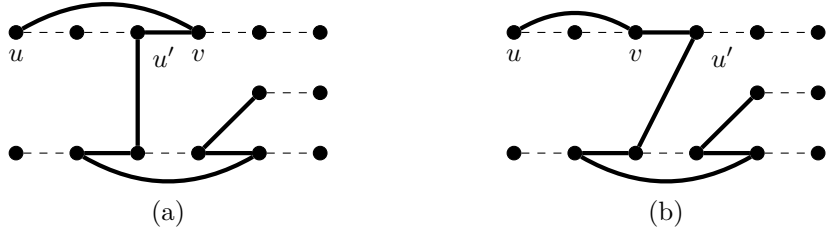

The definition of inward alternating paths uses that $\{u,v\} = \{v,u\}$ and thus the role of $u$ and $v$ can be exchanged. 
The intuition is that the path proceeds to the inside of cycles created by a connecting edges within a single path of $M$.
Suppose $u$, $v$, $Q$, and $P$ are as in property \ref{alt:inward} and $u'$ is adjacent to $u$ in~$Q$, but~$u'$ is not in the sub-path of $P$ between $u$ and $v$.
Then we say that $u$, $u'$, $v$, $\{u,u'\}$, and $\{u,v\}$ are \emph{involved} in the violation of the \inward property (see Fig.~\ref{fig:inwardexample}(b)).
The path $Q$ is a \emph{truncated alternating path} if it has the properties \ref{alt:start} and \ref{alt:int-cycle}--\ref{alt:alt}, but property \ref{alt:end} is violated.

A pair of vertices $s,t$ in a cycle $C$ is \emph{path-forming} if there is a path $P_{st}$ from~$s$ to $t$ in $G_x$ that contains all vertices of $C$ and no vertex of $V(G_x) \setminus V(C)$ (but~$P_{st}$ may use edges outside~$C$).
We say that we \emph{apply} an alternating path~$Q$ if we add all connecting edges of $Q$ to $M$ and remove all matching edges of $Q$.
If~$s$ is an end vertex of $Q$ and~$s$ is contained in a cycle $C$ of $M$, we additionally remove an edge incident to~$s$ in $C$ from $M$.
If both end vertices $s,t$ of $Q$ are in~$C$ such that $s$ and $t$ are path-forming in $C$, instead of removing two edges we replace the cycle by a path from $s$ to $t$ that visits all vertices of $C$.
If $Q$ starts and ends in the same vertex~$s$, we remove all edges outside $Q$ that are incident to $s$.

\begin{observation}
\label{appobs:alternatingend}
  Let $Q$ be an alternating path with two end vertices $s \neq t$ for a $2$-matching $M$ such that~$s$ and $t$ are end vertices of paths in $M$, and let $M'$ be  the $2$-matching obtained from applying $Q$ to $M$.
  Then all end vertices of paths in $M$ except $s$ and $t$ are also end vertices of paths in~$M'$.
\end{observation}
\begin{proof}
  For any vertex $u \notin \{s,t\}$ in a path of $M$, when applying $Q$ to $M$, the number of edges incident to~$u$ that are added matches the number of removed incident edges. In particular, no incident edge of~$u$ is changed if $u$ is not in $Q$.
  After applying $Q$, both $s$ and $t$ are vertices of degree two. 
\end{proof}

\begin{lemma}
\label{lem:alternating}
  Let $Q$ be an alternating path in $G_x$ with end vertices $s \neq t$ for a $2$-matching $M$.
  We can find an improved $2$-matching $M'$ if
  \begin{enumerate}[label=(\alph*)]
    \item $Q$ has length less than $3$; or \label{imp:one}
    \item $Q$ has length less than $5$ and both $s,t$ are end vertices of different paths; or \label{imp:three}
    \item $Q$ has length less than $7$, both $s,t$ are end vertices of paths in $M$, the vertices of $Q$ are not all in one path $P$ of $M$, and neither the first edge nor the last edge of $Q$ is
involved in a violation of the \inward property; or \label{imp:five}\label{imp:path}
  \item $Q$ is an \inward alternating path of length less than $7$ and if $s$ and $t$ are in one cycle $C$, they are path-forming in $C$; or   \label{imp:cycle} \label{imp:same}
  \item $Q$ has length less than $7$, both $s$ and $t$ are end vertices of one path $P$ in $M$, and neither the first edge nor the last edge of $Q$ is involved in a violation of the \inward property and there is a truncated \inward alternating path of length less than 5 from an end vertex of a path $P' \neq P$ to $P$.\label{imp:special}
\end{enumerate}
\end{lemma}
\begin{proof}
  We explore the number of cycles that can occur if the different types of alternating paths are applied.

  Suppose condition \ref{imp:one} is satisfied.
  By the definition of alternating paths, the length of $Q$ is odd.
  Therefore, $Q$ is a single edge between two end vertices and, due to property~\ref{alt:no-cord} in the definition of alternating paths, there is a basic improvement.
  That is, applying $Q$ decreases the number of components or, if $s$ and~$t$ are the two end vertices of one path, $M'$ has the same number of components and an increase in the number of cycles.

  For the remaining conditions, let us consider the following situation.
  Suppose there is a cycle~$C'$ in~$M'$ that is not in $M$.
  Then $C'$ is composed of edges in $E(M)$ and connecting edges in $E(Q)$.
  In particular,~$C'$ contains at least one connecting edge of $Q$. 
  This insight allows us to analyze the number of newly created cycles with respect to the number of connecting edges of $Q$.
  We will also use frequently that, unless $s$ or~$t$ is contained in a cycle of $M$, applying $Q$ increases the number of edges.

  Suppose condition \ref{imp:three} is satisfied.
  Then $|E(M')| > |E(M)|$ and the length of $Q$ is at most three.
  Due to condition \ref{imp:one}, we only have to consider a length of exactly three and thus $Q$ has exactly two connecting edges.
  At least one of these is between two different paths of $M$ and thus it cannot be the only connecting edge of a cycle.
  Therefore either the number of cycles increases by one (and the number of components stays the same) or the number of components decreases.

  Suppose condition \ref{imp:five} is satisfied, and therefore there are at most three connecting edges in~$Q$.
  We have $|E(M')| = |E(M)|+1$ and the statement of the lemma holds, unless the number of components is increased.
  By contradiction, let us assume that that $M'$ has more components than $M$.
  We note that the first and last edge of~$Q$ are not contained in cycles of $M'$ that only contain one connecting edge of~$Q$:
  either such an edge $e$ of~$Q$ spans between two components of~$M$ or the inward property ensures that there is a matching edge on the path of $M$ between the two ends of $e$. 
  In both cases a possible cycle of $M'$ contains at least two connecting edges.
  Thus the only possibility to increase the number of components is that we obtain one cycle with only one connecting edge and one cycle with two connecting edges such that the middle edge of $Q$ closes a cycle $C$ in a path $P'$ of $M$.
  As a result, the second cycle has to include all remaining edges of~$P'$ since there are only two connecting edges of $Q$ left and they have to be incident to $s$ and $t$. 
  Since by our assumption not both $s$ and $t$ are in $V(P')$, we conclude that the second cycle cannot be formed and thus $M'$ cannot have more components than $M$.

  Suppose condition \ref{imp:cycle} is satisfied.
  If $|E(M')| > |E(M)|$, the application of $Q$ did not remove any edges in cycles of $M$ (\ie, $s$ and $t$ are ends of paths in $M$).
  If $V(Q) \subseteq V(P)$ for some path~$P$ in $M$, the inward property enforces the direction of every edge of $Q$ and as a result, applying $Q$ transforms~$P$ into a single cycle. 
  If $V(Q)$ contains vertices not in $P$, condition \ref{imp:path} is satisfied and we are done.

  If $|E(M')| \le |E(M)|$, at least one end of $Q$ is contained in a cycle $C$ of $M$.
  Since after the application of $Q$ all vertices of $C$ are still in one component of $M'$, either $C$ is entirely included in a cycle of $M'$ or~$C$ is entirely included in a path of $M'$.
  This implies that unless $M'$ has a cycle that is not in $M$, the number of components is decreased.

  Due to the matching edges enforced by the inward property, we can exclude that in $M'$ there is a cycle $C'$ that is not in $M$ and contains only one connecting edge of $Q$.
  In particular, we may exclude that the middle edge of $Q$ is contained in a cycle of $M'$.
  We conclude that the only possibility to have a new cycle in~$M'$ 
  is that both $s$ and $t$ are in $C$.
  However, then the total number of edges in cycles is increased without changing the number of components or cycles.

  Suppose condition \ref{imp:special} is satisfied.
  If the length of $Q$ is at most three, applying~$Q$ results in a single cycle and we are done.
  Otherwise applying $Q$ creates up to two cycles.
  Additionally, the truncated \inward alternating path until the first vertex within $P$ becomes an \inward alternating path of length less than~$5$ to an end vertex of a cycle.
  By the proof of condition \ref{imp:cycle}, the number of components is reduced again and thus in total the number of components stays the same whereas the number of cycles increases.
\end{proof}

\begin{lemma}
\label{lem:alternatingclosed}
  Let $Q$ be an alternating path with respect to $M$ of length less than~$7$ whose both ends are the same end vertex $s$ of a path $P$ in $M$ and let $u$ be the vertex adjacent to $s$ in $M$.
  There is an improved $2$-matching~$M'$ if
  \begin{enumerate}[label=(\alph*)]
    \item $u$ is adjacent to an end vertex $t$ in another  path of $M$; or \label{imp:closed1}
    \item $\{s,u\}$ forms a component and there is an alternating path $Q'$ with respect to $M$ of length less than 7 that has both its ends in $u$\label{imp:closed3}
  \end{enumerate}
\end{lemma}
\begin{proof}
  Suppose condition \ref{imp:closed1} of Lemma~\ref{lem:alternating} is satisfied.
  Let $M''$ be the $2$-matching obtained from~$M$ by applying $Q$.
  By Observation~\ref{appobs:alternatingend}, $t$ is an end vertex of a path in~$M''$.
  Thus, applying $Q$ and adding~$\{u,t\}$ gives a new $2$-matching $M'$.
  We now argue that $M'$ is an improved $2$-matching compared to~$M$.
  
  The number of edges in $M'$ is $|E(M)|+1$, since applying $Q$ did not change the number of edges and we introduced $\{u,t\}$.
  Therefore, similar to the proof of Lemma~\ref{lem:alternating}, there is an improvement unless there are at least two cycles in $M'$ that are not in $M$.
  Again, each of the cycles has to contain connecting edges of~$Q$ and $Q$ has at most three connecting edges.
  In particular, the two connecting edges of $Q$ incident to~$s$ are in the same component of $M'$ and therefore the only possibility to generate two cycles is to have one cycles with one connecting edge of $Q$ and one cycle with the two connecting edges of $Q$ incident to~$s$.
  But then there are four internal vertices of $Q$ that are contained in one path of $M$ and there is no possibility to close a cycle with only the two connecting edges of $Q$ incident to $s$.

  Suppose condition \ref{imp:closed3} of Lemma~\ref{lem:alternating} is satisfied.
  By Lemma~\ref{lem:alternating}\ref{imp:path}, we may assume that $Q$ and $Q'$ are edge disjoint. In particular, this means that we may apply them one after the other.
  Let us first apply $Q$.
  This way, $u$ becomes a singleton component.
  We will show that there are at most two further applications of alternating paths such that in the end we obtain an improvement.
  \smallskip

  \noindent
  \textbf{Case 1:} Applying $Q$ does not introduce new cycles. Then we also apply $Q'$. 
  Since $\{s,u\}$ is only removed once, the two applications of alternating paths increase the number of edges.
  The application of $Q'$ can introduce at most one cycle and thus either the number of components is reduced or it stays the same as before and the number of cycles is increased.
  \smallskip

  \noindent
  \textbf{Case 2:} Applying $Q$ introduces a cycle that contains $\{s\}$.
  We first apply the alternating path $\{s,u\}$ such that $u$ together with the introduced cycle forms a new path.
  Note that the application of $Q$ also created a path $P$.
  If the length of $P$ is one, there is a basic improvement and overall we reduced the number of components.
  Otherwise, the number of components, cycles, edges in cycles, and singletons stayed unchanged but the number of components of size two is reduced.
  \smallskip

  \noindent
  \textbf{Case 3:} Applying $Q$ introduces a cycle that does not contain $\{s\}$.
  Then we apply an alternating path analogous to Lemma~\ref{lem:nosingle}: we introduce the edge $\{s,u\}$ and remove an edge incident to $s$ of LP value smaller then one.
  Let $\{s,s'\}$ be the removed edge. Then $s'$ is the end vertex of a path and $s'$ is incident to a cycle. 
  Thus we have found a basic improvement and overall we have reduced the number of components of size two while the number of components, cycles, edges in cycles, and singletons stayed unchanged. 
\end{proof}

The following lemma is simple but useful.
\begin{lemma}
\label{lem:threecycleimp}
  Let $M$ be a $2$-matching in graph $G_x$, let $s$ be the end vertex of a path in $M$ and let $Q=\{s,u\},\{u,v\},\{v,s\}$ be a path from $s$ to $s$ of length three such that $\{u,v\}$ is in some path~$P$ of~$M$.
  Let $\{u',u\}$ and $\{v,v'\}$ with $u' \neq v$ and $v' \neq u$ be the two edges incident to $u$ and $v$ in $P$ (if they exist).
  Then there is an improvement~if
  \begin{enumerate}[label=(\alph*)]
    \item $u$ or $v$ is an end vertex of $P$; or
    \item there is an end vertex $t \neq s$ and an edge $e=\{t,w\}$ for $w \in \{u',u,v,v'\}$ such that 
    $e \notin E(M)$.
  \end{enumerate}
\end{lemma}
\begin{proof}
  If $u$ or $v$ is an end vertex of $P$ then we simply add $\{s,u\}$ or $\{s,v\}$ to the 2-matching as to obtain a 2-matching with fewer components.

  If there is an edge $e$ as specified in the second condition, we add the edge $\{t,w\}$, remove~$\{w,u\}$ or~$\{w,v\}$ from $P$ (only one of them can be in $P$), and add $\{s,u\}$ or $\{s,v\}$ depending on the removed edge.
  Then either we reduced the number of components or increased the number of cycles without increasing the number of components.
\end{proof}

\section{Computing Structured 2-Matchings in the Support of \SERup}
\label{sec:integralitygapupperboundsfor12stsp}
\label{sec:thealgorithm}
In the following we give an algorithm that computes a 2-matching $M$ in the
support of an optimal solution to $\SERup(G)$ for instances $G$ of \STSP such
that~$M$ has some nice properties. 

\begin{algorithm}[hbt]
  \caption{\label{alg:stsp}\STSP algorithm. 
  }
  \SetAlgoNoLine
  \SetKwInOut{Input}{Input}
  \SetKwInOut{Output}{Output}
  \Input{An instance $G$ of \STSP such that $\SERup(G)$ has a solution.}
  \Output{A $2$-matching $M$ of $G$.}
  Compute an optimal solution $x^*$ of $\SERup(G)$\;
  Let $M$ be the $2$-matching in $G_{x^*}$ with $E(M)=\emptyset$ and let $M'$ be the $2$-matching in~$G_{x^*}$ that contains all $1$-edges \tcp*[r]{We may assume that $E(M') \neq \emptyset$~(Boyd and Pulleyblank~\cite{BoydPulleyblank1990})}
  \While{$M \neq M'$}{
    $M := M'$\;
    Apply Lemma~\ref{lem:nosingle} to $M'$
    \tcp*[r]{Remove the singleton components.}
    If there is an alternating path $Q$ in $G_{x^*}$ that satisfies properties of Lemma~\ref{lem:alternating}, Lemma~\ref{lem:alternatingclosed}, 
    or Lemma~\ref{lem:threecycleimp}, obtain an improved $M'$\;
  }
\end{algorithm}
Let us analyze the algorithm.
Each iteration of the \texttt{while} loop, except the last one, gives an improved $2$-matching.
Hence, Algorithm~\ref{alg:stsp} computes a feasible solution to \STSP.
We can compute~$x^*$ in polynomial time.
Since the total number of improvements is bounded by a polynomial in~$n$, the number of iterations of the while loop is also polynomial in $n$.
Since Lemma~\ref{lem:nosingle} provides an efficient algorithm and we can list all alternating paths of a fixed constant length in polynomial time, also all steps in the \texttt{while} loop can be done in polynomial time.
Thus, Algorithm~\ref{alg:stsp} runs in polynomial time.

In some theorems we require that paths in the $2$-matching $M'$ do not end with degree-$2$ vertices.
A natural way to obtain that property would be to ensure that all $1$-edges stay in $M'$, but this does not seem to be true in general.
However, without further assumptions we are able to show the following.

\begin{lemma}
\label{lem:enddegree}
  The execution of Algorithm~\ref{alg:stsp} 
  yields a $2$-matching so that all end vertices of paths have a degree of at least 3 in~$G_{x^*}$.
\end{lemma}
\begin{proof}
  We show the following more general statement:
  \begin{quote}
    \emph{Let $M$ be a $2$-matching in some graph $G_x$ such that each end vertex of a path in~$M$ has a degree of at least $3$ in $G_x$.
    If there is an improved $2$-matching $M''$ obtained from applying Lemma~\ref{lem:nosingle}, Lemma~\ref{lem:alternating}, Lemma~\ref{lem:alternatingclosed}, 
    or Lemma~\ref{lem:threecycleimp}, then there is an improved $2$ matching $M'$ such that each end vertex $s$ of a path in $M'$ has degree at least $3$ in~$G_x$.}
  \end{quote}

  Suppose $s \neq t$ and we obtained $M''$ by applying an alternating path $Q$ within one of the lemmas.
  Then, by Observation~\ref{appobs:alternatingend}, in $M''$ both the degree of $s$ and $t$ is $2$ and all other vertices in paths of $M$ have the same degree as in $M$.
  In particular, no degree-$1$ vertex of a path in~$M$ can become an end vertex.
  The only intersection of $Q$ with cycles of $M$ can occur if $s$ or $t$ are in some cycle, as otherwise we would have considered a shorter alternating path.
  If both $s$ and~$t$ are in the same cycle $C$, then we only considered~$Q$ if $s$ and $t$ are path-forming in $C$ and the application of $Q$ did not create an end vertex of a path in~$V(C)$.
  Otherwise, assume that~$s$ belongs to a cycle $C$ with $t \notin V(C)$.
  Let $e$ be the edge incident to $s$ in~$Q$.
  Since $x_e > 0$ and $x(\delta_{G_x}(s)) = 2$, $s$ can have at most one incident $1$-edge.
  When applying $Q$, we can choose to remove one of two edges in $C$ and in $M'$ we choose to not remove a $1$-edge of $C$ (while in~$M''$ a $1$-edge may have been removed).
  We conclude that for all improvements due to Lemma~\ref{lem:alternating} the claim of this lemma holds, since each improvement is obtained due to one or two applications of alternating paths.
  The same is true for Lemma~\ref{lem:nosingle} and Lemma~\ref{lem:threecycleimp}.

  Now let us assume that $Q$ is an alternating path that both starts and ends in the same vertex~$s$.
  Applying $Q$ potentially creates a degree $2$ vertex, but when using 
  Lemma~\ref{lem:alternatingclosed} this is excluded, since we apply subsequent alternating paths. 
\end{proof}

\section{Integrality Gap Upper Bounds for \texorpdfstring{Symmetric $(1,2)$-TSP}{Symmetric (1,2)TSP}}
\label{sec:integralityupperboundsfor12stsp}
We now introduce our method to determine the quality of the solution computed by
Algorithm~\ref{alg:stsp}.

\medskip

\newcommand{\LPX}{\ensuremath{\mathrm{LP}(x^*)}\xspace}
\newcommand{\LPN}{\ensuremath{\mathrm{LP}(x)}\xspace}
Our general approach is as follows.
Let~$G$ be an instance of \STSP and let $x^*$ be an optimal solution to $\SERup(G)$.
Further, let $M$ be a $2$-matching in $G_{x^*}$ and $\mathcal{C}$ be the set of components of $M$.
Our approach is to map the LP values $x^*_e$ of edges $e$ to components of $M$, in such a way that either the minimum sum of LP values (over all mappings) is at least some value $\alpha$ or we find an improvement to $M$.
To this end, we introduce a new linear program \LPX with variables~$y_{C,e}$ for each pair of a component $C \in \mathcal{C}$ and edge $e \in E(G_{x^*})$
\begin{align}
  \sum_{e \in E(G_{x^*})} y_{C,e}  &\ge \alpha &&\mbox{ for all } C \in \mathcal{C},\label{appcon:alpha}\\
  \sum_{C \in \mathcal{C}} y_{C,e} &\le x^*_e  &&\mbox{ for all } e \in E(G_{x^*}),\label{appcon:double}\\
                           y_{C,e} &\ge 0      &&\mbox{ for all } C \in \mathcal{C}, e \in E(G_{x^*})
\end{align}

We show that we either find an improvement of $M$, or we find a feasible solution to \LPX.
The rationale behind finding a feasible solution to \LPX is as follows.
Within \LPX, each $x^*_e$ is a fixed constant.
Since $\cost(x^*)=n$, all edges in $M$ are of cost one. 
Therefore, to a feasible solution of \LPX we have to add at most $n/\alpha$ edges of cost at most two in order to obtain a tour of $G$.
In other words, a feasible solution to \LPX can be augmented to a tour of $G$ of cost $n + n/\alpha$.
This way, we obtain an $(\alpha+1)/\alpha$-approximation for \STSP.

\medskip

We will now give the details of the approach.
We start with the following known result.
\begin{lemma}[Wolsey~\cite{Wolsey1980}]
\label{applem:pointset}
  Let $S$ be a set of vertices in $G$, $S \neq V(G)$. Then $x^*(\delta_G(S) \cup \delta_G(S,S)) \ge |S| + 1$.
\end{lemma}
\begin{proof}
  Since $x^*(\delta_G(S))\ge 2$ and $\sum_{v\in S} x^*(\delta_G(v))=2|S|$, we have
  \begin{equation*}
      x^*(\delta_G(S) \cup \delta_G(S,S)) = (x^*(\delta_{G}(S)) + \sum_{v \in S} x^*(\delta_{G}(v)))/2 \ge |S|+1,
  \end{equation*}
  where we have to divide the second term by two since the LP value of each edge is added twice.
\end{proof}

At this point we can already show that our approach works for $\alpha = 4$.
\begin{theorem}
\label{thm:stsp}
  There is a polynomial-time $5/4$ approximation algorithm for \STSP with respect to $\OptSERup(G)$.
\end{theorem}
\begin{proof}
  Let $M$ be the $2$-matching computed by Algorithm~\ref{alg:stsp} 
  and let $x^*$ be the corresponding optimal solution to $\SERup(G)$.
  Instead of $\delta_{G_{x^*}}$, we simply write~$\delta$.

  To analyze the algorithm, we construct a solution to \LPX for $\alpha = 4$.
  In order to obtain meaningful alternating paths, we ``complement'' the values $x^*$ for edges of $M$.
  Formally, let $z^*$ be the vector of length $|x^*| = |E(G)|$ with entries $0 \le z^*_e \le 1$ for all $e \in E(G)$.
  We set $z^*_e = x^*_e$ for all $e \in E(G) \setminus E(M)$ and $z^*_e = 1-x^*_e$ for all $e \in E(M)$.

  We note two properties of $z^*$.
  For any internal vertex $v$ of a component in $M$,
  \begin{equation}
  \label{appeq:internal}
    \begin{split}
      z^*(\delta(v)\setminus E(M)) & = 2 - x^*(\delta(v) \cap E(M))\\
                                 & = 2 - (2 - z^*(\delta(v) \cap E(M)))\\
                                 & = z^*(\delta(v) \cap E(M)) \enspace .
    \end{split}
  \end{equation}
  
  For any end vertex $s$ of a path in $M$,
  \begin{equation}
  \label{appeq:end}
    \begin{split}
      z^*(\delta(s)\setminus E(M)) - 1 & = 1 - x^*(\delta(s) \cap E(M))\\
                                     & = 1 - (1 - z^*(\delta(s) \cap E(M)))\\
                                     & = z^*(\delta(s) \cap E(M)) \enspace .
    \end{split}
  \end{equation}

  To simplify the discussion, we subdivide the edges as follows.
  Let $N \in \mathbb{N}$ be the smallest integer such that $x^*_e \cdot N$ is an integer for all $e \in E(G)$, where we used that each $x^*_e$ is a rational number.
  Note that this way, also $z^*_e \cdot N$ is an integer.
  Now we define scaled versions $x,z$ of $x^*,z^*$ by setting $x_e = N \cdot x^*_e$ and $z_e = N \cdot z^*_e$ for each $e \in E(G)$.
  Correspondingly, we aim to construct a solution $y$ to \LPN for $\alpha=4N$.

  For each edge $e$, we introduce $N$ parallel edges $e_1,e_2,\dotsc,e_N$ which we call the \emph{sub-edges} of $e$. 
  We extend the vectors~$x$ and $z$ by setting $x_{e_i} = 1$ for all $i \le x_e$ and $x_{e_i}=0$ for all remaining $i$.
  If $e \notin E(M)$, we set $z_{e_i} = x_{e_i}$ for all $i$ and, if $e \in E(M)$, $z_{e_i} = 1$ for all indices $x_e < i \le N$; for all remaining $i$ we set $z_{e_i} = 0$. A sub-edge of $e$ with index $i$ is a sub-edge of $x_e$ if $x_{e_i}=1$.
  Correspondingly the sub-edges of $x$ all sub-edges of edges in the support of $x$. The sub-edges of $z$ are defined analogously.

  The purpose of the subdivision is that we can assign sub-edges to components.
  That is, we reduced the problem to find a feasible solution $y$ to \LPN to assigning at least $4N$ sub-edges of $x$ to each of the components.

  Note that for each internal vertex $v$ of a path in $M$, by \eqref{appeq:internal},
  \begin{multline}
  \label{appeq:internalsub}
    |\{e_i :  e \in \delta(v) \cap E(M), 1 \le i \le N, z_{e_i} > 0\}|
        = |\{e_i :  e \in \delta(v) \setminus E(M), 1 \le i \le N, z_{e_i} > 0\}|
  \end{multline}
  and for an end vertex $s$ of a path in $M$, by \eqref{appeq:end},
  \begin{multline}
  \label{appeq:endsub}
    |\{e_i :  e \in \delta(s) \cap E(M), 1 \le i \le N, z_{e_i} > 0\}|
      = |\{e_i :  e \in \delta(s) \setminus E(M), 1 \le i \le N, z_{e_i} > 0\}| - N \enspace .
  \end{multline}
  Let $I$ be the set of edges that lead from an end vertex of a path of $M$ into the same path, that is,
    \[
    I := \{\{s,v\} \in E(G_{x}) \setminus E(M) :
    \mbox{$s$ is an end vertex of a path $C$ in $M$ and $v$ is in $C$}\}\enspace .
    \]
  The edges in $I$ are somewhat problematic, as they may become truncated inward alternating paths of length one.
  To this end, we exhibit a pairing of sub-edges from $I$ with sub-edges of either $x$ or $z$.
  A pairing with sub-edges of $x$ ensures that we can assign them directly to components and a pairing with~$z$ ensures that we can form inward alternating paths of length at least two.

  We obtain a family $\mathcal F$ of sub-edge pairs by iterating over all sub-edges $e'=\{s,v\}$ of $z$ at edges in~$I$, where $s$ is an end-vertex of some component $P$.
  Let $\{u,v\}$ be the edge of $P$ such that $u$ is located between~$s$ and $v$ in $P$.
  If there is an unpaired sub-edge $e''$ of $z$ at $\{u,v\}$, we add the pair $\{e',e''\}$ to $\mathcal F$.
  Otherwise, there is an unpaired sub-edge $e'''$ of $x$ at $\{u,v\}$ and we add $\{e',e'''\}$ to $\mathcal F$.
  Let $\mathcal F_z$ be the subset of pairs in $\mathcal F$ whose second component is a sub-edge of $z$, and let $\mathcal F_x$ be the subset of pairs in $\mathcal F$ whose second component is a sub-edge of $x$.
  
  Let us now consider the following procedure to create alternating paths.
  Initially all sub-edges are unmarked.
  The procedure then marks some sub-edges by one of two types: either to be used as internal edge of an alternating path, or to be used as the end of an alternating path.
  
  \begin{algorithm}[H]
  \begin{enumerate}
    \item Choose an end vertex $s$ of a path in $M$ with less than $N$ sub-edges that are marked to be\linebreak used as ends of alternating paths. 
    \item Extend a path $Q$ from $s$ by choosing unmarked sub-edges of $z$ that do not violate the \inward property (property \ref{alt:inward} of the definition) and that alternate between~$E(G_z) \setminus E(M)$ and\\
    sub-edges within paths of $M$.
    \begin{enumerate}
      \item Whenever $Q$ uses a sub-edge contained in some pair in $\mathcal F_z$, extend $Q$ with the second\\ sub-edge of that pair.
      \item Whenever $Q$ uses a sub-edge contained in some pair in $\mathcal F_x$, stop (and leave a truncated alternating path).
      \item If an end vertex $t$ of a path in $M$ is reached by a sub-edge not in $M$ and there are\linebreak less than $N$ edges marked as ends of paths incident to $t$, \textbf{stop}.
      \item If a cycle is reached, \textbf{stop}.
      \item If there is no unmarked sub-edge that can be followed, \textbf{stop}.
    \end{enumerate}
    \item Mark the first and (if not truncated) the last sub-edge of $Q$ to be an end of an alternating  \linebreak path and mark the remaining sub-edges of $Q$ to be internal sub-edges.
  \end{enumerate}
    \caption{Creating an \inward alternating path or a truncated \inward alternating path.}
  \end{algorithm}
  It is not hard to check that any path $Q$ created by the procedure is either an \inward alternating path or a truncated \inward alternating path.
  Let \apset be the set of paths obtained by iteratively applying the procedure until each end vertex of each path in~$M$ is the end vertex of exactly~$N$ inward alternating paths or truncated inward alternating paths.
  Note that by \eqref{appeq:endsub}, such a set~\apset exists.
  We emphasize that we need \apset only for the analysis.

  Let us now construct a solution $y$ to \LPN.
  For each cycle $C$ of $M$ we set $y_{C,e}= 1$ for all sub-edges~$e$ of $x$ incident to vertices of $C$.
  For the remaining cases, let us fix a path $Q \in \apset$.
  \begin{itemize}
    \item If $Q$ is an \inward alternating path of length less than $7$ with both ends at one vertex $s$ of component $C$, we set $y_{C,e} = 1$ for each sub-edge $e$ of $x$ such that $e \in \delta(s)$.
    \item If $Q=e_1,e_2,\dotsc,e_k$ is an \inward alternating path of length at least $7$ from component $C_1$ to component $C_2$ (possibly with $C_1=C_2$), we set $y_{C_1,e_1} = y_{C_1,e_3} = y_{C_2,e_k} = y_{C_2,k-2}=1$. 
    \item If $Q=e_1,e_2,\dotsc,e_k$ is a truncated \inward alternating path of length at least three that starts from a path $C$ of $M$, we set $y_{C,e_1} = y_{C,e_3} = 1$.
    (There is no truncated \inward alternating path starting from a cycle.)
    \item If $Q=e_1,e_2$ is a truncated \inward alternating path of length two,
      by the properties of~$z$, the only reason that $Q$ was forced to stop is that any prolongation violates the \inward property.
      Let~$u$ and~$v$ be the vertices such that $e_1=\{s,u\}$ and $e_2=\{u,v\}$.
      Let $P_s$ be the path of $M$ containing $s$ and~$P_u$ be the path of $M$ that contains $u$ (possibly $P_u = P_s$).
      Let $P'_u$ be the sub-path of $P_u$ without $u$ that is left when removing $v$ (see Fig.~\ref{fig:truncated}).
      If there is a sub-edge $e$ of $x_{\{v,w\}}$ for some $w$ in $P'_u$ such that $y_{C',e}=0$ for all $C'$, we set $y_{P_s,e}=1$.
      Otherwise we set $y_{P_s,e'}= 1$ for all sub-edges $e'$ of $x$ with $e' \in \delta(s)$.
      \begin{figure}[ht]
        \centering
        \begin{tikzpicture}[scale=0.8]
          \node (0) at (0,0) [label=below:$P_u$]{};
          \node (1) at (2,0) [vertex, label=below:$u$]{};
          \node (2) at (3,0) [vertex, label=below:$v$]{};
          \node (3) at (4,0) [label=below:$P'_u$]{};
          \node (4) at (5,0) [vertex]{};
          \node (5) at (6,0) [vertex]{};
          \node (6) at (7,0) [vertex]{};
          \node (7) at (9,0) {};
          \draw[] (0) -- (1);
          \draw[] (1) -- (2);
          \draw[] (2) -- (4);
          \draw[] (4) -- (5);
          \draw[] (5) -- (6);
          \draw[] (6) -- (7);
          \draw[thick] (3)--(7);
          \draw[dashed] (2) edge[bend left] (4);
          \draw[dashed] (2) edge[bend left] (5);
          \draw[dashed] (2) edge[bend left] (6);
        \end{tikzpicture}
      \caption{Assignment for truncated inward alternating paths of length two.}
      \label{fig:truncated}
    \end{figure}
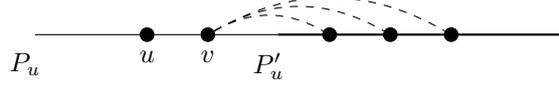
    \item Finally, let us consider the remaining case, that is, a truncated alternating path $e_1$ of length one that starts from a path $C$.
      This is only possible if $\{e_1,e'\} \in \mathcal F_x$ for some sub-edge $e'$ of $x$.
      We set $y_{C,e_1} = y_{C,e'} = 1$.
  \end{itemize}
  We now show that the constraints \eqref{appcon:double} are satisfied by $y$.
  If we only consider cycles, alternating paths between two different vertices, alternating paths with one end vertex of length at least $7$, and truncated alternating paths of length at least three, it is easy to verify that we assigned the LP value of each sub-edge to at most one component, and thus we did not violate the constraints.
  Note that we did not yet assign LP values of any edge in $E(M)$.

  For truncated alternating paths $Q=e$ of length one, additionally to the edge of the truncated alternating path we assigned the LP value of one sub-edges $e'$ of $E(M)$.
  We did not assign the value of one of these edges for two different truncated alternating paths of length one, since this would imply the existence of an \inward alternating path of length three which is excluded due to Lemma~\ref{lem:alternating}\ref{imp:cycle}.

  In the remaining cases, we may have to set $y_{C,e}=1$ for all sub-edges of $x$ with $e \in \delta(s)$ where $s$ is an end vertex of a path $P$ in $M$.
  These assignments do not interfere with assignments for cycles as otherwise there is a basic improvement.
  Note also that the sub-edges in $E(M)$ used in the assignment for truncated alternating paths of length one are not in $\delta(s)$.

  Let $s'$ be the vertex adjacent to $s$ in $P$.
  Suppose that $P$ has at least three vertices and there is no end vertex $t \notin E(P)$ of a path in $M$ such $\{s',t\} \in E(G_z)$.
  Let us consider any alternating path $Q \in \apset$ that contains a sub-edge of~$x_{\{s,s'\}}$.
  We split $Q$ into $Q_{s'}$ and $Q_s$ by removing the sub-edge of $x_{\{s,s'\}}$ such that $s \in V(Q_s)$.
  By our assumption and the definition of alternating paths, the length of $Q_{s'}$ is at least three.
  Note that due to the \inward property, the vertex adjacent to $s'$ in $Q_{s'}$ is not the other end vertex of $P$.
  By condition \ref{imp:cycle} of Lemma \ref{lem:alternating}, we can also exclude that $Q_s$ has less than 7 edges, unless $Q_s$ is truncated or $Q_s$ leads from $s$ to $s$.
  Therefore, none of the LP values of sub-edges incident to $s$ is assigned to another component than~$P$.

  As a result, we may assume for the remaining cases that either $P$ has exactly two vertices or $s'$ is adjacent to an end vertex of another path than $P$ of $M$.
  We lead both situations to a contradiction.

  By Lemma~\ref{lem:alternatingclosed}, there is no alternating path of length less than $7$ from~$s$ to $s$ unless $P$ has exactly two vertices and there is no alternating path shorter than $7$ that starts and ends at the second vertex.
  Therefore, if $Q$ is such an alternating path, assigning values of sub-edges in $\delta(s)$ to $P$ does not interfere with any of the previous cases.

  We continue with truncated \inward alternating paths of length two.
  If we already have assigned the third sub-edge consecutive to $Q$, we are done.
  Otherwise, we have assigned the sub-edges incident to $s$.
  As before, let $s,u,v$ be the vertices of the alternating path $Q$ and let $P_u$ be the path of $M$ that contains~$u$ and its sub-path $P_u'$.
  By the definition of \inward alternating paths and the definition of $z$, there is a vertex $w$ in $P'_u$ and a truncated alternating path $\hat{Q}$ such that it either is of length three and ends in $v$ with a sub-edge of $\{w,v\}$ or $\hat{Q}$ is of length two, ends in $w$, and a sub-edge of $\{w,v\}$ was assigned to the component where $\hat{Q}$ ends.
  Let $\tilde{Q}$ be the alternating path obtained by concatenating $Q$, $\hat{Q}$, and (if the length of $\hat{Q}$ is two) a sub-edge of~$\{u,w\}$.
  The length of $\tilde{Q}$ is exactly five.
  Let $w'$ and $t$ be the two remaining vertices of $\tilde{Q}$ such that $t$ is the end vertex.
  The only vertices of $\tilde{Q}$ that are possibly involved in the violation of the \inward property are those adjacent to $v$ and $w$ in $\tilde{Q}$, i,\,e., $u$, $v$, $w$, and $w'$.
  In particular, the first and last edge of $\tilde{Q}$ are not involved in the violation of the inward property.

  Therefore, by Lemma~\ref{lem:alternating}\ref{imp:five}, $\hat{Q}$ cannot start from a path other than~$P_s$.
  Since $u$, $v$, $w$, and $w'$ are in the same component of $M$ and none of them is an end vertex, $P_s$ has a length of at least~$5$.
  By our considerations above,
  also $s \neq t$.
  By Lemma~\ref{lem:alternating}\ref{imp:same} of Lemma~\ref{lem:alternating}, $\tilde{Q}$ contains internal vertices of $P_s$.

  Since $u$, $v$, $w$, and $w'$ are in the same component of $M$, we conclude that $\tilde{Q}$ does not leave~$P_s$.
  Since~$v$ and $w$ are the only vertices involved in the violation of the \inward property, we conclude that $\tilde{Q}$ has the type described in \ref{imp:special} of Lemma~\ref{lem:alternating}.
  Since $P_u$ has more than two vertices, by our assumption there is an end vertex of some path $P' \neq P_u$ adjacent to $P_u$; this case was considered by condition \ref{imp:special}.

  \medskip

  Finally we show that the constraints \eqref{appcon:alpha} are satisfied for $\alpha = 4N$.
  By Lemma~\ref{applem:pointset} this is clearly true for cycles.
  For each end vertex of a path in $M$, note that we assigned either $2N$ sub-edges to the component of $s$ or all sub-edges in $\delta(s)$, which also sums up to $2N$.
  Therefore, we assigned~$4N$ disjoint sub-edges to each path in $M$, which amounts to a total value of $4$ with respect to \LPX.
\end{proof}

When a half-integral solution becomes available, we can improve the approximation factor.
\begin{theorem}
\label{thm:halfint}
  If we can obtain an optimal half-integral solution $x^*$ to $\SERup(G)$ in polynomial time, then there is a polynomial time $7/6$ approximation algorithm for \STSP with respect to $\OptSERup(G)$.
\end{theorem}
\begin{proof}
  Since Proposition~\ref{pro:aux} preserves half-integrality, also Lemma~\ref{lem:pluseps} does and we may assume that all edges in $x^*$ have a cost of one.
  Let $M$ be the $2$-matching obtained from Algorithm~\ref{alg:stsp}. 
  Additionally, during the execution of the algorithm we also consider a very specific type of improvement, discussed in the end of this proof.

  The analysis of the algorithm is similar to the previous one, but for the constraints \eqref{appcon:alpha} we show that we may choose $\alpha=6$.
  In particular, we use the same $z$ as in the proof of Theorem~\ref{thm:stsp}.

  We use a modified set $\apset$, where for each new alternating path we mark the edges as we did in the original construction of $\apset$.
  We start with the construction of $\apset$ as in the proof of Theorem~\ref{thm:stsp}.
  Then, for each unmarked sub-edge of $x$ leaving a cycle, we extend a new \inward alternating path in the same manner as we did for end vertices of paths. 
  As a result, additionally to the $N$ sub-edges incident to each end vertex of a path of $M$, all sub-edges within \apset leaving cycles are marked to be ends of alternating paths.

  We swap certain edges in order to prevent a problematic type of path.
  We want to avoid truncated alternating paths of length one that are sub-paths of $1$-edges.

  Let $Q\in\apset$ be such a truncated alternating path of length one starting from $s$ such that its sub-edge~$e$ is part of a $1$-edge $\{s,u\}$.
  Let $e'$ be the second sub-edge of $x_{\{s,u\}}$.
  By the construction of $\apset$, $e'$ is in a (truncated) alternating path $Q' \in \apset$ and $Q'$ does not stop at $u$.

  By Lemma~\ref{lem:enddegree}, there is a third sub-edge $e''$ of $x$ with $e''\notin E(M)$ that is incident to $s$ and $e''$ is not a sub-edge of a $1$-edge.
  If $e''$ is not marked, we remove $Q$ from $\apset$ and extend a new \inward alternating path from~$e''$.
  Otherwise, we recombine the alternating paths incident to $s$.
  Let $Q'' \in \apset$ be the alternating path containing $e''$.
  If $Q''$ ends in $s$, $Q'$ does not since there are exactly two alternating paths ending in~$s$ and also $Q$ ends in $s$. 
  Then we cut $Q'$ at $s$.
  We extend the part of $Q'$ not containing $e'$ with $e$ and~$Q$ becomes the remaining path of $Q'$ starting with $e'$ at $s$.
  If $Q''$ does not end in $s$, we cut $Q''$ at $s$. Then the part of $Q''$ containing $e''$ ends in $s$ and we extend the other part of $Q''$ with $e$.
  We apply these modifications to all ``bad'' alternating paths in $\apset$ and update the marking of edges.

  Since we only consider edges in the support of $x^*$, for each matching edge $e$ such that a sub-edge of~$e$ is in $Q$, $z^*_e = x^*_e = 1/2$.

  Let us construct a solution $y$ to $\LPN$.
  The types of assignments are performed in the order as listed, that is, each type of assignment is applied iteratively until there is no further occurrence of that type.
  As in Theorem~\ref{thm:stsp}, we argue via the extended version of $\LPN$ for sub-edges.
  \smallskip

\noindent
{\bf Paths.}
\smallskip

\noindent
  Fix an end vertex $s$ of a path $P_s$ in $M$ and a path $Q \in \apset$ starting from~$s$.
  Let $s'$ be the vertex adjacent to~$s$ in $P_s$.
  We aim to find a set of $6$ sub-edges for~$s$ such that for each sub-edge $e$ in the set, $y_{P_s,e} = 1$.
  This way, the constraints \eqref{appcon:alpha} are satisfied for $\alpha=6$ since $P_s$ has two ends and thus 12 assigned sub-edges.

  In the following, we either directly find $6$ sub-edges whose values were assigned to $P_s$ or, since the degree of $s$ is at least three, we assign two times the values of $3$ sub-edges to $P_s$.
  \smallskip

  \noindent
  {\bf -- $\mathbf{Q}$ is a path of length at least 3.} 
  \smallskip
  
  \noindent
  Let $e_1,e_2,e_3$ be the first three sub-edges of $Q$.
  We set $y_{P_s,e_1} = y_{P_s,e_2} = y_{P_s,e_3} = 1$.
  This assignment covers the following situations entirely, since the assignment can take place on both sides of $Q$ without assigning values twice 
  (and therefore, in particular, we do not violate the constraints~\eqref{appcon:double}).
  \begin{itemize}
  \item ${Q}$ leads to the end vertex of a path which is not $s$.
  Then its length is at least $7$, by Lemma~\ref{lem:alternating}\ref{imp:path},\ref{imp:same}.
  \item ${Q}$ leads to a cycle.
  Then its length is at least $5$, by Lemma~\ref{lem:alternating}\ref{imp:cycle}.
  \item ${Q}$ leads from ${s}$ to ${s}$ and its length is at least 7.
  \end{itemize}

  Otherwise, we have to assign additional values. We distinguish the following cases.
  \smallskip

  \noindent
  Case 1: Suppose ${Q}$ leads from ${s}$ to ${s}$ and its length is exactly 5.
  Then we assign $y_{P_s,e} = 1$ for all~$e$ in~$Q$ and we still have to find an additional sub-edge.
  Let us consider the edge $\{s,s'\}$. 
  \smallskip

  \noindent
  Case 1.1: If $x_{\{s,s'\}} = 2$, none of its sub-edges is in any path of $\apset$.
  We choose one of the two sub-edges of $\{s,s'\}$, $e$, such that $y_{P,e}=0$ for all paths $P$ in $M$ and set $y_{P_s,e}=1$.
  Note that we that we consider $\{s,s'\}$ at most twice and therefore we can always find such a sub-edge $e$.
  \smallskip

  \noindent
  Case 1.2: If $x_{\{s,s'\}} = 1$, there is a sub-edge $e' \in \delta(v)$ that is neither a sub-edge of~$x_{\{s,s'\}}$ nor in $Q$.
  Then we set $y_{P_s,e'}=1$.
  We have to exclude that $e'$ was used for a previous assignment.
  \smallskip

  \noindent
  Case 1.2.1: If~$e'$ is in no alternating path of $\apset$, a double use would directly imply the existence of a basic improvement.\\
  Case 1.2.2: Otherwise, let $Q'\in\apset$ be the path containing $e'$.\\
  Case 1.2.2.1: If $s$ is an end vertex of $Q'$, there is nothing else to do since in this case, $e'$ is already assigned tp $P_s$ (\ie, $y_{P_s,e'}=1$) due to $Q'$.\\
  Case 1.2.2.2: Otherwise let $Q'_1$ and $Q'_2$ be the two sub-paths of $Q'$ that end at $s$ such that $e'$ is in $Q'_2$.
  Note that $Q'_2$ itself is a (truncated) \inward alternating path.
  By Lemma~\ref{lem:alternatingclosed}\ref{imp:closed1}, if $Q'_1$ leads to an end vertex, either its length is at least two and it leads to a cycle, or the length is at least four and it leads to the end of a path in $M$.
  In all cases, we did not assign $e'$ to any component yet.

  \noindent
  Case 2: Suppose $Q$ leads from $s$ to $s$ and its length is at exactly $3$.
  Let $s,u,v$ be the three vertices of $Q$ where~$\{u,v\}$ is an edge of a path $P$ in $M$.
  Lemma~\ref{lem:threecycleimp} implies that neither $u$ nor $v$ is an end vertex of~$P$.
  Also none of the sub-edges incident to $u$ or $v$ in $P$ is used by any of the assignments considered until now since this would require a violation of property (b).
  Let us consider the set $S$ of sub-edges incident to $u$ or $v$ within $Q$ or $P$.
  Then $|S| \ge 5$ and for each $e \in S$, we set $y_{P_s,e} = 1$.
  If $|S|=5$, we assign one more sub-edge to~$P_s$ analogously to Case 1.
  \smallskip

  \noindent
  {\bf -- $\mathbf{Q}$ is a path of length less than 3.}
  \smallskip

  \noindent
  Case 1: Suppose $Q = e$ is a truncated alternating path of length one.
  Then, due to the edge pairings in the construction of $\apset$ and the transformation of $\apset$, $Q$ is incident to a $1$-edge $e'$ of $P_s$ in direction towards $s$.
  Note that~$e'$ is not at either end of $P_s$ and, since $Q$ is a single edge to an end vertex,~$e'$ has not been used by any of the previous assignments (since we have Lemma~\ref{lem:threecycleimp}).
  We set $y_{P_s,e''} = 1$ for each of the two sub-edges~$e''$ of $e'$, and we set $y_{P_s,e}= 1$.
  Note that $e''$ is not assigned due to two different truncated alternating paths of length one, since otherwise there is an \inward alternating path of length three considered in Lemma~\ref{lem:alternating}\ref{imp:cycle}.
  \smallskip

  \noindent
  Case 2: Suppose that $Q = (e,e')$ is an alternating path of length exactly two.
  We assign the values similar to the proof of Theorem~\ref{thm:stsp}.
  That is, we set $y_{P_s,e} = y_{P_s,e'} = 1$.\\
  Case 2.1: If we find a sub-edge~$e''$ of $x$ adjacent to $e'$ that was not yet assigned (see the proof of Theorem~\ref{thm:stsp}), we set $y_{P_s,e''}=1$.
  \smallskip

  \noindent
  Case 2.2: Suppose that we were not able to assign three sub-edges.
  Then, as we have seen in the proof of Theorem~\ref{thm:stsp}, we can find an alternating path $\tilde{Q}$ of length $5$ from $s$ to some end vertex $t$ such that the first and last edge of $\tilde{Q}$ does not violate the \inward property.
  If $s=t$, the situation is exactly the same as for \inward alternating paths of length $5$ from $s$ to $s$ considered above.
  Otherwise, as we have seen in the proof of Theorem~\ref{thm:stsp}, we may assume that all vertices of $\tilde{Q}$ are in $P_s$ and the neighbor of $s$ in $P_s$ is not adjacent to the end vertex of another path.
  We distinguish two possibilities.
  (i) If $Q$ is the only truncated alternating path of length $2$ starting from $s$, we assign the value of a third edge in the same way as we did for the last edge in the case of \inward alternating paths from $s$ to $s$ of length $5$.
  (ii) Otherwise, we conclude that there is a second alternating path~$\tilde{Q}'$ of length $5$ from $s$ to $t$ and for $t$, we already have assigned $6$ sub-edges from $\tilde{Q}$ and $\tilde{Q}'$ to $P_s$.
  For each of the two alternating paths, we assigne the value of a third edge in the same way as we did for the last edge in the case of \inward alternating paths from $s$ to $s$ of length $5$, one of them incident to $s$ and the other one incident to $t$.
  \smallskip

  \noindent
  \textbf{Cycles.}
  \smallskip

  \noindent
  For each cycle $C$ of $M$ we set $y_{C,e}=1$ for all sub-edges $e$ of $x$ incident to vertices of $C$.
  \smallskip

  \noindent
  {\bf -- $C$ has length at least 5.}
  By Lemma~\ref{applem:pointset}, we have assigned a sufficient amount to satisfy the constraints~\eqref{appcon:alpha} of \LPN.
  \smallskip

  \noindent
  {\bf -- $C$ has length 4.}
  There are at least two different (truncated) alternating paths $Q,Q'$ starting from~$C$ with the edges $e_1,e_2$ resp.~$e'_1,e'_2$, 
  since there are at least $4$ sub-edges of $x$ leaving $C$ and there are no truncated alternating paths of length one starting from cycles. 
  Then both $e_2$ and $e_2'$ are matching edges.
  We set $y_{C,e_2} = y_{C,e'_2} = 1$ and thus, by Lemma~\ref{applem:pointset}, we have assigned a sufficient amount to satisfy the constraints \eqref{appcon:alpha} of \LPN.
  \smallskip

  \noindent
  {\bf -- $C$ has length 3.}
  Similar to $4$-cycles, if $|V(C)|=3$, for each alternating path $Q$ starting from $C$ with $e_1,e_2$, we set $y_{C,e_2}=1$. 
  However, the assignment does not provide a guarantee that sufficiently many sub-edges were assigned to $C$.
  By the subtour elimination constraints, there are at least $4$ alternating paths starting from $C$. 
  \smallskip

  \noindent
  Case 1: If no sub-edge $e_2$ was used twice, we have assigned an amount of at least two to $C$ additionally to the previously assigned value of at least $4$ (by Lemma~\ref{applem:pointset}).
  \smallskip
  
  \noindent
  Case 2: Otherwise, an edge $e_2$ was assigned twice, which implies that there is an alternating path from $C$ to $C$ of length exactly three. 
  By Lemma~\ref{lem:alternating}\ref{imp:cycle}, such an alternating path cannot have two different end vertices.
  We therefore name the vertices such that $V(C)=\{v_1,v_2,v_3\}$ and there is an alternating path~$Q$ of length three from $v_1$ to $v_1$.
  Then $\mathrm{deg}_{G_{x^*}}(v_1)=4$ and therefore $x^*_{\{v_1,v_2\}} = x^*_{\{v_1,v_3\}} = 1/2$.
  \smallskip
  
  \noindent
  Case 2.1: If also $x^*_{\{v_2,v_3\}} = 1/2$, then $\mathrm{deg}_{G_{x^*}}(v_2) = \mathrm{deg}_{G_{x^*}}(v_3) = 4$ and there are $6$ alternating paths starting from $C$ with at most three edges assigned twice.
  Then the total amount of values assigned to~$C$ is at least $6$, which is sufficient.
  \smallskip

  \noindent
  Case 2.2: The remaining case is that $x^*_{\{v_2,v_3\}} = 1$, in which case we have assigned $11$ sub-edges 
  to $C$ (see Fig.~\ref{appfig:triangle}).
  Thus, we aim to assign one more sub-edge to $C$.
  Let $u,u'$ be the two remaining vertices of $Q$ and let $P$ be the path of $M$ that contains $u$ and $u'$. 
  By Lemma~\ref{lem:alternating}\ref{imp:cycle}, there is no \inward alternating path of length less than $7$ from $v_1$ to another end vertex and therefore there are vertices $w,w' \in V(P)$ such that $\{w,u\},\{u,u'\},\{u',w'\} \in E(P)$.
  By the definition of $z$, $x^*_{\{u,u'\}} = 1/2$.
  \smallskip

  \noindent
  Case 2.2.1: If either $u$ or $u'$ is incident to a $1$-edges $e'$, we choose one of its two sub-edges $e'_1$ such that $y_{C',e'_1}=0$ for all components $C'$ and set $y_{C,e'_1}=1$.
  Since as in our previous discussion $e'$ is considered at most twice (once from each side), we can always find such a sub-edge $e'_1$.

  \noindent
  Case 2.2.2: Otherwise, there are vertices $s,s'$ such that there are edges $\{u,s\},\{u',s'\} \in E(G_{x^*})$.
  
  If a sub-edge $e'$ of $\{u,w\}$ has not yet been assigned, it is save to set $y_{C,e'} = 1$ and we are done, since consecutive appearences of Case 2.2.2 in $P$ are excluded by Lemma~\ref{lem:alternating}\ref{imp:cycle}.
  Therfore we may assume that $e'$ has already been assigned, which implies that either $u$ or $w$ has a connecting edge leading to an end vertex that is not $v_1$.
  However, $w$ cannot have an incident connecting edge to an end vertex, since otherwise either there is an improvement or its path in \apset does not include $e'$.
  We conclude that $s$ is an end vertex.
  In particular, $s$ has to be an end vertex of $P$ since otherwise there is an improvement.

  We now introduce a new type of improvement, announced in the beginning of the proof.
  Let us recall the discussion of alternating paths of length three for paths (Lemma~\ref{lem:alternating}\ref{imp:three}).
  We did not allow such paths to have ends at cycles, since the application of the alternating path may introduce a new cycle $C'$ without changing the number of components and thus it does not necessarily provide an improvement.
  There is an improvement, however, if the alternating path starts at a cycle of length three and the newly created cycle has a length of more than three: 
  the number of vertices in cycles is increased without decreasing the number of cycles or increasing the number of components.
  In the algorithm, we now additionally consider alternating paths of length three between two end vertices if such an increase of the cycle length occurs.
  In particular, in the following we assume that from $C$ there is no alternating path of length three to an end vertex of a path such that applying the alternating path produces a cycle of length at least four.
  Since the improvement is an application of an alternating path, we do not introduce end vertices of degree two.
  The total number of considered improvements for all appearances of the special situation only increases by a polynomial factor. 
  The additional improvement implies that $\{s,w\} \in E(P)$. 
  Note that $\{s,w\}$ must have a sub-edge $e''$ that has not yet been assigned: 
  it is not contained in any path of \apset as one of the first three edges and alternating paths from $s$ may only have assigned a sub-edge of $\{s,w\}$ if $\{s,w\}$ is a $1$-edge.
  We set $y_{C,e''}=1$ such that there are at least $12$ sub-edges assigned to $C$.
      \begin{figure}[h]
        \centering
        \begin{tikzpicture}[scale=0.8]
          \node (v1)  at (2.5,1) [vertex, label=right:$v_1$]{};
          \node (v2)  at (2,2) [vertex, label=above:$v_2$]{};
          \node (v3)  at (3,2) [vertex, label=above:$v_3$]{};
          \node (s)   at (0,0) [vertex, label=below:$s$]{};
          \node (w)   at (1,0) [vertex, label=below:$w$]{};
          \node (u)   at (2,0) [vertex, label=below:$u$]{};
          \node (up)  at (3,0) [vertex, label=below:$u'$]{};
          \node (wp)  at (4,0) [vertex, label=below:$w'$]{};
          \node (sp)  at (5,0) [vertex, label=below:$s'$]{};
          \node (v2p) at (1,2.2) {};
          \node (v3p) at (4,2.2) {};
          \node (wa) at (0,-0.2) {};
          \node (wpa) at (5,-0.2) {};
          \node (sa) at (-1,0) {};
          \node (spa) at (6,0) {};
          \draw[] (s) -- (w);
          \draw[dashed] (w) -- (u);
          \draw[dashed] (u) -- (up);
          \draw[dashed] (up) -- (wp);
          \draw[] (wp) -- (sp);
          \draw[dashed] (v1) -- (v2);
          \draw[dashed] (v1) -- (v3);
          \draw[] (v2) -- (v3);
          \draw[dashed] (v2) -- (v2p);
          \draw[dashed] (v3) -- (v3p);
          \draw[dashed] (v1) -- (u);
          \draw[dashed] (v1) -- (up);
          \draw[dashed] (s) to[bend left] (u);
          \draw[dashed] (sp) to[bend right] (up);
          \draw[dashed] (w) -- (wa);
          \draw[dashed] (wp) -- (wpa);
          \draw[dashed] (sp) -- (spa);
          \draw[dashed] (s) -- (sa);
        \end{tikzpicture}
      \caption{Special case of three cycles. The dashed edges have an LP-value of $1/2$.}
      \label{appfig:triangle}
    \end{figure}
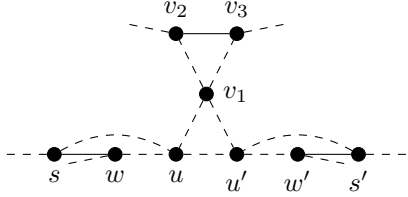
\end{proof}

\section{Subcubic Support Graphs}
\label{app:subcubicfractionallyhamiltoniansupports}
In this section we show the following theorem. 

\begin{theorem}
\label{appthm:subcubicapprox}
  For instances of \STSP admitting an optimal basic solution to $\SER(G)$ with subcubic support such that $\OptSER(G)=|V(G)|$, the tight upper bound on the integrality gap is $10/9$.
\end{theorem}
The argumentation to proof this theorem differs considerably from the previous two, because we do not assign the LP values to the components.
Instead, we charge every vertex with a coin and redistribute these coins fractionally such that each component obtains at least $9$ coins.
The basic idea is that cycles in the $2$-matching are well-behaved and that for each path we can collect two coins for each edge leaving an end vertex. However, the complications arise from the interferences if there are edges from the end of a path to the path itself.

Before we present the actual result of this subsection, let us observe some properties of cycles in a $2$-matchings $M$ within the support graph.

For each component $C$ of $M$, let $\out(C)$ be the set of edges in $\delta(V(C))$ incident to end vertices of $C$.

\begin{lemma}
\label{applem:cyclecut}
  Let $G_x$ be a subcubic support graph of an optimal basic solution $x$ of $\SER(G)$, and let~$C$ be a cycle of a $2$-matching $M$ in $G_x$ with $|V(C)| \le 6$.
  Then $|\out(C)| \ge 3$ and there are two vertices $u,v \in V(C)$ such that $u$ and $v$ are path-forming and $|\out(C) \cap \delta_{G_x}(\{u,v\})| = 2$.
\end{lemma}
\begin{proof}
  Clearly $|\out(C)| \ge 2$, since the constraints of \SER imply that $x_e \le 1$ for each edge $e \in E(G)$, unless $G$ consists of only $2$ vertices (if $x_e > 1$, $x(\delta_{G_x}(V(e))) < 2$).
  For the sake of contradiction, suppose that $|\out(C)|<3$ and let $s$ and $t$ be the two vertices with incident edges leaving $C$.
  Then the degree of $s$ and $t$ is $3$ and the two edges leaving $C$ are $1$-edges since $x(\delta_{G_x}(V(C)))\ge 2$.
  As a consequence, the edges incident to $s$ and $t$ within $C$ cannot be $1$-edges (due to the equality constraints).
  Since a vertex in~$G_x$ can only have degree $2$ if it is incident to two $1$-edges, all adjacent vertices of $s$ and $t$ have degree $3$.
  However, this implies that there is a chord in $C$ (an edge $e \notin E(C)$ with both ends in $V(C)$).
  The vertices~$s$ and~$t$ cannot be neighbors, since otherwise $x(\delta_{G_x}(V(C) \setminus \{s,t\})) < 2$.
  Therefore, $|V(C)| \notin \{3,5\}$.
  If $|V|=4$ there is one chord, and if $|V| = 6$ there are two chords.
  It is not hard to check that in all possible arrangements, the solution is a convex combination of two paths from $s$ to $t$, contradicting that~$x$ is a basic solution.

  The second claim follows immediately from our considerations since either there are two neighbors within $C$ that both have neighbors in $V(G_x) \setminus V(C)$ or $C$ has chords. 
\end{proof}

For the proof of Theorem~\ref{appthm:subcubicapprox}, let $x$ be an optimal basic solution to $\SER(G)$ with subcubic support graph $G_x$.
We write $\delta$ as shorthand for $\delta_{G_x}$.

We start with running Algorithm~\ref{alg:stsp} 
and obtain a 2-matching $M$ of the support graph $G_x$ for which none of the considered improvements are possible.
However, within the proof we use some types of improvements that are quite specific and therefore we did not include them into Lemma~\ref{lem:alternating}.
We will argue for these improvements that we do not violate Lemma~\ref{lem:enddegree} if we consider them in Algorithm~\ref{alg:stsp}.

For each vertex $v\in V(G_x)$, let $\mathsf{comp}_M(v)$ denote the component of $M$ that $v$ belongs to (that is, the component of $M$ that has either one   or two edges incident to $v$), and call each vertex $v'\in V(G_x)$ that is adjacent to $v$ in $M$ a \emph{component neighbor} of $v$ (with respect to $M$).

We now analyze the number of components in $M$.
To this end, we put coins on each vertex, and redistribute these coins according to certain rules.
Initially, we assign one coin to each vertex of $G$.
Then we redistribute the coins fractionally to the components of $M$.
We show how to assign at least nine coins to each component of $M$.

\subsection{Cycles}
\label{app:subcubic-cycles}
We first show how to assign at least nine coins to each component $C$ of $M$ that is a cycle.
Let $N_{\textnormal{end}}(C) = \{w\in V(G)\setminus V(C)~|~\{v,w\} \in \delta(V(C))\}$ be the set of vertices that are neighbors of~$C$.
We define a set $S$ depending on the size of $C$:
\begin{itemize}
  \item If $|V(C)| \ge 7$, set $S = \{v\}$ for each vertex $v \in V(S)$ that has some neighbor outside $C$.
  \item If $V(C) \in \{4,5,6\}$, set $S = \{u,v\}$ for the two vertices considered in Lemma~\ref{applem:cyclecut}.
  \item If $|V(C)| = 3$, set $S = V(C)$.
\end{itemize}
Let $S' := \{v \in N_{\textnormal{end}}(C)~|~\{v,w\} \in \delta(S)\mbox{ for some }w \in V(G_x)\}$ be its set of neighbors not in~$C$.

Recall that each vertex $v$ in the support graph is equipped with exactly one coin.
Now we redistribute these coins and halfs of them to all cycles $C$, according to the following rules:
\begin{enumerate}[leftmargin=0.8in]
  \item[\emph{Rule C1:}] each vertex $v\in V(C)$ assigns its whole coin to $C$.
  \item[\emph{Rule C2:}] each component neighbor $w$ of $S'$ assigns $1/2$ of its coin to $C$.
\end{enumerate}
Additionally, if $|V(C)|=4$, let $s \in V(C) \setminus S$ be such that $\{s,t\} \in E(G_x)$ for $t \notin V(C)$.
By Lemma~\ref{applem:cyclecut}, such a vertex $s$ exists.
\begin{enumerate}[leftmargin=0.8in]
  \item[\emph{Rule C3:}] If $t$ is not a component neighbor of some vertex in~$S'$, assign the coin of $t$ to $C$.
  \item[\emph{Rule C4:}] If $t$ is a component neighbor of some vertex in $S'$, assign half a coin of $t$ to $C$, and assign half a coin from $t$'s component neighbor $w \notin S'$ to $C$.
\end{enumerate}
Next we show that all considered vertices exist, and that we did not assign more coins than available.

First, notice that all vertices in $S'$ are internal vertices of paths, for otherwise there is a basic improvement because there is an alternating path of length one between two end vertices.

Second, fix a cycle $C' \neq C$ in $M$.
Let $w \in N_{\textnormal{end}}(C)$ and $w' \in N_{\textnormal{end}}(C')$, and let $v\in V(C),v'\in V(C')$ be such that $\{v,w\}\in \delta(V(C)),\{v',w'\}\in \delta(V(C'))$.
Notice that $w\notin V(C')$ and, symmetrically, $w'\notin V(C)$ as otherwise there is a basic improvement (an alternating path of length one between $C$ and $C'$).
Now, for any path $P$ in $M$ that contains both $w,w'$, the two vertices cannot be consecutive in $P$, for otherwise there is an inward alternating path of length~$3$, which is excluded due to Lemma~\ref{lem:alternating}\ref{imp:cycle}.

Third, let $w,w' \in N_{\textnormal{end}}(C)$, and let $v,v'\in V(C)$ be such that $\{v,w\},\{v'w'\}\in \delta(V(C))$.
Now, for any path $P$ in $M$ that contains both $w,w'$, if $v$ and $v'$ are path-forming, these two vertices cannot be consecutive on $P$.
Otherwise, there is an inward alternating path of length~$3$, which is excluded due to Lemma~\ref{lem:alternating}\ref{imp:cycle}.

Fourth, let $C$ be a cycle with $|V(C)|=4$ and let $\{u,v,s\} \subset V(C)$ be the three vertices considered above.
Let $u',v',s'$ their neighbors in $N_{\textnormal{end}}(C)$ where $S' = \{u',v'\}$.
Then $s'$ cannot be incident to both~$u'$ and $v'$, since either $s$ and $u$ or $s$ and $v$ are path forming, a situation considered above.

With these four observations, clearly all assigned coins are available.
Note that each cycle obtained at least $9$ coins.

We conclude this subsection with an observation that follows immediate from the four observations above and that will be useful for assigning coins to paths.
\begin{observation}
\label{appobs:assignment}
  Let $C$ be a cycle in $M$ and let $u,v \in V(C)$ be two vertices that are not path-forming.
  Let $u',v'$ be the neighbors of $u$ and $v$ not in $V(C)$.
  Then
  \begin{itemize}
    \item the total number of coins assigned to $C$ by $u$ and $v$ is at most $3$, and
    \item half a coin of each component neighbor of either~$u'$ or of~$v'$ has not been assigned to any cycle, unless~$u'$ or $v'$ has a neighbor $w'$ such that the whole coin of $w'$ was assigned to $C$,
  \end{itemize}
\end{observation}

\subsection{Paths}
\label{app:subcubic-paths}
Now we are dealing with paths $P, P', P'', \hdots$ that form components of the 2-matching $M$.
For a path $P$, let $v_1,u_1,u_2,\hdots,u_t,v_2$ be an ordering of its such that $\{v_1,u_1\},\{u_t,v_2\}$ and $\{u_i,u_{i+1}\}$ for $i = 1,\hdots,t-1$ are exactly the edges of $P$.
Hence, $v_1,v_2$ are the two end vertices of $P$, which form the set $\mathsf{end}(P)$.
By $\mathsf{int}(P) = \{u_1,\hdots,u_t\}$ we denote the set of \emph{internal vertices} of $P$.

Our goal is to show that each path $P$ of $M$ can receive at least 9 coins, where $P$ receives coins according to the following rules:
\begin{itemize}[leftmargin=0.8in]
  \item[\emph{Rule P1:}] Every end vertex $v\in\mathsf{end}(P)$ sends 1/2 coin to the path $P$.
  \item[\emph{Rule P2:}] Every internal vertex $v\in\mathsf{int}(P')$, for $P'\not=P$, that is adjacent to an end vertex of~$P$, sends 1 coin to $P$.
  \item[\emph{Rule P3:}] Every vertex $v\in\mathsf{int}(P')\cup\mathsf{end}(P')$, for $P'\not=P$, that is not adjacent to an end vertex of $P$ but is adjacent to a vertex covered by Rule P2, sends 1/2 coin to $P$.
  \item[\emph{Rule P4:}] Every other vertex $v$ sends its coin to the path $\mathsf{comp}_M(v)$ that it belongs to.
\end{itemize}
Notice that these half coins have not been assigned to cycles or to other paths before since otherwise there is a basic improvement.
They may, however have been assigned to the same path twice.
We will ensure that we do not double-count them.

To analyze the redistribution of coins, we use the tool of alternating paths, which we denote by $Q, Q', Q'',\hdots$.
For each type of path $P$ that can occur, we show that it can receive at least~9 coins from its internal vertices, its end vertices and vertices of paths $P'\not= P$ connected to~$P$ by connecting edges, for otherwise we can find an improvement $M'$ of $M$ by exhibiting some alternating paths.

Similar to cycles, let $N_{\textnormal{end}}(P) = \{w\in V\setminus V(P)~|~\{v,w\} \in \delta(\mathsf{end}(P),V \setminus V(P))\}$ be the set of vertices in components of $M$ other than $P$ that are adjacent to end vertices of $P$.
As was the case with cycles, $N_{\textnormal{end}}(P)$ only contains internal vertices of paths.
Recall that the degree of the vertices in $\mathsf{end}(P)$ is exactly three, by Lemma~\ref{lem:enddegree}.
We distinguish five cases, for $|\out(P)|\in \{0,1,2,3,4\}$.
Due to the degree restriction, $|\out(P)| \le 4$.

\newcounter{word}
\makeatletter
\newcommand*{\LBL}{%
  \@dblarg\@LBL
}
\def\@LBL[#1]#2{%
  \begingroup
    \renewcommand*{\theword}{#2}%
    \refstepcounter{word}%
    \label{#1}%
    #2%
  \endgroup
}
\makeatother

\smallskip
\noindent
\textbf{\LBL{Case 1}:} $|\out(P)| = 4$.

\noindent
\textbf{\LBL{Case 1.1}:} The four vertices in $N_{\textnormal{end}}(P)$ form an independent set.\\
Then the path $P$ receives at least 9 coins, namely $2 \times 1/2 = 1$ coin from its two end vertices, and then $4 \times  (1 + 1/2 + 1/2) = 8$ coins from each of the four vertices in $N_{\textnormal{end}}(P)$ and their respective neighbors in components of $M$ other than $P$.

\noindent
\textbf{\LBL{Case 1.2}:} There are two vertices $w_1,w_2\in N_{\textnormal{end}}(P)$ that are adjacent in $M$.

\noindent
\textbf{\LBL{Case 1.2.1}:} The vertices $w_1,w_2$ are adjacent to distinct vertices in $N_{\textnormal{end}}(P)$.\\
Then there is an 
\inward alternating path of length three considered by Lemma~\ref{lem:alternating}\ref{imp:cycle}.
\smallskip

\noindent
\textbf{\LBL{Case 1.2.2}:} Vertices $w_1,w_2$ are in component $C$ and adjacent to the same end vertex $v_i \in \mathsf{end}(P)$.\\
Then we can add $v_i$ to $C$, which makes the neighbor $u\in\{u_1,u_t\}$ of $v_i$ in $P$ an end vertex.
In particular, if there is an edge $e$ incident to $u$ and an end vertex such that $e$ is not in $P$, there is an improvement without creating new
end vertices, a contradiction.

We assign to $P$ the two coins of $w_1,w_2$ the two half coins of the two component neighbors of $w_1,w_2$, half a coin of $v_i$, and half a coin of the component neighbor of $v_i$ to $P$.
Thus, the two edges gain a total of~4 coins for $P$.
\smallskip

\noindent
This finishes \ref{Case 1.2}, because in either subcase 1.2.1 and 1.2.2 we can assign at least $8$ coins to~$P$ additionally to the one coin that $P$ had already.
\smallskip

\smallskip
\noindent
\textbf{\LBL{Case 2}:} $|\out(P)| = 3$.

\noindent
\textbf{\LBL{Case 2.1}:} The three vertices in $N_{\textnormal{end}}(P)$ form an independent set.\\
Then $P$ collects $3 \times 2 = 6$ coins from vertices outside $P$ analogous to \ref{Case 1}, and $2 \times 1/2$ coins from its end vertices.
Assume, without loss of generality, that $v_1\in \mathsf{end}(P)$ is the end vertex of $P$ for which there is an internal vertex $u_i$ in $P$ such that $\{v_1,u_i\}$ is an edge in $G_x$ for some $i\in\{2,\hdots,t\}$.
Hence, $P$ collects half a coin from $u_i$.
Thus, $P$ receives at least $7\frac{1}{2}$ coins, and needs to collect $1\frac{1}{2}$ more coins.

\noindent
\textbf{\LBL{Case 2.1.1}:} There is a connecting edge $e\in E(G_x)\setminus E(M)$ with $e = \{z,w\}$ for $z \in \{u_{i-1},u_{i+1}\}$ such that $w$ is an end vertex of a path $P'\not= P$.\\
In this case, $Q = (v_1,u_i,z,w)$ is an alternating path of length $3 < 5$, with end vertices $s = v_1$ and $t = w$ of different paths.
Thus, by Lemma~\ref{lem:alternating}\ref{imp:three}, there is an improvement of $M$, a contradiction.

\noindent
\textbf{\LBL{Case 2.1.2}:} There is a connecting edge $e\in E(G_x)\setminus E(M)$ with $e = \{u_{i-1},w\}$ such that~$w$ is an end vertex of a cycle.\\
In this case, $Q = (v_1,u_i,u_{i-1},w)$ is an inward alternating path of length $3 < 5$, with end vertices $s = v_1$ and $t = w$ of different paths.
Thus, by Lemma~\ref{lem:alternating}\ref{imp:cycle}, there is an improvement of $M$, a contradiction.

\noindent
\textbf{\LBL{Case 2.1.3}:} There is a connecting edge $e\in E(G_x)\setminus E(M)$ with $e = \{u_{i+1},w\}$ such that $w$ is an end vertex of a cycle and the
previous cases do not apply.

\noindent
\textbf{\LBL{Case 2.1.3.1}:} There is no connecting edge $e' \in E(G_x)\setminus E(M)$ with $e = \{u_{j},w'\}$ for $j<i$ such that $w'$ is an end vertex of some component.\\
Then none of the coins between $v_1$ and $u_i$ has been assigned.
We assign the remaining half coin of $v_i$, and one coin of $u_{i-1}$ to $P$.
Therefore total number of coins assigned to $P$ is $9$.

\noindent
\textbf{\LBL{Case 2.1.3.2}:} There is a connecting edge $e' \in E(G_x)\setminus E(M)$ with $e = \{u_{j},w'\}$ for some $j<i$ such that~$w'$ is an end vertex of some component.

\noindent
\textbf{\LBL{Case 2.1.3.2.1}:} For some choice of $j$ and $j'$, $j < i < j'$, such that $w'$ and $w''$ are end vertices adjacent to $u_j$ and $u_j'$, both $w',w''$ are not path-forming and belong to one cycle $C$.\\
Then, by Observation~\ref{appobs:assignment}, at most $3$ coins have been sent to $C$ due to $w'$ and $w''$ and therefore the component neighbors of $u_j$ and $u_{j'}$ still have together at least two half coins that have not yet been assigned.
We assign these two half coins to $P$ and additionally half a coin of $u_{i-1}$.
Again, the total number of coins assigned to $P$ is $9$.

\noindent
\textbf{\LBL{Case 2.1.3.2.2}:} Otherwise, there is an improvement by including the edges $\{v_1,u_i\}$, $\{u_j,w'\}$ and $\{u_{i+1},w\}$ and removing up to four edges.\\
Notice that for each removal we can choose an edge such that either we do not create a new end vertex or the removed edge is not a $1$-edge.
Therefore we can keep up the invariant of Lemma~\ref{lem:enddegree}.

\noindent
\textbf{\LBL{Case 2.1.4}:} No such connecting edge $e$ exists.\\
Then $P$ also receives $2 \times 1/2 = 1$ coin from the two neighbors of $u_i$ and the whole coin of $u_i$, and hence receives $9$ coins in total.
\smallskip

\noindent
\textbf{\LBL{Case 2.2}:} There are two vertices $w_1,w_2\in N_{\textnormal{end}}(P)$ that are adjacent in $M$.

\noindent
\textbf{\LBL{Case 2.2.1}:} Vertices $w_1,w_2$ are adjacent to distinct vertices in $N_{\textnormal{end}}(P)$.\\
Then there is an improvement by removing the edge $\{w_1,w_2\}$ and adding the two edges from~$P$ to~$w_1$ and $w_2$.
\smallskip

\noindent
\textbf{\LBL{Case 2.2.2}:} Vertices $w_1,w_2$ are in component $C$ and adjacent to the same end vertex $v_i \in \mathsf{end}(P)$.\\
Then we can add $v_i$ to $C$, which makes the neighbor $u\in\{u_1,u_t\}$ of $v_i$ in $P$ an end vertex.
In particular, if there is an edge $e$ incident to $u$ and an end vertex such that $e$ is not in $P$, there is an improvement, a contradiction.

We assign to $P$ the two coins of $w_1,w_2$, the two half coins of the two component neighbors of $w_1,w_2$, half a coin of $v_i$, and half a coin of the component neighbor of $v_i$ to $P$.
Thus, the two edges gain a total of~4 coins for $P$.
Additionally, $P$ collects the remaining $5$ coins analogously to~\ref{Case 2.1}, where a collected coin from $v_i$ implies that there is an improvement. This finishes \ref{Case 2.2}, because $P$ can collect $9$ coins.
\smallskip

\smallskip
\noindent
\textbf{\LBL{Case 3}:} $|\out(P)| = 2$.\\
Then there are two edges $e_1,e_2\in E(G_x)\setminus M$ that each have one endpoint in $\mathsf{end}(P)$ and its other endpoint in $\mathsf{int}(P)$, and their endpoints in $\mathsf{int}(P)$ are distinct.
We analyze the possibilities how these edges may interfere with each other.
Let $e_1 = \{v,u_i\}$ and $e_2 = \{v',u_j\}$, where $v,v'\in\mathsf{end}(P)$ with possibly $v = v'$, and $u_i,u_j\in\mathsf{int}(P)$ with $u_i \not= u_j$.

Analogous to \ref{Case 1}, since $|\out(P)| = 2$, the path $P$ already collects $2 \times 2 = 4$ coins from vertices outside $P$ and $2 \times 1/2 = 1$ coin from its two end vertices.
\smallskip

\noindent
\textbf{\LBL{Case 3.1}:} Vertex $v = v'$.\\
By renaming we may assume that $v=v_1$.

\noindent
\textbf{\LBL{Case 3.1.1}:} Vertices $u_i,u_j$ are adjacent in $P$.\\
We say that a vertex $w$ in $P$ \emph{acts as an end vertex} if there is a path $P'$ in $G_x$ with $V(P)=V(P')$ such that $w$ is an end vertex of $P'$. We only consider simple transformations that can be done efficiently.
Then the neighbor $u_1$ of $v$ in $P$ acts as an end vertex.
Let $v$ be closer to $u_i$ than to~$u_j$ in $P$ (that is, $i < j$).
Then we assign to $P$ the full coin of $v$ and $u_i$, half a coin of $u_j$, half a coin of the neighbor $u_1$ of $v$ in $P$, and half a coin of $u_{i-1}$. All of these coins can be assigned unless there is an \inward alternating path of length three from $v_1$ via $u_i$ and $u_{i-1}$.
There is one more coin that we have to assign.
That is, if a coin that acts as an end vertex is adjacent to an end vertex of another component, there is an improvement without creating new end vertices.
\smallskip

\noindent
\textbf{\LBL{Case 3.1.1.1}:} The subpath of $P$ from $v$ to $u_j$ has at least five vertices.\\
Then the analysis is analogous to \ref{Case 2}, that is, unless there is an improvement we either we can assign sufficiently many coins of the sub-path from $v_1$ to $u_i$ to $P$ or we can assign half a coin of both $u_j$ and~$u_{j+1}$ to $P$.
\smallskip

\noindent
\textbf{\LBL{Case 3.1.1.2}:} The subpath of $P$ from $v$ to $u_j$ has exactly four vertices.\\
Then all vertices of $P$ but $u_j$ act as end vertices, and by the LP constraints, there has to be at least one more edge leaving the subpath.
If the destination of that edge is an end vertex, we have found an 
\inward alternating path of length at most three, which is excluded by Lemma~\ref{lem:alternating}\ref{imp:cycle}.
Otherwise, that vertex still has its coin an we transfer the coin to $P$.
\smallskip

\noindent
\textbf{\LBL{Case 3.1.2}:} The vertices $u_i,u_j$ are not adjacent in $P$.\\
The analysis is analogous to \ref{Case 2}, but the argument has to be applied for $u_i$ and $u_j$ separately.
\smallskip

\noindent This finishes \ref{Case 3.1} since $P$ was able to collect $4$ coins additionally to the $5$ initially collected coins.
\smallskip

\noindent
\textbf{\LBL{Case 3.2}:} Vertex $v \neq v'$.\\
\smallskip
Let $u_i,u_j$ be internal vertices of $P$ such that $\{v_1,u_i\},\{v_2,u_j\}$ are edges of $E(G_x)\setminus E(M)$.

\noindent
\textbf{\LBL{Case 3.2.1}:} $u_i$ and $u_j$ are not consecutive.\\
This case is analogous to \ref{Case 2} applied to $u_i$ and $u_j$ independently.
\smallskip

\noindent
\textbf{\LBL{Case 3.2.2}:} $u_i$ and $u_j$ are consecutive.\\
If $i>j$, there is an inward alternating path of length three from $v_1$ to $v_2$, considered in Lemma~\ref{lem:alternating}\ref{imp:cycle}.
Therefore we may assume $i < j$.
\smallskip

\noindent
\textbf{\LBL{Case 3.2.2.1}:} There is no $i' < i$ such that there is an edge $\{u_{i'},w\}$ for an end vertex $w$ of a component $C \neq P$.
Then we assign all coins of the vertices $v_1,u_1,\dotsc,u_j$ to $P$ and half a coin from~$u_{j+1}$.
Clearly, $j\ge 3$ and none of the assigned coins has been assigned previously.
Together with the $4\frac12$ coins already assigned to $P$, the total number of coins adds up to at least $9$.
\smallskip

\noindent
\textbf{\LBL{Case 3.2.2.2}:} There is no $j' > j$ such that there is an edge $\{u_{j'},w\}$ for an end vertex $w$ of a component $C \neq P$.\\
This case is analogous to \ref{Case 3.2.2.1}.
\smallskip

\noindent
\textbf{\LBL{Case 3.2.2.3}:} There is an $i' < i$ and an edge $\{u_{i'},w\}$ for an end vertex $w$ of a component $C \neq P$, and there is a $j' > j$ and an edge $\{u_{j'},w'\}$ for an end vertex $w'$ of a component $C' \neq P$.\\
This case is analogous to \ref{Case 2.1.3.2}.
\smallskip

\noindent
\textbf{\LBL{Case 4}:} $|\out(P)| = 1$.\\
Then $P$ collects $2 \times 1/2 = 1$ coin from its two end vertices, plus 2 coins from vertices outside~$P$.
There are three edges $e_1 = \{v_1,u_i\},e_2 = \{v_1,u_j\},e_3 = \{v_2,u_p\}\in E(G_x)\setminus E(M)$ that each have one endpoint in $\mathsf{end}(P)$ and their endpoints in $\mathsf{int}(P)$ are distinct.
Hence, $P$ collects $3 \times 1/2 = 3/2$ coins from $u_i,u_j,u_p$.
Thus, $P$ needs to collect an additional $4\frac{1}{2}$ coins.

We analyze the possibilities how these edges may interfere with each other.
Assume, without loss of generality, that $i < j$. Notice that \ref{Case 4} is very similar to \ref{Case 3}. In particular, we
use \ref{Case 3.1} in order to assign the coins related to $u_i$ and $u_j$.
There are some complications that we address in the following cases.

Notice that any edge $\{u_{j'},v_2\}$ for $j' < j$ behaves the same way as an edge $\{u_{j'},w\}$ for
an end vertex~$w$ of a component $C \neq P$ with respect to $u_j$.
The analogous statements are true for $u_i$ and~$u_p$.
Therefore the only cases that are not analogous to \ref{Case 1} or \ref{Case 2} are the following.

\noindent
\textbf{\LBL{Case 4.1}:} $u_{j+1} = u_p$.\\
We claim that the coins of $u_j$ and $u_{j+1}$ have not been assigned twice.
Notice that assigning the whole coin of $u_j$ to $P$ while considering $v_1,u_i,u_j$ implies that there is a $j'<j$ such that $j'$ is adjacent to an
end vertex of another component $C$.
Then the remaining argument is analogous to \ref{Case 2.1.3.2}.
\smallskip

\noindent
\textbf{\LBL{Case 4.2}:} $u_{i+1} = u_p$.\\
This case is analogous to \ref{Case 4.1}.
\smallskip

\noindent
\textbf{\LBL{Case 5}:} $\out(P) = \emptyset$.\\
Then $P$ collects $2 \times 1/2 = 1$ coin from its two end vertices.
Then there are four edges $e_1 = \{v_1,u_i\},e_2 = \{v_1,u_j\},e_3 = \{v_2,u_p\},e_4 = \{v_2,u_q\}\in E(G_x)\setminus E(M)$ that each have one endpoint in $\mathsf{end}(P)$ and their endpoints in $\mathsf{int}(P)$ are distinct.
Hence, $P$ collects $4 \times 1/2 = 2$ coins from $u_i,u_j,u_h,u_q$.
Thus, $P$ needs to collect an additional 6 coins.

We analyze the possibilities how these edges may interfere with each other.
Assume, without loss of generality, that $i < j$, and (by symmetry) that $p > q$.
Notice that \ref{Case 5} is very similar to \ref{Case 3}.
In particular, we use \ref{Case 3.1} in order to assign the coins related to $u_i$ and $u_j$, as well as to assign the coins related to $u_p,u_q$.
There are some complications that we address in the following cases.

Notice that any edge $\{u_{j'},v_2\}$ for $j' < j$ behaves the same way as an edge $\{u_{j'},w\}$ for an end vertex~$w$ of a component $C \neq P$ with respect to $u_j$.
The analogous statements are true for $u_i$, $u_p$, and $u_q$.
Therefore the only cases that are not analogous to \ref{Case 1} or \ref{Case 2} are the following.

\noindent
\textbf{\LBL{Case 5.1}:} $u_{j+1} \in \{u_p,u_q\}$.\\
We claim that the coins of $u_j$ and $u_{j+1}$ have not been assigned twice.
Notice that assigning the whole coin of $u_j$ to $P$ while considering $v_1,u_i,u_j$ implies that there is a $j'<j$ such that $j'$ is adjacent to an
end vertex of another component $C$.
The analogous statement is true for the coin of $u_q$ and therefore there is a $q'>q$ adjacent to an end vertex of a component $C' \neq P$.
Then the remaining argument is analogous to \ref{Case 2.1.3.2}.
\smallskip

\noindent
\textbf{\LBL{Case 5.2}:} $u_{i+1} \in \{u_p,u_q\}$.\\
This case is analogous to \ref{Case 5.1}.

\qed

\section{\texorpdfstring{Asymmetric $(1,2)$-TSP}{Asymmetric (1,2)TSP}}
\label{app:atsp}
In this section, we consider linear programming-based approximations for \ATSP.
\begin{theorem}
\label{appthm:atsp}
  There is a polynomial-time $3/2$ approximation algorithm for \ATSP with respect to $\OptSERup(G)$.
\end{theorem}
\begin{proof}
  Let $x^*$ be an optimal solution to $\SERup(G)$ and let $x,z$ be as in the proof of Theorem~\ref{thm:stsp}.
  Again, we subdivide the arcs into sub-arcs with respect to $z$ just as we did with the edges in the proof of Theorem~\ref{thm:stsp}.
  Let $M$ be a directed 2-matching of $G_{x^*}$ such that that all 1-arcs are in~$M$.
  We use directed alternating paths as defined in the proof of Lemma~\ref{lem:nosingle}.

  We use an algorithm similar to Algorithm~\ref{alg:stsp}.
  However, instead of Lemma~\ref{lem:alternating} and Lemma~\ref{lem:alternatingclosed} we only use the simple observation that if there is an arc from some end vertex to some start vertex, there is an improvement of $M$, unless both vertices are in one cycle.
  Note that in $M$, there may be cycles of length two.

  We construct a set of alternating paths $\apset_3$.
  Write $\delta$ as short for $\delta_{G_{x^*}}$.
  Initially, all sub-arcs in cycles are marked and all remaining sub-arcs are unmarked.
  Then we proceed as follows:
  \begin{enumerate}
    \item Choose an end vertex $t$ with an unmarked sub-arc $a$ of $x$ with $a \in \delta^+(t)$.
    \item Extend $a$ to a directed alternating path $Q$ of length three, using two further unmarked sub-arcs of $x$.
    \item Mark the three sub-arcs of $Q$.
  \end{enumerate}
  Similar to the proof of Theorem~\ref{thm:stsp}, we obtain $\apset_3$ by applying the procedure iteratively until all sub-arcs of $x$ that start from end vertices are marked.

  Let us first verify that each of the alternating paths can be extended to length three.
  There are two reasons that could possibly prevent an extension: there is no sub-arc that is part of a feasible alternating path or all sub-arcs that are feasible candidates are marked already.

  Let $a=(t,u)$ be the first sub-arc of some alternating path.
  Since adding $a$ to $M$ does not lead to an improvement, $u$ is not a start-vertex.
  We obtain
  \begin{equation*}
    z(\delta^-(u) \setminus E(M)) = x^*(\delta^-(u) \setminus E(M)) = 1 - x^*(\delta^-(u) \cap E(M)) = z(\delta^-(u) \cap E(M)) \enspace .
  \end{equation*}
  In other words, from $u$ one can extend as many alternating path as there are sub-arcs of $x$ leading to~$u$.
  Let $a'=(v,u)$ be the subsequent arc in $Q$.
  Then
  \begin{equation*}
    z(\delta^+(v) \cap E(M)) = 1 - x^*(\delta^+(v) \cap E(M)) = x^*(\delta^+(v) \setminus E(M)) = z(\delta^+(v) \setminus E(M)) \enspace .
  \end{equation*}
  Thus there are sufficiently many sub-arcs of $x$ from $v$ to extend $Q$.

  Note that in $\apset_3$, each arcs leaving end vertices either is arcs within a cycle or it is the first arc of an alternating path of length three.

  Let \LPX be analogous to Theorem~\ref{thm:stsp}.
  For any alternating path starting at some component (cycle or arc), let $a_1,a_2,a_3$ be the three sub-arcs that form $Q$.
  We set $y_{C,y_1} = y_{C,y_3}=1$. Since all considered alternating paths are disjoint, clearly we did not violate any constraint from \eqref{appcon:double}.

  Since for each end vertex $t$ of a path, $x^*(\delta^+(t))=1$ and for each cycle $C$, $x^*(\delta^+(V(C))) \ge 1$, we have assigned at least $2N$ sub-arcs to each component, which is sufficient to satisfy the constraints~\eqref{appcon:alpha} for $\alpha=2$. 
\end{proof}

\begin{corollary}
\label{appcor:halfintatsp}
  If we can obtain a half-integral solution to $\SERup(G)$ in polynomial time, there is a polynomial~time $4/3$ approximation algorithm for \ATSP with respect to $\OptSERup(G)$.
\end{corollary}
\begin{proof}
  We use the same proof as for Theorem~\ref{appthm:atsp}, except the following.
  For any alternating path starting at some component (cycle or arc), let $a_1,a_2,a_3$ be the three sub-arcs that form $Q$.
  We set $y_{C,y_1} = y_{C,y_2} = y_{C,y_3}=1$.

  This is justified since as in the symmetric case, for half-integral instances $x^* = z$.
\end{proof}

\section{Implications for \SER}
\label{sec:SER}

Optimal fractional \SERup solutions may have a cost that is higher by a value $\varepsilon < 1$ than optimal fractional \SER solutions.
Therefore, lower bounds on the integrality gap for \SERup are also lower bounds for \SER.
Here, we want to analyze the other direction: we want to deduct integrality gap upper bounds for \SER from \SERup upper bounds.
In absolute terms, the difference between the optimal fractional costs of the two relaxations is at most one since the only impact of the introduced hyperplane is to round up the objective value to the next integer.
This insight implies that the \SER integrality gap and the \SERup integrality gap converge to a single value with increasing instance size.
In particular, the convergence provides a purely mechanical method to obtain \SER integrality gap results from \SERup integrality gap upper bounds.
For some value $\varepsilon$, suppose we know an integrality gap upper bound for \SERup that lies between the known upper and lower bound for \SER such that the two upper bounds are a factor $(1+\varepsilon)$ apart.
Then there is a $c \in O_\varepsilon(1)$ such that we simply have to check the \SER integrality gaps of all instances of order at most $c$:
the value $c$ is chosen in such a way that we have either found an instance that provides a new integrality gap lower bound or we have shown that the integrality gap upper bound of \SER smaller than $1+\varepsilon$ times the one of \SERup.

It is known that for $c \le 12$, the \SER integrality gap for \STSP is at most $10/9$~\cite{QSWvZ15}.
In the following we will show how to amplify the convergence such that we obtain new results already for the known values $c$ (see Fig.~\ref{fig:comparison}).
To this end, we will show two claims that loosely speaking state the following.
\begin{enumerate}
\item[(i)] Given the support graph of a \SER integrality gap instance for \STSP, we can introduce a degree-two vertex without shrinking the integrality gap by too much (depending on $c$). 
\item[(ii)] Given an \STSP instance with integrality gap $\alpha$ for \SER such that the support graph has a degree two vertex, we can construct an \STSP instance with integrality gap at least $\alpha$ for \SERup.
\end{enumerate}
\begin{figure}
\begin{center}
\begin{tabular}{|l|rl|l|}\hline
                    & \SERup&                                               & \SER\\\hline
General             & $5/4$ &$= 1.25$ (Theorem~\ref{thm:stsp})              & $1.2693$\\
                    & $26/21$&$\approx 1.239$ (\cite{QSWvZ15})   & $1.257$\\
1/2-Integral        & $7/6 $&$\approx 1.166$ (Theorem~\ref{thm:halfint})    & $1.179$\\\hline
\end{tabular}
\end{center}
\caption{\label{fig:comparison}Comparison of integrality gap upper for \SERup and \SER based on Theorem~\ref{thm:use-computation} and computational results for instances of order at most $12$.}
\end{figure}

The first claim is based on the following insight.
\begin{lemma}
\label{lem:split_edge}
Let $G$ be a \STSP instance of order $n$. Then there is an instance $G'$ such that 
\begin{itemize}
\item $\Opt(G') \ge \Opt(G) + 1$,
\item $\OptSER(G') = \OptSER(G) + 1$ and
\item we obtain an optimal basic solution $x'$ to $\SER(G')$ that has a degree two vertex in $G'_{x'}$.
\end{itemize}
Suppose we are given an optimal half-integral basic solution to $G$, then we can ensure that also $x'$ is half-integral.
\end{lemma}
\begin{proof}
We first compute an optimal basic solution $x^*$ of $\SER(G)$.
Now we obtain the graph $G''$ from $G$ as follows.
We set $V(G'') = V(G)$ and form a complete graph on the vertex set.
For each edge $e \in E(G_{x^*}$, the cost of $e$ in $G''$ is set to the cost in $G$.
For each edge $e \in E(G) \setminus E(G_{x^*})$, the cost of $e$ in $G''$ is two.
Clearly, $x^*$ is a basic solution to $\SER(G'')$ since we only changed the cost vector in the objective function but not the polytope itself. 
Since the edge costs within the support graph did not change and we did not decrease edge costs of the remaining edges, $x^*$ is optimal for $G''$.
For the same reason, also the cost of an optimal integral solution and therefore also the \SER integrality gap cannot be decreased in $G''$.

Boyd and Pulleyblank~\cite{BoydPulleyblank1990} showed that any basic solution to \SER has a $1$-edge.
Therefore we can find a $1$-edge $e'=\{u,w\} \in E(G''_{x^*})$.
By Proposition~\ref{pro:aux}, we may assume that that $\cost(e')=1$.

We obtain the aimed-for graph $G'$ from $G''$ by subdividing $e'$, \ie, we introduce a new vertex $v$ to~ $G'$, set $c(\{u,v\}) = c(\{v,w\}) = 1$ and $c(\{u,w\})=2$.
For all $v' \notin \{u,v,w\}$, we set $c(\{v,v'\})=2$.

To show the first claim, suppose by contradiction that there was an integral solution in $G'$ of cost $\Opt(G)$.
Then let $u',u''$ be the two vertices adjacent to $v$ in that solution.
We simply remove $v$ and add the edge $\{u',u''\}$ to obtain an improved solution to $G$.
If $\{u',u''\} = \{u,w\}$, $\cost(\{u',u''\})=1$ and $\cost(\{u',v\}) + \cost(\{v,u''\}) = 2$.
Otherwise, $\cost(\{u',u''\})=2$ and $\cost(\{u',v\}) + \cost(\{v,u''\}) \ge 3$.
In both cases, we obtain an integral solution to $G$ of cost smaller than $\Opt(G)$, a contradiction.

Let $x'$ be the vector obtained from $x^*$ by setting $x'_{e'}=0$, $x'_{\{u,v\}} = x'_{\{v,w\}} = 1$, and $x'_e = x^*_e$ for each of the remaining edges $e$.
We claim that $x'$ is an feasible basic solution for $\SER(G')$.
We have that $x'$ is feasible since $x'(\delta(u)) = x'(\delta(v)) = x'(\delta(w))=2$ and each cut $S$ between $u$ and $w$ still contains a $1$-edge, just as in $x^*$.
Since $\SER$ implicitly enforces that $x'_e \le 1$ for each edge $e$ and since $x^*$ is a basic solution, also $x'$ is a basic solution. 
Furthermore, since $x^*$ is optimal, $x'$ cannot be improved by moving to an adjacent solution in the polytope and thus it is optimal.

Clearly, $v$ is a degree-two vertex in $G'_{x'}$ and if $x^*$ is half-integral, also $x'$ is.
\end{proof}

Therefore the integrality gap $\alpha'$ of $G'$ is at least $(\Opt(G) + 1)/(\OptSER(G) + )$.
By simple arithmetics we obtain that 
$\alpha' \ge \alpha - (\alpha - 1)/(\OptSER(G) + 1)$, where $\alpha = \Opt(G)/\OptSER(G)$ is the integrality gap of $G$.
Conversely, $\alpha \le \alpha' + (\alpha' - 1)/\OptSER(G)$.

Claim (ii) follows by generalizing a result of Williamson~\cite{Williamson1990}.
\begin{lemma}
\label{lem:extend}
  Let $G$ be an instance of \STSP and let $x^*$ be an optimal solution to $\SER(G)$ such that $\Opt(G)/\OptSER(G) = \alpha$ and there is a degree-two vertex in the support $G_{x^*}$.
  Then, for any $\gamma \in \mathbb{N}$, there is an instance~$G'$ of order $|V(G')| = \gamma \cdot |V(G)|$ with $\Opt(G')/\OptSER(G') \ge \alpha$.
\end{lemma}
\begin{proof}
  Let $n$ be the order of $G$.
  If $n \le 2$, then $\alpha=1$ and the claim is trivially true.

  Henceforth, let $n > 2$.
  We show how to obtain an instance $G'=(V',E')$ of order $2n$ and an optimal solution $x'$ to $\SER(G')$ such that the support of $x'$ has a vertex~$v$ with only two adjacent vertices and $\Opt(G')/\OptSER(G') = \alpha$.
  Since we can repeat the construction arbitrarily often, the claim follows inductively.

  Without loss of generality, we assume that all edges not in the support of~$x^*$ are of cost two, since the assumption does not change $\cost(x^*)$ and cannot decrease $\Opt(G)$.
  Let $s,t$ be the two vertices adjacent to $v$.
  Then, due to the equality constraints, $x_{\{v,s\}} = x_{\{v,t\}} = 1$ and $s \neq t$.
  Furthermore, we can assume $\cost{\{v,s\}} = \cost{\{v,t\}} = 1$.
  Otherwise, suppose that $\cost{\{v,s\}} = 2$.
  Decreasing the cost to one also decreases $\cost(x^*)$ by one and $\OptSER(G)$ by at most one.
  Hence, the integrality gap cannot decrease.
  The same holds for the edge~$\{v,t\}$.

  For $i \in \{1,2\}$, let $G_i=(V_i,E_i)$ be a copy of $G$ and let $v_i$, $s_i$, and $t_i$ be the copies of $v$, $s$, and $t$.
  We set $V' = V_1 \cup V_2$.
  We first set the edge costs in $E'$ to those of $E_1$ and $E_2$.
  For the remaining edges $ e \in \{\{u,w\} : u \in V_1, w \in V_2\}$, we set $\cost(e)=2$.
  Afterwards, we change the cost of four edges as follows (see also Fig.~\ref{fig:intgap}):
  \begin{equation*}
    \begin{split}
      \cost(\{v_1,s_2\}) & =  \cost(\{v_2,s_1\})  = 1,\\
      \cost(\{v_1,s_1\}) & =  \cost(\{v_2,s_2\})  = 2 \enspace .\\
    \end{split}
  \end{equation*}
  Correspondingly, we set
  \begin{equation*}
    \begin{split}
      x'_{\{v_1,s_2\}} & =  x'_{\{v_2,s_1\}}  = 1,\\
      x'_{\{v_1,s_1\}} & =  x'_{\{v_2,s_2\}}  = 0 \enspace .\\
    \end{split}
  \end{equation*}
  We set $x'$ for all remaining edges of $G_1$ and $G_2$ according to $x^*$.
  For all edges not yet considered, we set~$x'$ to zero.
  Clearly, $\cost(x') = 2 \cost(x^*)$, and the support of $x'$ has some vertex of degree~2.

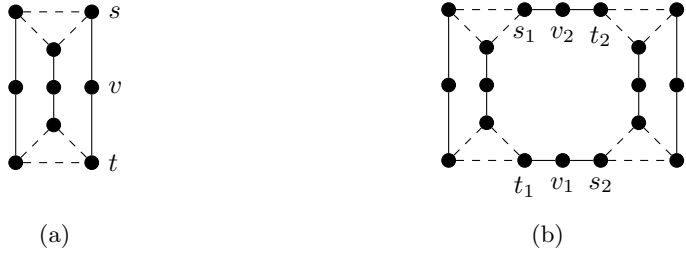
\begin{figure}[htb]
  \centering
  \begin{subfigure}[b]{0.4\textwidth}
    \centering
    \begin{tikzpicture}[scale=0.5]
      \node (1) at (0,0) [vertex, label=below:\phantom{M}]{};
      \node (2) at (2,0) [vertex, label=right:$t$]{};
      \node (3) at (1,1) [vertex]{};
      \node (4) at (0,2) [vertex]{};
      \node (5) at (1,2) [vertex]{};
      \node (6) at (2,2) [vertex, label=right:$v$]{};
      \node (7) at (1,3) [vertex]{};
      \node (8) at (0,4) [vertex]{};
      \node (9) at (2,4) [vertex, label=right:$s$]{};
      \draw[dashed](1)--(2);
      \draw[dashed](1)--(3);
      \draw[](1)--(4);
      \draw[dashed](2)--(3);
      \draw[](2)--(6);
      \draw[](3)--(5);
      \draw[](4)--(8);
      \draw[](5)--(7);
      \draw[](6)--(9);
      \draw[dashed](7)--(8);
      \draw[dashed](7)--(9);
      \draw[dashed](8)--(9);
    \end{tikzpicture}
    \caption{$~$}
    \label{fig:gap_undirected}
  \end{subfigure}
  \begin{subfigure}[b]{0.4\textwidth}
    \centering
    \begin{tikzpicture}[scale=0.5]
      \node (1) at (0,0) [vertex, label=below:\phantom{M}]{};
      \node (2) at (2,0) [vertex, label=below:$t_1$]{};
      \node (3) at (1,1) [vertex]{};
      \node (4) at (0,2) [vertex]{};
      \node (5) at (1,2) [vertex]{};
      \node (6) at (3,0) [vertex, label=below:$v_1$]{};
      \node (7) at (1,3) [vertex]{};
      \node (8) at (0,4) [vertex]{};
      \node (9) at (2,4) [vertex, label=below:$s_1$]{};
      \node (10) at (4,0) [vertex, label=below:$s_2$]{};
      \node (20) at (6,0) [vertex]{};
      \node (30) at (5,1) [vertex]{};
      \node (40) at (3,4) [vertex, label=below:$v_2$]{};
      \node (50) at (5,2) [vertex]{};
      \node (60) at (6,2) [vertex]{};
      \node (70) at (5,3) [vertex]{};
      \node (80) at (4,4) [vertex, label=below:$t_2$]{};
      \node (90) at (6,4) [vertex]{};
      \draw[dashed](1)--(2);
      \draw[dashed](1)--(3);
      \draw[](1)--(4);
      \draw[dashed](2)--(3);
      \draw[](2)--(6);
      \draw[](3)--(5);
      \draw[](4)--(8);
      \draw[](5)--(7);
      \draw[](6)--(10);
      \draw[dashed](7)--(8);
      \draw[dashed](7)--(9);
      \draw[dashed](8)--(9);
      \draw[dashed](10)--(20);
      \draw[dashed](10)--(30);
      \draw[](9)--(40);
      \draw[dashed](20)--(30);
      \draw[](20)--(60);
      \draw[](30)--(50);
      \draw[](40)--(80);
      \draw[](50)--(70);
      \draw[](60)--(90);
      \draw[dashed](70)--(80);
      \draw[dashed](70)--(90);
      \draw[dashed](80)--(90);
    \end{tikzpicture}
    \caption{$~$}
    \label{fig:gap_undirected2}
  \end{subfigure}
  \caption{Integrality gap lower bound instances for \STSP, where the edges are drawn if
  and only if they are of cost one. There is a solution
  $x$ to \SER such that $x_e=1/2$ for all dashed edges $e$ and $x_e=1$
  otherwise.}
  \label{fig:intgap}
\end{figure}

  Let us now verify that $x'$ is a feasible solution to $\SER(G')$.
  For the equality constraints it suffices to consider $v_1$, $v_2$, $s_1$, and $s_2$, because the edges incident to all other vertices did not change.   
  Since we simply replaced edges by others with identical value, all equality constraints are satisfied.

  For all $S \subseteq V_1$ and $S \subseteq V_2$, the subtour elimination constraints are satisfied, where we use that any of these cuts containing $\{s_i,v_i\}$ also contains $\{s_i,v_{3-i}\}$, $i \in \{1,2\}$.
  Fix a set~$S$ not yet considered.
  If either both $v_1$ and $v_2$ or none of them is in $S$, $\delta_{G'}(S)$ contains all edges of a cut in either $G_1$ or~$G_2$ and thus we are done.
  Thus, since $\delta(S) = \delta(V' \setminus S)$, we may assume that $v_1 \in S$ and $v_2 \notin S$.
  Let $S_1 = S \cap V_1$ and $S_2 = S \cap V_2$.
  We obtain that $x'(\delta(S)) = x'(\delta(S_1)) + x'(\delta(S_2)) - 2 x'_{\{v_1,s_2\}} \ge 2 + 2 - 2 = 2$.

  Finally, we have to argue that $\Opt(G') \ge 2\Opt(G)$.
  For the sake of contradiction, assume that $\Opt(G') < 2\Opt(G)$.
  We derive a contradiction by constructing solutions within~$G_1$ and~$G_2$ such that the smaller one has a cost of at most $\Opt(G')/2$.
  Fix an optimal solution~$C'$ for $G'$.
  Let $k$ be the number of edges $e=\{u,w\}$ such that $\cost(e)=2$, $u \in V_1$, and $w \in V_2$.
  We remove all of these edges and are left with $k$ paths if $k \ge 1$ or one cycle if $k=0$.
  We replace $\{v_1,s_2\}$ by $\{v_1,s_1\}$, and replace $\{v_2,s_1\}$ by $\{v_2,s_2\}$, if these are contained in $C$.
  As a result, there are no edges between $G_1$ and $G_2$ left and all vertices have a degree of at most two.
  Note that within each of the two graphs $G_1$ and $G_2$, there is either a single cycle or a collection of paths (but not both). There are exactly $2k$ vertices of degree one.
  Therefore, we can introduce $k$ edges of cost at most~$2k$ in order to form Hamiltonian cycles $C_1$ and $C_2$ in~$G_1$ and $G_2$.
  Let us rename the graphs such that $\cost(C_1) \le \cost(C_2)$.
  By dropping the indices, we use $C_1$ to form a tour~$C$ in~$G$.
  Since $\cost(\{v,s\})=1$ in $G$, $\cost(C) \le \Opt(G')/2 < \Opt(G)$, a contradiction.
\end{proof}

Let $G'$ be the graph obtained from an instance $G$ by applying Lemma~\ref{lem:split_edge}.
Since any basic solution to $\SER$ or $\SERup$ is a vector of rational values, we can apply Lemma~\ref{lem:extend} to amplify any $G'$ as to obtain a graph $G''$ where
$\SER(G'') = \SERup(G'')$ and $\Opt(G')/\OptSER(G') = \Opt(G'')/\OptSER(G'')$.

To summarize, the main result of this section can be stated as follows.
\begin{theorem}
\label{thm:use-computation}
Let $\mathcal{G}$ be a class of instances such that for each $G \in \mathcal{G}$, either $\OptSER(G) = \lfloor(\OptSER(G))\rfloor$
or $\OptSER(G) - \lfloor(\OptSER(G))\rfloor \ge \gamma$.
Suppose that for any $G \in \mathcal{G}$, $\Opt(G)/\OptSERup(G) \le \alpha'$.
Then, for any $c \in \mathbb{N}$, either the \SER integrality gap of any instance in $\mathcal{G}$ is at most
$\beta := \alpha' + (\alpha' - 1)/(c+\gamma)$
or there is an instance $G \in \mathcal{G}$ such that $\OptSER(G) \le c$ and $\Opt(G)/\OptSERup(G) \ge \beta$.
\end{theorem}
The values in Fig.~\ref{fig:comparison} are obtained from setting $c=13$ and, for the class of instances with half-integral support, $\gamma = 1/2$.

\section{A Note on $\boldsymbol{\ATSP}$ Integrality Gap Instances}

We did not attempt to show similar properties of \SER for \ATSP, since there are no known bounds for small instances. 
We have, however, a similar amplification technique to obtain arbitrarily large integrality gap instances given that the small instance at hand has certain properties.

We have seen, in Theorem~\ref{lem:simple_intgap}, that the instance of Fig.~\ref{fig:gap_directed} provides an intagrality gap of $6/5$.
Let us now verify that also the instance of Fig.~\ref{fig:gap_directed2} provides an integrality gap lower bound of $6/5$.
Clearly, by checking all cuts we can observe easily that assigning $x_e=1/2$ to each arc yields a feasible solution to \SER of cost $10$.

To show a lower bound on the size of an optimal integral solution, we analyze the number of cost two arcs.
If there are at least two arcs of cost two, the overall cost of the solution is at least~$12$ and we are done.
By the argument for the previous instance, there is at least one arc of cost two.
By means of contradiction, let us assume that there is an integral solution with exactly one arc $a$ of cost two.

By symmetry, again we can assume without loss of generality that $(v_1,t_1)$ is an arc of the integral solution.
Then the tail of $a$ is in the left 4-cycle, since otherwise there is a vertex in the left 4-cycle that cannot be collected without using an arc of cost two.
The head of $a$ cannot be $v_1$ since otherwise we created a short cycle (or used another arc of cost two).
Thus, the arc~$(s_2,v_1)$ must belong to the solution.
By an analogue argument as above, the head of $a$ belongs to the right 4-cycle.
However, to collect $v_2$ the tour has to move back to the left 4-cycle and can only reach $v_1$ by using another arc of cost $2$.
Therefore, indeed any integral solution uses at least two arcs of cost two.

We now show that the construction can be extended to arbitrarily large size.
\begin{lemma}
\label{applem:aextend}
  Let $G$ be an instance of \ATSP with an optimal solution $x^*$ to $\SER(G)$ such that $\Opt(G)/\OptSER(G) = \alpha$ and a vertex $v$ with exactly two adjacent vertices in the support~$G_{x^*}$.
  Suppose that any two paths $P_1,P_2$ with $V(P_1) \cap V(P_2) = \{v\}$ and $V(P_1) \cup V(P_2) = V(G)$ that either both start in~$v$ or both end in $v$ satisfy $\cost(P_1) + \cost(P_2) \ge \Opt(G)-2$.
  Then there are infinitely many instances~$G'$ of \ATSP with $\Opt(G')/\OptSER(G') \ge \alpha$.
\end{lemma}
\begin{proof}
  We show how to obtain an instance $G'=(V',E')$ of order $2n$ and a solution~$x'$ to $\SER(G')$ such that the support graph $G'_{x'}$ again satisfies the conditions of the lemma and $\Opt(G')/\OptSER(G') \ge \alpha$ (see also Fig.~\ref{appfig:gap_join}).
  The idea is similar to the proof of Lemma~\ref{lem:extend}.
  Again, we create two copies~$G_1,G_2$ of~$G$. 

  Without loss of generality, we assume that all arcs not in the support of~$x^*$ are of cost two, since the assumption does not change $\cost(x^*)$ and cannot decrease $\Opt(G)$.
  Let $s$ and $t$ be the two vertices adjacent to $v$.

  For $i \in \{1,2\}$, let $G_i=(V_i,E_i)$ be a copy of $G$ and let $v_i,s_i,t_i$ be the copies of~$v,s,t$, respectively.
  We set $V(G') = V(G_1) \cup V(G_2)$.
  We first set the arc costs in $E(G')$ to those of~$E(G_1)$ and~$E(G_2)$.
  For the remaining arcs $ e \in \{(u,w),(w,u) : u \in V(G_1), w \in V(G_2)\}$, we set $\cost(e)=2$.
  Afterwards, we change the weights of the following eight arcs:
  \begin{align*}
    \cost((s_1,v_2)) & := \cost((s_1,v_1)),&  \cost((v_2,s_1)) & := \cost((v_1,s_1)),\\
    \cost((s_2,v_1)) & := \cost((s_2,v_2)),&  \cost((v_1,s_2)) & := \cost((v_2,s_2)),\\
    \cost((s_1,v_1)) & := 2,               &  \cost((s_2,v_2)) & := 2,\\
    \cost((v_1,s_1)) & := 2,               &  \cost((v_2,s_2)) & := 2 \enspace .
  \end{align*}
  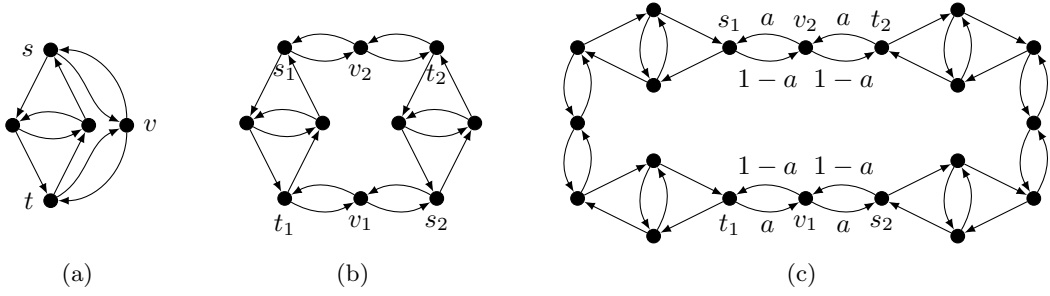
\begin{figure}[tb]
    \centering
  \begin{subfigure}[b]{0.22\textwidth}
    \centering
    \begin{tikzpicture}[->,scale=0.5]
      \node (1) at (0,2) [vertex]{};
      \node (2) at (1,0) [vertex, label=below:\phantom{M}, label=left:$t$]{};
      \node (3) at (1,4) [vertex, label=left:$s$]{};
      \node (4) at (2,2) [vertex]{};
      \node (5) at (3,2) [vertex, label=right:$v$]{};
      \draw[>=latex] (1) -> (2);
      \draw[>=latex] (2) -> (4);
      \draw[>=latex] (4) -> (3);
      \draw[>=latex] (3) -> (1);
      \draw[>=latex] (1) to[bend right] (4);
      \draw[>=latex] (4) to[bend right] (1);
      \draw[>=latex] (2) to[out=30, in=205] (5);
      \draw[>=latex] (5) to[out=270, in = 0] (2);
      \draw[>=latex] (3) to[out=330, in=155] (5);
      \draw[>=latex] (5) to[out=90, in=0] (3);
    \end{tikzpicture}
    \caption{$~$}
    \label{fig:gap_directed}
  \end{subfigure}
  \begin{subfigure}[b]{0.22\textwidth}
    \centering
    \begin{tikzpicture}[->,scale=0.5]
      \node (1) at (0,2) [vertex]{};
      \node (2) at (1,0) [vertex, label=below:\phantom{M}, label=below:$t_1$]{};
      \node (3) at (1,4) [vertex, label=below:$s_1$]{};
      \node (4) at (2,2) [vertex]{};
      \node (5) at (3,4) [vertex, label=below:$v_2$]{};
      \node (10) at (4,2) [vertex]{};
      \node (20) at (5,0) [vertex, label=below:\phantom{M}, label=below:$s_2$]{};
      \node (30) at (5,4) [vertex, label=below:$t_2$]{};
      \node (40) at (6,2) [vertex]{};
      \node (50) at (3,0) [vertex, label=below:$v_1$]{};
      \draw[>=latex] (1) -> (2);
      \draw[>=latex] (2) -> (4);
      \draw[>=latex] (4) -> (3);
      \draw[>=latex] (3) -> (1);
      \draw[>=latex] (1) to[bend right] (4);
      \draw[>=latex] (4) to[bend right] (1);
      \draw[>=latex] (2) to[bend right] (50);
      \draw[>=latex] (50) to[bend right] (2);
      \draw[>=latex] (3) to[bend right] (5);
      \draw[>=latex] (5) to[bend right] (3);
      \draw[>=latex] (10) -> (20);
      \draw[>=latex] (20) -> (40);
      \draw[>=latex] (40) -> (30);
      \draw[>=latex] (30) -> (10);
      \draw[>=latex] (10) to[bend right] (40);
      \draw[>=latex] (40) to[bend right] (10);
      \draw[>=latex] (20) to[bend right] (50);
      \draw[>=latex] (50) to[bend right] (20);
      \draw[>=latex] (30) to[bend right] (5);
      \draw[>=latex] (5) to[bend right] (30);
    \end{tikzpicture}
    \caption{$~$}
    \label{fig:gap_directed2}
  \end{subfigure}
  \begin{subfigure}[b]{0.5\textwidth}
  \centering
      \begin{tikzpicture}[->,scale=0.5]
        \node (1) at (2,0) [vertex]{};
        \node (2) at (0,1) [vertex]{};
        \node (t1) at (4,1) [vertex, label=below:$t_1$]{};
        \node (4) at (2,2) [vertex]{};
        \node (v1) at (6,1) [vertex, label=below:$v_1$]{};
        \node (10) at (2,4) [vertex]{};
        \node (20) at (0,5) [vertex]{};
        \node (s1) at (4,5) [vertex, label=above:$s_1$]{};
        \node (40) at (2,6) [vertex]{};
        \node (50) at (0,3) [vertex]{};
        \node (100) at (10,0) [vertex]{};
        \node (s2) at (8,1) [vertex, label=below:$s_2$]{};
        \node (300) at (12,1) [vertex]{};
        \node (400) at (10,2) [vertex]{};
        \node (500) at (12,3) [vertex]{};
        \node (1000) at (10,4) [vertex]{};
        \node (t2) at (8,5) [vertex, label=above:$t_2$]{};
        \node (3000) at (12,5) [vertex]{};
        \node (4000) at (10,6) [vertex]{};
        \node (v2) at (6,5) [vertex, label=above:$v_2$]{};
        \draw[>=latex] (1) -> (2);
        \draw[>=latex] (2) -> (4);
        \draw[>=latex] (4) -> (t1);
        \draw[>=latex] (t1) -> (1);
        \draw[>=latex] (1) to[bend right] (4);
        \draw[>=latex] (4) to[bend right] (1);
        \draw[>=latex] (2) to[bend right] (50);
        \draw[>=latex] (50) to[bend right] (2);
        \draw[>=latex] (t1) to[bend right] node[pos=.5,below] {$a$} (v1);
        \draw[>=latex] (v1) to[bend right] node[midway,above] {$1 - a$} (t1);
        \draw[>=latex] (10) -> (20);
        \draw[>=latex] (20) -> (40);
        \draw[>=latex] (40) -> (s1);
        \draw[>=latex] (s1) -> (10);
        \draw[>=latex] (10) to[bend right] (40);
        \draw[>=latex] (40) to[bend right] (10);
        \draw[>=latex] (20) to[bend right] (50);
        \draw[>=latex] (50) to[bend right] (20);
        \draw[>=latex] (s1) to[bend right] node[midway,below] {$1 - a$} (v2);
        \draw[>=latex] (v2) to[bend right] node[pos=.5,above] {$a$} (s1);
        \draw[>=latex] (100) -> (s2);
        \draw[>=latex] (s2) -> (400);
        \draw[>=latex] (400) -> (300);
        \draw[>=latex] (300) -> (100);
        \draw[>=latex] (100) to[bend right] (400);
        \draw[>=latex] (400) to[bend right] (100);
        \draw[>=latex] (s2) to[bend right] node[midway,above] {$1 - a$} (v1);
        \draw[>=latex] (v1) to[bend right] node[pos=.5,below] {$a$} (s2);
        \draw[>=latex] (300) to[bend right] (500);
        \draw[>=latex] (500) to[bend right] (300);
        \draw[>=latex] (1000) -> (t2);
        \draw[>=latex] (t2) -> (4000);
        \draw[>=latex] (4000) -> (3000);
        \draw[>=latex] (3000) -> (1000);
        \draw[>=latex] (1000) to[bend right] (4000);
        \draw[>=latex] (4000) to[bend right] (1000);
        \draw[>=latex] (t2) to[bend right] node[pos=.5,above] {$a$} (v2);
        \draw[>=latex] (v2) to[bend right] node[midway,below] {$1 - a$} (t2);
        \draw[>=latex] (3000) to[bend right] (500);
        \draw[>=latex] (500) to[bend right] (3000);
      \end{tikzpicture}
      \caption{~}
      \label{appfig:gap_join}
      \end{subfigure}
    \caption{
        \label{fig:directed_intgps}
        Example how to join two instances.
        }
  \end{figure}
  Let $a := x^*_{(v,s)}$.
  Then $x^*_{(t,v)} = a$ and $x^*_{(s,v)} = x^*_{(v,t)} = 1-a$, again due to the equality constraints.

  We set
  \begin{align*}
    x'_{(v_1,s_2)}, x'_{(v_2,s_1)}& := a, & x'_{(s_1,v_2)}, x'_{(s_2,v_1)}& := 1 - a,\\
    x'_{(s_1,v_1)}, x'_{(v_1,s_1)}& := 0, & x'_{(s_2,v_2)}, x'_{(v_2,s_2)}& := 0 \enspace .
  \end{align*}
  For all remaining arcs of $G_1$ and $G_2$ we set~$x'$ according to $x^*$.
  For all arcs not yet considered, we set~$x'$ to zero.
  Clearly, $\cost(x') = 2 \cost(x^*)$, since for each arc of $G_1$ and $G_2$ there is a corresponding arc with the same LP value and the same cost in $G'$.

  Let us now verify that $x'$ is a feasible solution to $\SER(G)$.
  For the equality constraints it suffices to inspect $v_1$, $v_2$, $s_1$, $s_2$, $t_1$, $t_2$, because the arcs incident to all other vertices did not change.
  Since we simply replaced arcs by others with both identical cost and LP value, all equality constraints are satisfied.

  Similarly, for all $S \subseteq V_1$ and $S \subseteq V_2$, the subtour elimination constraints are satisfied, where we use that any cut containing $(s_1,v_1)$ also contains $(s_1,v_2)$; for all other changed arcs, the situation is analogous.
  If either both $v_1,v_2$ or none of $v_1,v_2$ belongs to $S$, then consider the graph where $v_1$ and~$v_2$ are identified to a single vertex, keeping parallel arcs.
  We have $x'(\delta^-_{G'}(S)) \ge x'(\delta^-_{G'}(S) \cap E(G_1)) \ge x'(\delta^-_{G_1}(S \cap V(G_1))) = x^*(\delta^-_{G_1}(S \cap V(G_1))) \ge 1$, and thus we are done.

  Finally, let us fix a set $S$ not yet considered. 
  Since $\delta^-_{G'}(S) = \delta^-_{G'}(V(G') \setminus S)$, we may assume that $v_1 \in S$ and $v_2 \notin S$.

  We have $x^*(\delta^-_{G_2}(S \cap V(G_2))) \ge 1$, and $(v_2,s_2)$ is the only arc in $\delta^-_{G_2}(S \cap V(G_2)) \cap E(G_2)$ whose value in~$G'$ and $x'$ could have changed.
  Therefore, either $(s_2,v_1) \in \delta^-_{G'}(S)$ or $x'(\delta^-_{G'}(S \cap V(G_2))\cap E(G_2)) \ge 1-a$.

  Analogously, we have $x^*(\delta^-_{G_1}(S \cap V(G_1))) \ge 1$, and $(s_1,v_1)$ is the only arc in $\delta^-_{G'}(S \cap V(G_1)) \cap E(G_1)$ whose value in $G'$ and $x'$ could have changed.
  Therefore, either $(v_2,s_1) \in \delta^-_{G'}(S)$ or $x'(\delta^-_{G'}(S \cap V(G_1))\cap E(G_1)) \ge a$.

  Thus,
  \begin{equation*}
    \begin{split}
      x'(\delta^-_{G'}(S)) & \ge  \min\{x'(\delta^-_{G'}(S \cap V(G_2))\cap E(G_2)), x'(\delta^-_{G'}(S) \cap (s_2,v_1))\}\\
                           &\qquad + \min\{x'(\delta^-_{G'}(S \cap V(G_1))\cap E(G_1)), x'(\delta^-_{G'}(S) \cap (v_2,s_1))\} \\
                           & = 1 - a + a = 1,
    \end{split}
  \end{equation*}
  where $x'(\emptyset)=0$.
  Therefore, $x'$ is a feasible solution to $\SER(G')$.
  
  We now show that $\Opt(G') = 2 \Opt(G)$.
  To see that $\Opt(G') \le 2 \Opt(G)$, consider a tour~$C$ in $G$. 
  If neither $(v,s)$ nor $(s,v)$ is in $E(C)$, the cost of copies $C_1,C_2$ of the tours in $G_1$ and~$G_2$ does not change in $G'$, and we can combine both tours $C_1,C_2$ to a single tour by removing two arcs of cost two and introducing two arcs of cost at most two.
  Otherwise, by symmetry we may assume that $(s,v) \in E(C)$. 
  Then in $G'$ we replace $(s_1,v_1)$ of $C_1$ by $(s_1,v_2)$ and $(s_2,v_2)$ of $C_2$ by $(s_2,v_1)$ to obtain a new tour, whose cost again is at most $2\cost(C)$.
  This shows that $\Opt(G') \le 2 \Opt(G)$.

  We continue to prove that $\Opt(G') \ge 2 \Opt(G)$.
  Assume, for the sake of contradiction, that there is a tour $C$ of cost $\cost(C) < 2\Opt(G)$.
  Let $k$ be the number of cost two arcs in $C$. 
  We remove these $k$ arcs of cost two from $C$, and obtain a collection $C'$ of paths (unless $k=0$).

  \newcommand{\rev}[1]{\ensuremath{\mathrm{rev}(#1)}\xspace}
  For any subgraph $H$ of $G'$, we define a reverse construction \rev{H} from~$H$ as follows:
  \begin{itemize}
    \item if $(s_1,v_2) \in E(H)$, remove $(s_1,v_2)$ and add $(s_1,v_1)$;
    \item if $(v_2,s_1) \in E(H)$, remove $(v_2,s_1)$ and add $(v_1,s_1)$;
    \item if $(s_2,v_1) \in E(H)$, remove $(s_2,v_1)$ and add $(s_2,v_2)$;
    \item if $(v_1,s_2) \in E(H)$, remove $(v_1,s_2)$ and add $(v_2,s_2)$.
  \end{itemize}
  The arc costs in \rev{H} are reversed accordingly.
  Then $\rev{C'}$ is a modified collection of paths and cycles without paths crossing between $G_1$ and $G_2$.
  For $i = 1,2$, let $C'_i$ be the restriction of \rev{C'} to $G_i$.

  Let us first assume that
  \begin{align*}
  |\{(s_1,v_2),(t_1,v_1)\} \cap E(C')| < 2, &\qquad |\{(v_2,s_1),(v_1,t_1)\} \cap E(C')| < 2,\\
  |\{(v_2,t_2),(v_1,s_2)\} \cap E(C')| < 2, &\qquad |\{(t_2,v_2),(s_2,v_1)\} \cap E(C')| < 2.
  \end{align*}
  Intuitively, this means that if there are two paths crossing between $G_1$ and $G_2$ via~$v_1$ and $v_2$, they have opposite directions.

  Then both the in-degree and out-degree of each vertex in \rev{C'} is at most one.
  Since we did not change the number of arcs, there are still~$k$ paths and either zero, one or two cycles. 
  By introducing~$k$ arcs of cost two, we obtain a tour~$C_1$ in~$G_1$ and a tour $C_2$ in $G_2$ such that $\cost(C_1) + \cost(C_2) \le \Opt(G')$.
  Therefore, $\min\{\cost(C_1),\cost(C_2)\} < \Opt(G)$, a contradiction.

  Otherwise, by symmetry, it is sufficient to consider the case $|\{(v_2,s_1),(v_1,t_1)\} \cap E(C')| = 2$ such that $|E(C'_1)| \ge |E(C'_2)|$. 
  Then in $C'_1$ there are two paths $P'_1$ and~$P'_2$ starting from $v_1$ and $v_2$ and all vertices have in-degrees and out-degrees of at most one. 
  By adding at most $k/2-1$ arcs to~$C'_1$, we can extend~$P'_1$ and $P'_2$ to~$P_1$ and $P_2$ such that in $\rev{G'}$ it holds $V(P_1) \cap V(P_2) = \{v_1\}$, $V(P_1) \cup V(P_2) = V(G_1)$, and both $P_1$ and $P_2$ start in $v_1$. 

  However, since there are exactly as many arcs crossings from $G_2$ to $G_1$ as from $G_1$ to $G_2$, in~$G'$ there are at least two arcs of cost two from $G_1$ to $G_2$.
  Therefore, $k \ge 2$ and $E(P_1) \cup E(P_2)$ contains at most $k/2-1$ arcs of cost two.
  By basic graph theory, $|E(P_1) \cup E(P_2)| = |V(G_1)|-1$, and therefore $\cost(P_1) + \cost(P_2) \le |V(G_1)|-1 + k/2 - 1 \le \Opt(G')/2 - 2 < \Opt(G) - 2$, contradicting the assumptions of the lemma.

  To finish the proof, we have to show that also $G'$ satisfies the conditions of the lemma.
  Again, we show the claim by means of contradiction. 
  Suppose there are two paths $P_1$, $P_2$ in~$G'$ such that $V(P_1) \cap V(P_2) = \{v_1\}$, $V(P_1) \cup V(P_2) = V(G')$ with $\cost(P_1) + \cost(P_2) < \Opt(G')-2$.
  
  Consider the case when both $P_1$ and $P_2$ start in $v_1$.
  Let $k$ be the number of arcs of cost two in $E(P_1) \cup E(P_2)$.
  We remove all $k$ arcs of cost two and obtain a collection $C$ of $k+1$ paths (where the component containing $v_1$ is counted as one path).
  Note that the paths in $C$ can only pass between $G_1$ and~$G_2$ via $v_1$ and $v_2$.
  Therefore, we can use $k$ arcs to connect the paths of $C$ such that we obtain two paths $P'_1$ and $P'_2$ where $V(P'_1) \cap V(P'_2) = \{v_1\}$, $V(P'_1) \cup V(P'_2) = V(G')$, both $P'_1,P'_2$ start in $v_1$, $\cost(P'_1) + \cost(P'_2) < \Opt(G')-2$, and additionally there are no cost two arcs between $G_1$ and $G_2$.
  Note that one of the paths, say $P'_1$, has to contain $v_2$.
  Thus, $\rev{P'_1}$ contains a tour in either $G_1$ or in $G_2$.
  By renaming, let us assume that the tour is in $G_1$.
  Then $\rev{P'_2}$ and the part of $\rev{P'_1}$ in $G_2$ form two paths $P''_1, P''_2$ in $G_2$ such that $V(P''_1) \cap V(P''_2) = \{v_2\}$, $V(P''_1) \cup V(P''_2) = V(G_2)$, both $P''_1,P''_2$ start in~$v_2$ or both $P''_1,P''_2$ end in $v_2$, and by assumption of the lemma, $\cost(P''_1) + \cost(P''_2) \ge \Opt(G)-2$.
  We obtain $\cost(P_1) + \cost(P_2) \ge \cost(P'_1) + \cost(P'_2) \ge \Opt{G} + \Opt{G} - 2 = \Opt{G'}-2$, a contradiction.

  The remaining cases where $P_1$ and $P_2$ start in $v_2$, or end in $v_1$, or end in $v_2$, are dealt with analogously.
\end{proof}

\section{Intractability of Local Search for \STSP}
\label{sec:intractabilityoflocalsearch}
In this section we show that searching the $k$-edge change neighborhood is $\mathsf{W}[1]$-hard for \STSP, which means that finding the best tour in the $k$-edge change neighborhood essentially requires complete search.
Our result strengthens the result by D{\'a}niel Marx~\cite{Marx2008} that local search for metric TSP with three distances is $\mathsf{W}[1]$-hard parameterized by the edge change distance.
Compared to his proof, we only have two distances, one and two, and we make some small simplifications.
We choose to follow the lines of Marx so that the reader can easily make comparisons.

The local search problem for \STSP is, given a tour $\mathcal T$, to find the best tour in the $k$-edge change neighborhood of $\mathcal T$, which are those tours that can be reached from the current tour by replacing at most~$k$ edges.
On the one hand, if $k$ is part of the input, then this problem is $\mathsf{NP}$-hard, as for $k = n$, the problem is equivalent to finding the best possible tour.
On the other hand, the problem is polynomial-time solvable for every fixed value of $k$, in $n^{O(k)}$ time by complete search.

The main result of this section is that finding the best tour in the $k$-edge change neighborhood is $\mathsf{W}[1]$-hard, which implies that the problem is not fixed-parameter tractable, unless $\mathsf{W}[1]=\mathsf{FPT}$ (which would imply subexponential time algorithms for many canonical $\mathsf{NP}$-complete problems.

We consider a tour as a set of $n$ ordered pairs of cities.
If $X$ and $Y$ are two tours on the same set of cities, then the \emph{distance} of $X$ and $Y$ is $|X\setminus Y| = |Y \setminus X|$.
Formally, we study the parameterized complexity of the following problem:\\

\framebox[0.95\textwidth]{
   \begin{tabular}{rl}
     \multicolumn{2}{l}{{\sc $k$-Edge Change \STSP}}\\
     Input: & A set $V$ of $n$ cities, a symmetric distance matrix over all city\\
            & pairs with distances 1 and 2, a tour $C$, and an integer $k\in\mathbb N$.\\
     Question: & Is there a tour $C'$ with length less than that of $C$, and with\\
               & distance at most $k$ from $C$?
   \end{tabular}
}\\

\begin{theorem}
\label{thm:kchange12tspisw1hard}
  Given an instance of \STSP and a Hamiltonian cycle $C$ of it, it is $\mathsf{W}[1]$-hard to decide whether there is a Hamilton cycle $C'$ such
  that the cost of $C'$ is strictly less than the cost of $C$ and the distance between $C$ and $C'$ is at most $k$.
\end{theorem}
The proof is by reduction from {\sc $t$-Clique}: Given an undirected unweighted graph $H$ where we have to find a clique of size $t$, we construct an equivalent instance of {\sc $k$-Edge Change \STSP} on an undirected graph $G$.
Assume, without loss of generality, that $t$ is odd, for otherwise we add a new vertex that is adjacent to all vertices of $H$ and find a clique of size~$t + 1$.

\subsection{The Switch Gadget}
The graph $G$ is built from several copies of the switch gadget shown in Fig.~\ref{fig:switchgadget_a}.
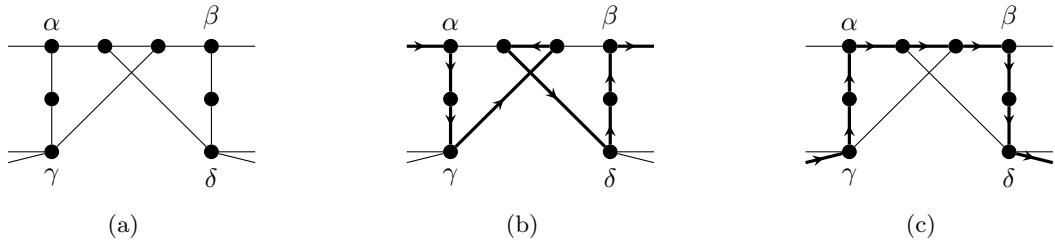
\begin{figure}[htpb]
  \centering
  \begin{subfigure}[b]{0.3\textwidth}
    \centering
    \begin{tikzpicture}[scale=0.7]
      \node (1) at (0,1){};
      \node (alpha) at (1,1) [vertex,label=above:$\alpha$]{};
      \node (2) at (2,1) [vertex]{};
      \node (3) at (3,1) [vertex]{};
      \node (beta) at (4,1) [vertex,label=above:$\beta$]{};
      \node (4) at (5,1){};
      \node (5) at (1,0)[vertex]{};
      \node (6) at (4,0)[vertex]{};
      \node (7) at (0,-1){};
      \node (8) at (0,-1.25){};
      \node (gamma) at (1,-1)[vertex,label=below:$\gamma$]{};
      \node (delta) at (4,-1)[vertex,label=below:$\delta$]{};
      \node (9) at (5,-1)[]{};
      \node (10) at (5,-1.25)[]{};
      \draw[](1)--(alpha);
      \draw[](alpha)--(2);
      \draw[](2)--(3);
      \draw[](3)--(beta);
      \draw[](beta)--(4);
      \draw[](alpha)--(5);
      \draw[](beta)--(6);
      \draw[](5)--(gamma);
      \draw[](6)--(delta);
      \draw[](7)--(gamma);
      \draw[](8)--(gamma);
      \draw[](gamma)--(3);
      \draw[](2)--(delta);
      \draw[](delta)--(9);
      \draw[](delta)--(10);
    \end{tikzpicture}
    \caption{$~$}
    \label{fig:switchgadget_a}
  \end{subfigure}
  \quad
  \begin{subfigure}[b]{0.3\textwidth}
    \centering
    \begin{tikzpicture}[scale=0.7]
      \node (1) at (0,1){};
      \node (alpha) at (1,1) [vertex,label=above:$\alpha$]{};
      \node (2) at (2,1) [vertex]{};
      \node (3) at (3,1) [vertex]{};
      \node (beta) at (4,1) [vertex,label=above:$\beta$]{};
      \node (4) at (5,1){};
      \node (5) at (1,0)[vertex]{};
      \node (6) at (4,0)[vertex]{};
      \node (7) at (0,-1){};
      \node (8) at (0,-1.25){};
      \node (gamma) at (1,-1)[vertex,label=below:$\gamma$]{};
      \node (delta) at (4,-1)[vertex,label=below:$\delta$]{};
      \node (9) at (5,-1)[]{};
      \node (10) at (5,-1.25)[]{};
      \draw[postaction={decorate},very thick](1)--(alpha);
      \draw[](alpha)--(2);
      \draw[postaction={decorate},very thick](3)--(2);
      \draw[](3)--(beta);
      \draw[postaction={decorate},very thick](beta)--(4);
      \draw[postaction={decorate},very thick](alpha)--(5);
      \draw[postaction={decorate},very thick](6)--(beta);
      \draw[postaction={decorate},very thick](5)--(gamma);
      \draw[postaction={decorate},very thick](delta)--(6);
      \draw[](7)--(gamma);
      \draw[](8)--(gamma);
      \draw[postaction={decorate},very thick](gamma)--(3);
      \draw[postaction={decorate},very thick](2)--(delta);
      \draw[](delta)--(9);
      \draw[](delta)--(10);
    \end{tikzpicture}
    \caption{$~$}
    \label{fig:switchgadget_b}
  \end{subfigure}
  \quad
  \begin{subfigure}[b]{0.3\textwidth}
    \centering
    \begin{tikzpicture}[scale=0.7]
      \node (1) at (0,1){};
      \node (alpha) at (1,1) [vertex,label=above:$\alpha$]{};
      \node (2) at (2,1) [vertex]{};
      \node (3) at (3,1) [vertex]{};
      \node (beta) at (4,1) [vertex,label=above:$\beta$]{};
      \node (4) at (5,1){};
      \node (5) at (1,0)[vertex]{};
      \node (6) at (4,0)[vertex]{};
      \node (7) at (0,-1){};
      \node (8) at (0,-1.25){};
      \node (gamma) at (1,-1)[vertex,label=below:$\gamma$]{};
      \node (delta) at (4,-1)[vertex,label=below:$\delta$]{};
      \node (9) at (5,-1)[]{};
      \node (10) at (5,-1.25)[]{};
      \draw[](1)--(alpha);
      \draw[postaction={decorate},very thick](alpha)--(2);
      \draw[postaction={decorate},very thick](2)--(3);
      \draw[postaction={decorate},very thick](3)--(beta);
      \draw[](beta)--(4);
      \draw[postaction={decorate},very thick](5)--(alpha);
      \draw[postaction={decorate},very thick](beta)--(6);
      \draw[postaction={decorate},very thick](gamma)--(5);
      \draw[postaction={decorate},very thick](6)--(delta);
      \draw[](7)--(gamma);
      \draw[postaction={decorate},very thick](8)--(gamma);
      \draw[](gamma)--(3);
      \draw[](2)--(delta);
      \draw[](delta)--(9);
      \draw[postaction={decorate},very thick](delta)--(10);
    \end{tikzpicture}
    \caption{$~$}
    \label{fig:switchgadget_c}
  \end{subfigure}
  \caption{The switch gadget.}
\label{fig:switchgadget}
\end{figure}
This gadget was first introduced by Engebretsen and Karpinski~\cite{EngebretsenKarpinski2006}.
The gadget is connected to the rest of the graph at the vertices $\alpha, \beta, \gamma, \delta$.
It is easy to see that if a switch gadget is part of a larger graph and there is a Hamiltonian cycle in the larger graph, then this cycle traverses the gadget in one of the two ways presented in Fig.~\ref{fig:switchgadget_b} and~\ref{fig:switchgadget_c}.
Either the cycle enters at $\alpha$ and leaves at $\beta$ in as~\ref{fig:switchgadget_b}; or it enters at $\gamma$ and leaves at $\delta$ as in~\ref{fig:switchgadget_c}, or vice-versa.
We remark that the whole construction is undirected; the orientations only describe the order in which vertices are traversed by the tours.
The parameters of the reduction will be set in such a way that a Hamiltonian cycle $C'$ has to traverse a selection of switch gadget in way~\ref{fig:switchgadget_c} in order to prevent that the total cost of the cycle will be too large.
The gadget effectively acts as a switch: either it is used as an $\alpha\rightarrow\beta$-path, or as a $\gamma\rightarrow\delta$-path.
In the first case, we say that the cycle uses the \emph{upper gates} of the gadget, in the second case we say that the cycle uses the \emph{lower gates}.

\subsection{Vertex and Edge Segments}
Let $n$ be number of vertices in $H$, and let $m$ be its number of edges.
For convenience, we identify the vertices with the integers $\{1, \dotsc, n\}$ and the edges with the integers $\{1, \hdots,m\}$.
We construct an undirected graph $G$ that consists of $(n + m)t(t - 1)/2$ copies of the switch gadget and some additional vertices.
The graph $G$ represents an instance of \STSP where all edges of $G$ have cost $1$ and all edges not in $G$ have cost $2$.

We have the following switch gadgets:
\begin{itemize}
  \item \emph{vertex gadget} $V_{i,\{j_1,j_2\}}$ for each $i\in \{1,\hdots,n\}$ and $\{j_1,j_2\}\in {\{1,\hdots,t\} \choose 2}$ and
  \item \emph{edge gadget} $S_{i,\{j_1,j_2\}}$ for each $i\in \{1,\hdots,m\}$ and $\{j_1,j_2\}\in{\{1,\hdots,t\} \choose 2}$.
\end{itemize}
We form \emph{segments} from one or more gadgets.

The \emph{vertex segment} $V_{i,j}$, for $i \in \{1,\hdots,n\}$ and $j\in \{1,\hdots,t\}$, consists of the \emph{entrance vertex} $a_{i,j}$, the \emph{exit vertex}~$b_{i,j}$, and $(t - 1)/2$ gadgets $V_{i,\{j,j'\}}$.

Consider the pairs $\{j_1, j_2\}\in {\{1,\hdots,t\} \choose 2}$ and let $P_1,\dotsc, P_{t(t-1)/2}$ be an ordering of these pairs such that the second element of a pair is the same as the first element of the next pair, that is, $P_{\ell} = \{p_{\ell},p_{\ell+1}\}$ for every $\ell\in {\{1,\hdots,t\}\choose 2}$.
Such an ordering exists: consider for instance the sequence of edges of an Eulerian tour in $K_t$, the complete graph of order $t$.
(Since $t$ is odd, $K_t$ is an Eulerian graph.)
These pairs define a partition of all ${t \choose 2}$ pairs into $t$ classes $R_1,R_2,\dotsc,R_t$ of $(t-1)/2$ vertices each, where $R_j = \{\{p_{j'},p_{j'+1}\}~|~p_{j'} = j\}$ for $j \in \{1,\hdots,t\}$.
The vertex segment $V_{i,j}$ contains all vertex gadgets of pairs in $R_j$; see Fig.~\ref{fig:thevertexsegment}.
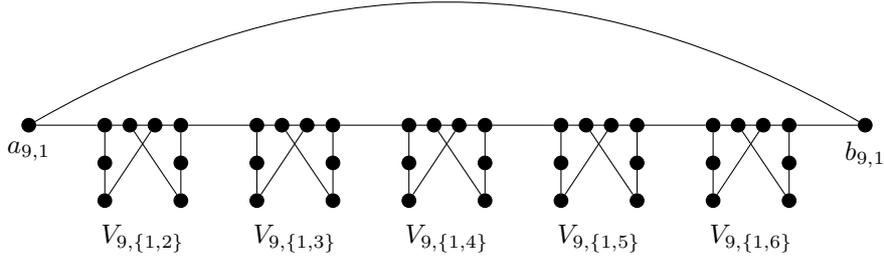
\begin{figure}[htpb]
  \centering
  \begin{tikzpicture}
    \node (1a) at (1,1) [vertex]{};
    \node (2a) at (2,1) [vertex]{};
    \node (3a) at (2,0) [vertex]{};
    \node (4a) at (1,0) [vertex]{};
    \node (5a) at (1,0.5)[vertex]{};
    \node (6a) at (2,0.5)[vertex]{};
    \node (7a) at (1.33,1)[vertex]{};
    \node (8a) at (1.66,1)[vertex]{};
    \draw(1a)--(2a);
    \draw(2a)--(3a);
    \draw(4a)--(1a);
    \draw(3a)--(7a);
    \draw(4a)--(8a);
    \node (1b) at (3,1) [vertex]{};
    \node (2b) at (4,1) [vertex]{};
    \node (3b) at (4,0) [vertex]{};
    \node (4b) at (3,0) [vertex]{};
    \node (5b) at (3,0.5)[vertex]{};
    \node (6b) at (4,0.5)[vertex]{};
    \node (7b) at (3.33,1)[vertex]{};
    \node (8b) at (3.66,1)[vertex]{};
    \draw(1b)--(2b);
    \draw(2b)--(3b);
    \draw(4b)--(1b);
    \draw(3b)--(7b);
    \draw(4b)--(8b);
    \draw[](2a)--(1b);
    \node (1c) at (5,1) [vertex]{};
    \node (2c) at (6,1) [vertex]{};
    \node (3c) at (6,0) [vertex]{};
    \node (4c) at (5,0) [vertex]{};
    \node (5c) at (5,0.5)[vertex]{};
    \node (6c) at (6,0.5)[vertex]{};
    \node (7c) at (5.33,1)[vertex]{};
    \node (8c) at (5.66,1)[vertex]{};
    \draw(1c)--(2c);
    \draw(2c)--(3c);
    \draw(4c)--(1c);
    \draw(3c)--(7c);
    \draw(4c)--(8c);
    \draw[](2b)--(1c);
    \node (1d) at (7,1) [vertex]{};
    \node (2d) at (8,1) [vertex]{};
    \node (3d) at (8,0) [vertex]{};
    \node (4d) at (7,0) [vertex]{};
    \node (5d) at (7,0.5)[vertex]{};
    \node (6d) at (8,0.5)[vertex]{};
    \node (7d) at (7.33,1)[vertex]{};
    \node (8d) at (7.66,1)[vertex]{};
    \draw(1d)--(2d);
    \draw(2d)--(3d);
    \draw(4d)--(1d);
    \draw(3d)--(7d);
    \draw(4d)--(8d);
    \draw[](2c)--(1d);
    \node (1e) at (9,1) [vertex]{};
    \node (2e) at (10,1) [vertex]{};
    \node (3e) at (10,0) [vertex]{};
    \node (4e) at (9,0) [vertex]{};
    \node (5e) at (9,0.5)[vertex]{};
    \node (6e) at (10,0.5)[vertex]{};
    \node (7e) at (9.33,1)[vertex]{};
    \node (8e) at (9.66,1)[vertex]{};
    \draw(1e)--(2e);
    \draw(2e)--(3e);
    \draw(4e)--(1e);
    \draw(3e)--(7e);
    \draw(4e)--(8e);
    \draw[](2d)--(1e);
    \node (exit) at (11,1) [vertex,label=below:$b_{9,1}$]{};
    \node (entry) at (0,1) [vertex,label=below:$a_{9,1}$]{};
    \draw[](entry)--(1a);
    \draw[](2e)--(exit);
    \node at (1.5,-0.5) {$V_{9,\{1,2\}}$};
    \node at (3.5,-0.5) {$V_{9,\{1,3\}}$};
    \node at (5.5,-0.5) {$V_{9,\{1,4\}}$};
    \node at (7.5,-0.5) {$V_{9,\{1,5\}}$};
    \node at (9.5,-0.5) {$V_{9,\{1,6\}}$};
    \draw [] (entry) edge[bend left] (exit);
  \end{tikzpicture}
  \caption{The vertex segment $V_{9,1}$ for $t = 11$ and one possible $R_1$.}
\label{fig:thevertexsegment}
\end{figure}

To simplify the notation, let $W_1,\dotsc,W_{(t-1)/2}$ be an arbitrary ordering of these gadgets.
For every $\ell\in\{1,\dotsc,(t-1)/2-1\}$, there is an edge from vertex $\beta$ of~$W_{\ell}$ to vertex $\alpha$ of $W_{\ell + 1}$.
Furthermore, there is an edge from $a_{i,j}$ to vertex $\alpha$ of gadget $W_1$, and there is an edge from vertex $\beta$ of gadget~$W_{t-1}$ to $b_{i,j}$.
Finally, there is a \emph{bypass edge} from $a_{i,j}$ to $b_{i,j}$.

The \emph{edge segment} $E_{i,\{j_1,j_2\}}$, for $i\in\{1,\hdots,m\}$ and $j_1,j_2\in\{1,\dotsc,t\}$ with $j_1\not=j_2$ contains an \emph{entrance vertex} $z_{i,\{j_1,j_2\}}$, an \emph{exit vertex} $q_{i,\{j_1,j_2\}}$, and the gadget $S_{i,\{j_1,j_2\}}$, see Fig.~\ref{fig:theedgesegment}.
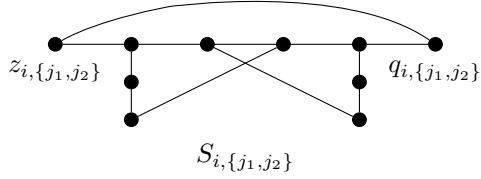
\begin{figure}[htpb]
  \centering
  \begin{tikzpicture}
      \node (1f) at (9,1) [vertex]{};
      \node (2f) at (12,1) [vertex]{};
      \node (3f) at (12,0) [vertex]{};
      \node (4f) at (9,0) [vertex]{};
      \node (h1) at (10,1)[vertex]{};
      \node (h2) at (11,1)[vertex]{};
      \node (g1) at (9,0.5)[vertex]{};
      \node (g4) at (12,0.5)[vertex]{};
      \draw(1f)--(2f);
      \draw(2f)--(3f);
      \draw(4f)--(1f);
      \draw(3f)--(h1);
      \draw(4f)--(h2);
      \node (exit) at (13,1) [vertex,label=below:$q_{i,\{j_1,j_2\}}$]{};
      \node (entry) at (8,1) [vertex,label=below:$z_{i,\{j_1,j_2\}}$]{};
      \draw[](entry)--(1f);
      \draw[](2f)--(exit);
      \node at (10.5,-0.5) {$S_{i,\{j_1,j_2\}}$};
      \draw [] (entry) to [out=+30,in=-5,in looseness=-1.5] +(1.5,+0.5) to [out=+5,in=145,out looseness=-1.5](exit);
  \end{tikzpicture}
  \caption{The edge segment $E_{i,\{j_1,j_2\}}$.}
\label{fig:theedgesegment}
\end{figure}
There is an edge from $z_{i,\{j_1,j_2\}}$ to vertex $\alpha$ of $S_{i,\{j_1,j_2\}}$ and an edge from vertex $\beta$ of $S_{i,\{j_1,j_2\}}$ to
$q_{i,\{j_1,j_2\}}$.
Additionally, there is a \emph{bypass edge} from $z_{i,\{j_1,j_2\}}$ to $q_{i,\{j_1,j_2\}}$.

Consider an arbitrary ordering of the $nt+mt(t-1)/2$ segments defined above.
Add an edge from the exit of each segment to the entrance of the next segment.
Note that in the \STSP instance defined by $G$, there is an edge of cost 2 that goes from the exit of the last segment (denote it by $v_{\textnormal{last}}$) to the entrance of the first segment (denote it by $v_{\textnormal{first}}$).
There will be some more edges in the graph~$G$, but before completing the description of $G$, we first define the Hamiltonian cycle~$C$.
The cycle starts at $v_{\textnormal{first}}$, goes through the upper gate of the gadget(s) in the first segment, leaves the segment at the exit, enters the second segment at its entrance, etc.
The cycle does not use the bypass edges, thus it visits every vertex of every gadget.
Finally, when~$C$ reaches the exit of the last segment ($v_{\textnormal{last}}$), it goes back to the entrance of the first segment~($v_{\textnormal{first}}$) using the edge of cost $2$.
The cycle traverses exactly 8 vertices in each switch gadget and additionally two vertices of each segment.
Therefore, the total number of vertices in~$G$ is
\newcommand{\orderG}{\ensuremath{4 n t^2 - 2 n t + 5 m t (t-1)}\xspace}
$nt(8(t-1)/2 + 2) + (8 + 2) m t(t-1)/2 = \orderG$, and the total cost of~$C$ is $\orderG + 1$.

\subsection{Encoding the Graph}
As discussed above, we will ensure that the cycle $C'$ traverses every gadget as either ~\ref{fig:switchgadget_b} or~\ref{fig:switchgadget_c} in Fig.~\ref{fig:switchgadget}.
In the latter case, we say that the gadget is \emph{active}.
We will show that the active gadgets describe a $t$-clique of graph $H$.
If gadget~$S_{i,\{j_1,j_2\}}$ in edge segment $E_{i,\{j_1,j_2\}}$ is active, then this means that the edge $i$ is the edge connecting the $j_1$-th and $j_2$-th vertices of the clique.
If gadget~$V_{i,\{j_1,j_2\}}$ is active, then this means that vertex $i$ is incident to the edge between the $j_1$-th and $j_2$-th vertex of the clique.
We connect the  vertices of gadgets in a way that enforces that the active gadgets describe a clique.

For every edge gadget, we add edges as follows (see Fig.~\ref{fig:connectingsegments}).
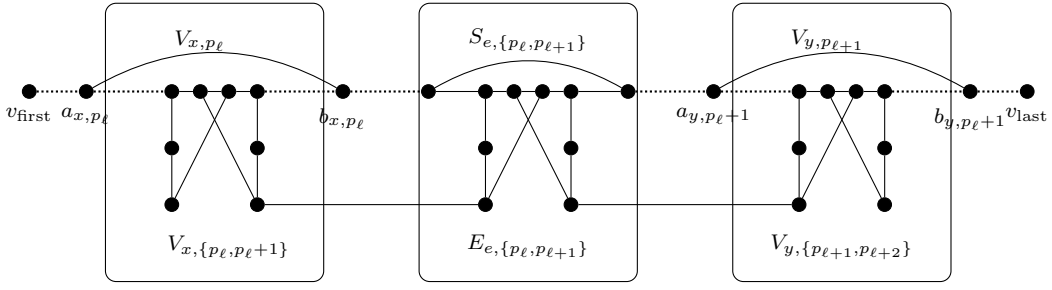
\begin{figure}
  \centering
  \begin{tikzpicture}[scale=0.75]
    \tikzstyle{every node}=[font=\footnotesize]
    \node (v-first) at (1.5,1) [vertex,label=below:$v_{\textnormal{first}}$]{};
    \node (axpl)    at (2.5,1) [vertex,label=below:$a_{x,p_{\ell}}$]{};
    \node (1a)      at (4,1) [vertex]{};
    \node (1b)      at (4.5,1) [vertex]{};
    \node (1c)      at (4,0) [vertex]{};
    \node (1e)      at (4,-1) [vertex]{};
    \node (2a)      at (5,1) [vertex]{};
    \node (2b)      at (5.5,1) [vertex]{};
    \node (2d)      at (5.5,0) [vertex]{};
    \node (2f)      at (5.5,-1) [vertex]{};
    \node (bxpl)    at (7,1) [vertex,label=below:$b_{x,p_{\ell}}$]{};
    \node (e-entry) at (8.5,1)[vertex]{};
    \node (3a)      at (9.5,1)[vertex]{};
    \node (3b)      at (10,1)[vertex]{};
    \node (3c)      at (9.5,0)[vertex]{};
    \node (3e)      at (9.5,-1)[vertex]{};
    \node (3ax)      at (10.5,1)[vertex]{};
    \node (3bx)      at (11,1)[vertex]{};
    \node (3dx)      at (11,0)[vertex]{};
    \node (3fx)      at (11,-1)[vertex]{};
    \node (e-exit)  at (12,1)[vertex]{};
    \node (aypl1)   at (13.5,1)[vertex,label=below:$a_{y,p_{\ell}+1}$]{};
    \node (4a)      at (15,1)[vertex]{};
    \node (4b)      at (15.5,1)[vertex]{};
    \node (4c)      at (15,0)[vertex]{};
    \node (4e)      at (15,-1)[vertex]{};
    \node (5a)      at (16,1)[vertex]{};
    \node (5b)      at (16.5,1)[vertex]{};
    \node (5d)      at (16.5,0)[vertex]{};
    \node (5f)      at (16.5,-1)[vertex]{};
    \node (bypl1)   at (18,1)[vertex,label=below:$b_{y,p_{\ell}+1}$]{};
    \node (v-last)  at (19,1)[vertex,label=below:$v_{\textnormal{last}}$]{};
    %
    %
    \node           at (5,-1.75){$V_{x,\{p_{\ell},p_\ell+1\}}$};
    \node           at (10.25,-1.75){$E_{e,\{p_{\ell},p_{\ell+1}\}}$};
    \node           at (15.75,-1.75){$V_{y,\{p_{\ell+1},p_{\ell+2}\}}$};
    \node           at (4.5,1.9){$V_{x,p_{\ell}}$};
    \node           at (10.25,1.9){$S_{e,\{p_{\ell},p_{\ell+1}\}}$};
    \node           at (15.5,1.9){$V_{y,p_{\ell+1}}$};
    %
    %
    \draw[densely dotted,thick](v-first)--(axpl);
    \draw[densely dotted,thick](axpl)--(1a);
    \draw[](1a)--(1b);
    \draw[](1a)--(1c);
    \draw[](1c)--(1e);
    \draw[](1b)--(2a);
    \draw[](2a)--(2b);
    \draw[](2b)--(2d);
    \draw[](2d)--(2f);
    \draw[](1e)--(2a);
    \draw[](2f)--(1b);
    \draw[densely dotted,thick](2b)--(bxpl);
    \draw[densely dotted,thick](bxpl)--(e-entry);
    \draw[](e-entry)--(3a);
    \draw[](3a)--(3b);
    \draw[](3a)--(3c);
    \draw[](3c)--(3e);
    \draw[](3b)--(3ax);
    \draw[](3ax)--(3bx);
    \draw[](3bx)--(3dx);
    \draw[](3dx)--(3fx);
    \draw[](3bx)--(e-exit);
    \draw[](3e)--(3ax);
    \draw[](3fx)--(3b);
    \draw[densely dotted,thick](e-exit)--(aypl1);
    \draw[densely dotted,thick](aypl1)--(4a);
    \draw[](4a)--(4b);
    \draw[](4a)--(4c);
    \draw[](4c)--(4e);
    \draw[](4b)--(5a);
    \draw[](5a)--(5b);
    \draw[](5b)--(5d);
    \draw[](5d)--(5f);
    \draw[](4e)--(5a);
    \draw[](5f)--(4b);
    \draw[densely dotted,thick](5b)--(bypl1);
    \draw[densely dotted,thick](bypl1)--(v-last);
    \draw[](2f)--(3e);
    \draw[](3fx)--(4e);
    \draw[] (axpl) edge[bend left] (bxpl);
    \draw [] (aypl1) edge[bend left] (bypl1);
    \draw [] (e-entry) edge[bend left] (e-exit);
    \node[rectangle, draw, text width=7.5em, text centered, rounded corners, minimum height=10.5em] (leftbox) at (4.75,0.1) {};
    \node[rectangle, draw, text width=7.5em, text centered, rounded corners, minimum height=10.5em] (midbox) at (10.25,0.1) {};
    \node[rectangle, draw, text width=7.5em, text centered, rounded corners, minimum height=10.5em] (rightbox) at (15.75,0.1) {};
  \end{tikzpicture}
  \caption{If $x$ and $y$ are the two end points of edge $e$, then segments
  $V_{x,p_{\ell}}$, $V_{y,p_{\ell + 1}}$, and $S_{e,\{p_{\ell},p_{\ell + 1}\}}$  are connected as shown above.
  The dotted lines represent missing sequences of gadgets and segments.}
\label{fig:connectingsegments}
\end{figure}
If vertices~$i,i'$ are endpoints of edge~$r$, then there is an edge from vertex $\delta$ of $V_{i,\{p_{\ell},p_{\ell + 1}\}}$ to vertex~$\gamma$  of $E_{r,\{p_{\ell},p_{\ell + 1}\}}$, and there is an edge from vertex~$\delta$ of $E_{r,\{p_{\ell},p_{\ell + 1}\}}$ to vertex $\gamma$ of $V_{i',\{p_{\ell+1},p_{\ell + 2})}$.
Furthermore, for every $i\in\{1,\dotsc,n\}$, there is an edge from~$v_{\textnormal{last}}$ to vertex~$\gamma$ of $V_{i,\{p_1,p_2\}})$, and for each edge $r$ there is an edge from vertex $\delta$ of~$S_{r,\{p_{t(t-1)/2-1},p_{t(t-1)}\}}$ to~$v_{\textnormal{first}}$.
This completes the description of the graph $G$.

\subsection{\texorpdfstring{{\sc $k$-Edge Change \STSP} $\Rightarrow$ {\sc $t$-Clique}}{k-Edge Change \STSP -> t-Clique}}
We claim that $G$ admits a Hamiltonian cycle $C'$ of weight strictly less than $\orderG + 1$ that is at distance at most $k = 7t(t-1) + 2(t + 1)$ from $C$, then there is a $t$-clique in $H$.
The cycle~$C'$ has as many edges as vertices, hence the only way the total weight is at most $\orderG$ is if $C'$ does not use the edge of weight 2 that goes from $v_{\textnormal{last}}$ to~$v_{\textnormal{first}}$.
Furthermore, every gadget has to be traversed either as~\ref{fig:switchgadget_b} or~\ref{fig:switchgadget_c} of Fig.~\ref{fig:switchgadget}.

Let us think about the cycle $C'$ as a path that starts from and returns to~$v_{\textnormal{first}}$.
Similar to $C$, the cycle~$C'$ has to go through the segments one by one.
It is clear that if $C'$ enters a segment at its entrance, then it has to leave it via its exit as otherwise it has gadgets that cannot be collected entirely with edges of $G$.
However, inside a segment, $C'$ can do two things: either it goes through the gadget(s) (similarly to~ $C$), or it skips the gadget(s) using the bypass edge.
In the latter case, we say that the segment is \emph{active}.
If vertex segment $V_{i,j}$ is active, then we will take it as an indication that vertex $i$ should be the $j$-th vertex of the clique.
If edge segment~$E_{i,(j_1,j_2)}$ is active, then this will mean that the $j_1$-th and the $j_2$-th vertices of the clique are connected by edge $i$.
By the time~$C'$ reaches $v_{\textnormal{last}}$, every gadget is completely traversed, or not visited at all.
The cycle has to return to $v_{\textnormal{first}}$ by visiting all the skipped gadgets.
As a result, the only possibility to continue from $v_{\textnormal{last}}$ without an edge of cost $2$ is to go to the~$\gamma$-vertex of some gadget $V_{i,\{p_1,p_2\}}$ of an active vertex segment $V_{i,p_1}$.
The tour has to continue from~$\delta$ of~$V_{i,\{p_1,p_2\}}$ to~$\gamma$ of some edge gadget $S_{r,\{p_1,p_2\}}$ where $r$ is an edge incident to $i$ and $E_{r,\{p_1,p_2\}}$ is an active edge segment.

We argue that unless the tour traverses exactly $t(t-1)$ switch gadgets from $\gamma$ to $\delta$, there is no tour in $G$ without an edge of cost $2$.
Consider a pair $\{p_{\ell},p_{\ell+1}\}$ for $1< \ell < t(t-1)/2$ such that for some edge $r$, the tour traverses $S_{r,\{p_{\ell-1},p_{\ell}\}}$.
Then it has to continue via a vertex gadget~$V_{i',\{p_{\ell},p_{\ell+1}\}}$ and an edge gadget $S_{r',\{p_{\ell},p_{\ell+1}\}}$ for some vertex $i'$ and some edge $r'$.
Therefore, for each pair, there are exactly two gadgets to traverse and no pair is used more frequently.
After a gadget $S_{r'',\{p_{t(t-1)/2-1},p_{t(t-1)/2}\}}$ is traversed, the only possibility to continue is to go to~$v_{\textnormal{first}}$ and thus to finish the tour.

We now analyze the number of active vertex segments.
Each traversed vertex gadget has to be in an active vertex segment and there are exactly $t(t-1)/2$ traversed vertex gadgets.
Since each vertex segment contains $(t-1)/2$ vertex gadgets, there have to be at least $t$ active vertex segments.
At the same time, there are at most $t$ active vertex segments as otherwise there is an active vertex segment with a vertex gadget that was not visited.
The $t$ vertex segments correspond to $t$ vertices in $H$ and the $t(t-1)/2$ active edge gadgets correspond to $t(t-1)/2$ distinct edges between the $t$ vertices, which is only possible if there is a clique of size $t$ in $H$.

The value $k = 7t(t-1) + 2(t + 1)$ arises as follows.
In each of the $t(t-1)$ switch gadgets, we have to change $4$ edges in order to traverse from $\gamma$ to $\delta$.
There are $t(t-1)/2 + t$ segments that have to be activated, where we have to introduce one edge in each of them.
Furthermore, for each entered~$\gamma$ vertex we have to remove one edge and introduce another edge, $2t(t-1)$ changed edges in total.
We also have to add one edge to connect from $\delta$ of $S_{r'',\{p_{t(t-1)/2-1},p_{t(t-1)/2}\}}$ to $v_{\textnormal{first}}$ and we have to remove the edge of cost $2$.
For each active segment, we also have to remove one edge from the last~$\beta$ vertex as it is not incident to another~$\alpha$-vertex.

\subsection{\texorpdfstring{{\sc $t$-Clique} $\Rightarrow$ {\sc $k$-Edge Change \STSP}}{t-Clique -> k-Edge Change \STSP}}
To prove the other direction, we have to show that if there is a clique $K$ of size~$t$ in $H$, then there is Hamiltonian cycle $C'$ in $G$ of cost $\orderG$ at distance at most $k$ from~$C$.
Let $v_1,\hdots, v_t$ be the vertices in $K$ and for indices $j_1,j_2$ let $e_{j_1,j_2}$ denote the edge between $v_{j_1}$ and $v_{j_2}$.

Cycle $C'$ starts from $v_{\textnormal{first}}$ and goes through the segments, similarly to $C$.
However, for every $j\in\{1,\hdots,t\}$, cycle $C'$ traverses vertex segment $V_{v_j,j}$ in a way different from $C$: after the entrance $a_{v_j,j}$, the cycle goes to exit $b_{v_j,j}$ on the bypass edge, completely avoiding the gadgets in the segment.
For each $j_1 \neq j_2$, in segment $E_{e_{j_1,j_2},\{j_1,j_2\}}$ the cycle $C'$ uses the bypass edge from the entrance to the exit and avoids the gadget $S_{e_{j_1,j_2},\{j_1,j_2\}}$.

After the cycle reaches $v_{\textnormal{last}}$, it has to return to $v_{\textnormal{first}}$ by visiting the skipped gadgets, which we can do as follows.
First, we go from $v_{\textnormal{last}}$ to vertex $\gamma$ of the gadget $V_{p_1,\{p_1,p_2\}}$.
Now assume that we are at vertex $\gamma$ of $V_{v_{p_{\ell}},\{p_{\ell},p_{\ell+1}\}}$ for some $\ell\in\{1,\hdots,t(t-1)/2-2\}$.
The cycle uses the lower gates to visit $V_{v_{p_{\ell}},\{p_{\ell},p_{\ell+1}\}}$, leaves the gadget at vertex $\delta$, and goes to vertex $\gamma$ of edge gadget~$E_{e_{p_{\ell},p_{\ell+1}},\{p_{\ell},p_{\ell+1}\}}$.
After going through this edge gadget on the lower gates, the cycle goes to vertex $\gamma$ of $V_{v_{p_{\ell+1}},(p_{\ell+1},p_{\ell+2})}$.
By the definition of $R$, $\{p_\ell,p_{\ell+1}\} \in R_{p_{\ell}}$ and $\{p_{\ell+1},p_{\ell+2}\} \in R_{p_{\ell+1}}$, which implies that both vertex segments are active.
We continue this way until vertex $\delta$ of gadget $E_{e_{p_{\ell}},\{p_{\ell},p_{\ell+1}\}}$ is reached for $\ell = t (t - 1)/2 -1$.
At that point the cycle $C'$ is terminated by an edge going to~$v_{\textnormal{first}}$.
It is clear that every skipped gadget is visited exactly once.
By the discussion in the first direction of the proof, the distance of~$C'$ from~$C$ is exactly~$k$.
Furthermore, cycle $C'$ does not use the edge from $v_{\textnormal{last}}$ to $v_{\textnormal{first}}$, hence its total cost is $\orderG$, which is strictly smaller than the weight of $C$.
This completes the proof of Theorem~\ref{thm:kchange12tspisw1hard}.
\hfill$\qed$

\medskip

Theorem~\ref{thm:kchange12tspisw1hard} implies that {\sc $k$-Edge Change \STSP} cannot be solved in time $f(k) \cdot n^{O(1)}$ for any computable function $f$, unless $\mathsf{W}[1]=\mathsf{FPT}$.
If we make a stronger complexity-theoretic assumption, then we can actually prove a lower bound on the exponent of $n$:
The Exponential Time Hypothesis~\cite{ImpagliazzoEtAl2001} says that $n$-variable 3-SAT cannot be solved in time $2^{o(n)}$.
\begin{theorem}
  There is no $f(k) \cdot n^{o(\sqrt[3]{k})}$ time algorithm for {\sc $k$-Edge Change \STSP} with~$n$ cities, unless the Exponential Time Hypothesis fails.
\end{theorem}
\begin{proof}
  The proof in Theorem~\ref{thm:kchange12tspisw1hard} takes an instance $(H, t)$
  of {\sc $t$-Clique}, and turns it into an equivalent instance of {\sc$k$-Edge Change \STSP} with $k=O(t^3)$.
  Therefore, an $f(k)\cdot n^{o(\sqrt[3]{k})}$-time algorithm for {\sc $k$-Edge Change \STSP} would be able to solve any instance of {\sc $t$-Clique} in time $f'(t)\cdot n^{o(t)}$.
  As shown by Chen et al.~\cite{ChenEtAl2004}, this would imply that $n$-variable 3-SAT can be solved in time~$2^{o(n)}$.
\end{proof}